\documentclass[11pt]{article}
\usepackage{jheppub}
\usepackage[utf8]{inputenc}
\usepackage[usenames,dvipsnames]{xcolor}
\usepackage{tikz-cd}
\usepackage{cancel}
\usepackage{amsmath,amsthm,amsfonts,mathtools}
\usepackage{yhmath}
\usepackage[mathscr]{eucal}
\usepackage[sans]{dsfont}
\usepackage{tensor}
\usepackage{bigints}
\usepackage{calc}
\usepackage{makecell}
\usepackage{bbm}
\usepackage{bm}
\usepackage{verbatim}
\usepackage{hyperref}
\hypersetup{colorlinks = true, urlcolor=MidnightBlue, citecolor=red, linkcolor=ao!90!black}
\definecolor{ao}{rgb}{0.0, 0.5, 0.0}

\newcommand{\phsp}{\mathcal{P}}
\newcommand{\ham}{H}
\newcommand{\opalg}{\mathcal{L}(\mathcal{H})}

\newcommand{\be}{\begin{equation}}
\newcommand{\ee}{\end{equation}}
\newcommand{\bea}{\begin{eqnarray}}
\newcommand{\eea}{\end{eqnarray}}

\newcommand{\R}{\mathbb{R}}
\newcommand{\C}{\mathbb{C}}

\newcommand{\rd}{\mathsf{d}}
\newcommand{\Lie}{\mathcal L}
\newcommand{\lie}{\Lie}

\newcommand{\Sl}{\mathfrak{sl}}
\newcommand{\su}{\mathfrak{su}}
\newcommand{\slr}{\Sl(2,\R)}

\newcommand{\SL}{\mathrm{SL}}
\newcommand{\SO}{\mathrm{SO}}

\newcommand{\ap}{B}     %
\newcommand{\ac}{E}     %

\newcommand{\algX}{X}  %
\newcommand{\algZ}{Z}
\newcommand{\ii}{\mathtt{i} }
\newcommand{\id}{ {\bullet} }
\newcommand{\jfxn}{\mathbbm{j}}  %
\newcommand{\ja}{\jfxn}
\newcommand{\defj}{\wh{\jfxn}}
\DeclareMathOperator{\Mat}{Mat}
\newcommand{\cstd}{\wt{c}}
\newcommand{\nfxn}{\mathbbm{n}}

\newcommand{\defn}{\wh{\nfxn}}
\newcommand{\defe}{\wh{\mathbbm{e}}}
\newcommand{\phfxn}{\mathbbm{h}}
\newcommand{\defph}{\wh{\phfxn}}
\newcommand{\opH}{\wh{H}}  
\newcommand{\upb}{\epsilon}

\newcommand{\la}[1]{#1^{(1)} }  %

\newcommand{\si}{\sigma}
\newcommand{\pa}{\partial}

\DeclareMathOperator{\SU}{SU}

\DeclareMathOperator{\tr}{Tr}
\DeclareMathOperator{\ad}{ad}
\DeclareMathOperator{\Id}{Id}

\newcommand{\mfk}[1]{\mathfrak{#1}}

\newcommand{\mbb}[1]{\mathbb{#1}}

\newcommand{\mbf}[1]{\mathbf{#1}}
\newcommand{\msf}[1]{\mathsf{#1}}
\newcommand{\mcal}[1]{\mathcal{#1}}

\newcommand{\tenofo}[1]{\text{\normalfont #1}}

\renewcommand\div{\mathrm{div}}
\newcommand{\op}{\mathcal{O}}
\newcommand{\wt}[1]{\widetilde{#1}}
\newcommand{\wh}[1]{\widehat{#1}}

\newcommand{\mrln}{\quad \mathrel{\substack{\vspace{-.15cm}\textstyle\xrightarrow{\tenofo{Matrix Regularization}}
                      \\
                      \vspace{-.04cm}\textstyle\xleftarrow[\tenofo{Large-$N$ Limit}]{\hphantom{\tenofo{Matrix Regularization}}}}}\quad}

\newcommand{\gslr}{\mathfrak{g}_{\mfk{sl}(2,\mbb{R})}(S)}

\newcommand{\cslr}{\mathfrak{c}_{\mfk{sl}(2,\mbb{R})}(S)}
\newcommand{\crr}{\mathfrak{c}_{\mbb{R}}(S)}
\newcommand{\crn}{{\mathfrak{c}}_{\mbb{R}}(N)}
\newcommand{\cslrn}{\mathfrak{c}_{\mfk{sl}(2,\mbb{R})}(N)}
\newcommand{\Gslr}{G_{\text{SL}(2,\mbb{R})}(S)}
\newcommand\Diff{\text{Diff}(S)}

\newcommand{\diff}{\mathfrak{diff}(S)}
\newcommand{\sdiff}{\mathfrak{sdiff}(S)}

\newcommand{\ep}{\epsilon}
\newcommand{\ind}[1]{\indices{#1}}
\newcommand{\tj}[1]{\begin{pmatrix} #1 \end{pmatrix}} 
\newcommand{\sj}[1]{\begin{Bmatrix} #1 \end{Bmatrix}}
\newcommand{\tjn}[1]{\begin{bmatrix} #1 \end{bmatrix}}
\newcommand{\beq}{\begin{equation}}
\newcommand{\eeq}{\end{equation}}
\newcommand{\del}{\nabla}

\newcommand{\lmconj}[1]{\bar{#1}}

\newcommand{\sys}[2]{Y^{#1}_{#2}}

\newcommand{\mm}{\text{-}} 

\newcommand{\bth}{\bar\eth}

\newtheorem{lem}{Lemma}[section]

\newenvironment{eqaligned}
{%
\begin{equation}
    \begin{aligned}
    } 
{%
\end{aligned}
\end{equation}
\vspace{-\lineskip}
\\
\noindent
}

\title{\centering Matrix Quantization of Gravitational Edge Modes}
\author[a]{William Donnelly,} 
\author[a]{Laurent Freidel,}
\author[b]{Seyed Faroogh Moosavian,}
\author[a,c]{and Antony J. Speranza}
\affiliation[a]{Perimeter Institute for Theoretical Physics, 31 Caroline St. N.,  Waterloo ON, N2L 2Y5, Canada}
\affiliation[b]{Department of Physics, McGill University, Ernest Rutherford Physics Building, 3600 Rue University, Montr\'eal, QC H3A 2T8}
\affiliation[c]{Department of Physics, University of Illinois, Urbana-Champaign, Urbana IL 61801, USA}

\emailAdd{williamdonnelly@gmail.com\,\!}
\emailAdd{\quad lfreidel@pitp.ca\,\!}
\emailAdd{sfmoosavian@gmail.com\,\!}
\emailAdd{\quad asperanz@gmail.com}

\abstract{
Gravitational subsystems with boundaries carry the action of an infinite-dimensional symmetry algebra, with potentially profound implications for the quantum theory of gravity.
We initiate an investigation into the quantization of this corner symmetry algebra for the phase space of gravity localized to a region bounded by a 2-dimensional sphere.  
Starting with the observation that the algebra $\mathfrak{sdiff}(S^2)$ of area-preserving diffeomorphisms
of the 2-sphere admits a deformation to the finite-dimensional algebra $\mathfrak{su}(N)$, we derive novel finite-$N$ deformations for two important subalgebras of the gravitational corner symmetry algebra.  
Specifically, we find that the area-preserving hydrodynamical algebra $\mathfrak{sdiff}(S^2)\oplus_{\mathcal{L}}
\mathbb{R}^{S^2}$ arises as the large-$N$ limit of 
$\mathfrak{sl}(N,\mathbb C)\oplus\mathbb{R}$ and  that 
the full area-preserving corner symmetry algebra 
$\mathfrak{sdiff}(S^2)\oplus_{\mathcal{L}}\mathfrak{sl}(2,\mathbb{R})^{S^2}$ is 
the large-$N$ limit of the pseudo-unitary group $\mathfrak{su}(N,N)$.  
We find matching conditions for the Casimir elements of the deformed and continuum algebras and show how these determine the value of the deformation parameter $N$ as well as the representation of the deformed algebra associated with a quantization of the local gravitational phase space.  
Additionally, we present a number of novel results related to the various algebras appearing, including a detailed analysis 
of the asymptotic expansion of the $\mathfrak{su}(N)$ structure constants,
as well as an explicit computation of the full $\mathfrak{diff}(S^2)$ 
structure constants in the spherical harmonic basis.
A consequence of our work is the definition of an area operator which is compatible with the deformation of the area-preserving corner symmetry at finite $N$.}

\begin{document}

\maketitle

\section{Introduction}\label{section:introduction}

Symmetry has long been a guiding principle in developing and understanding 
physical theories. This is especially true in quantum gravity, where an absence of experimental constraints forces us to rely on general physical principles such as symmetry to elucidate the conceptual and technical maze.

The defining symmetry of general relativity is the group of diffeomorphisms of spacetime.
Being gauge symmetries, these have long been considered to be devoid of physical content and a mere redundancy of description.
However, when considering a spacetime with a boundary, the situation changes drastically.
This boundary could be a boundary at infinity, with suitable falloff conditions on the fields, or, motivated by considerations of entangling surfaces, could be located at a finite distance.
Boundaries force us to consider degrees of freedom---\emph{edge modes}---localized at the boundary which otherwise would be pure gauge.
These dynamical variables transform under a symmetry group
which, for general relativity in the metric formulation, is given by \cite{Donnelly:2016auv}
\begin{equation}\label{eq:the corner symmetry group}
    G_{\tenofo{SL}(2,\mbb{R})}(S)=\tenofo{Diff}(S)\ltimes\tenofo{SL}(2,\mbb{R})^S,
\end{equation}
where the \emph{corner} $S$ is the boundary of 
a spatial or null Cauchy surface for
the region under consideration, hence codimension-2 in spacetime.\footnote{It has also been shown that this symmetry group of corner-preserving transformations, which we are interested in here, is universal for all diffeomorphism-invariant theories \cite{Speranza:2017gxd}.
This symmetry group can be extended to include
surface deformations, arising from diffeomorphisms that move the 
corner $S$ itself.  We will not consider this extended group 
here, but it has been examined in several recent works \cite{Speranza:2017gxd, CiambelliLeigh202104,  Ciambelli:2021nmv, Freidel:2021dxw, Speranza:2022lxr}. Recent studies of the corner symmetry group, its extension and its link with the asymptotic symmetry group also include
\cite{Freidel:2020xyx, Freidel:2020svx,
Chandrasekaran2020,  Freidel:2021cbc, Chandrasekaran:2021vyu,Ciambelli:2022cfr}.}
$\tenofo{SL}(2,\mbb{R})^S$ denotes the group of $\SL(2,\mbb{R})$-valued functions on $S$, $\tenofo{Diff}(S)$ the group of diffeomorphisms of $S$, and $\ltimes$ indicates a semidirect product structure in which 
the diffeomorphisms act on functions in the usual way via pullbacks. In what follows we will primarily be interested in the Lie algebra of $\Gslr$:
\begin{equation} \label{gslr}
\mfk{g}_{\mfk{sl}(2,\mbb{R})}(S)=\mfk{diff}(S)\oplus_{\mcal L}\mfk{sl}(2,\mbb{R})^S,
\end{equation}
where the subscript $\mcal{L}$ means the action by Lie derivative. 
In this work, we make two important restrictions: we consider $3+1$-dimensional spacetimes, so that $S$ is two-dimensional, and restrict $S$ to have the topology of a sphere.

Equipped with a physical system (a region of space) and a physical symmetry group \eqref{eq:the corner symmetry group}, we can then follow the spirit of Wigner's approach to quantum mechanics \cite{Wigner1931}: studying unitary representations of this symmetry group and identifying the carrying space of relevant unitary representations as the Hilbert space of our theory. 
There are broadly two ways to proceed, which we dub \emph{representation} and \emph{quantization}.\footnote{There is a third possibility, which we might call \emph{polymerization} where the measure on the sphere is taken to be a discrete measure rather than continuous. This possibility is explored in \cite{Freidel:2020ayo} and is close in spirit to the loop gravity approach of quantum gravity \cite{thiemann2008modern}. This results in a discrete representation of the Lie group, which is discontinuous i.e.\ one in which the Lie algebra generators are not differentiable along the sphere. Since we are primarily interested in continuous representations of the Lie algebra, we will not follow this option here.}
\begin{enumerate}
    \item \textbf{Representation.} 
    In this approach, we construct unitary representations of $\Gslr$ using existing methods, such as the method of induced representations \cite{Mackey195201,Mackey195309,Mackey:1978za} or the method of coadjoint orbits \cite{Kirillov196202,Kirillov1976,Kirillov199908,Rawnsley197501,Baguis199705}. 
    For example, in the case of three-dimensional BMS group, this program of quantization has been carried out by Barnich and Oblak using both induced representation \cite{BarnichOblak201403} and coadjoint representation \cite{BarnichOblak201502}, and their equivalence has been argued \cite[Section 4.1]{BarnichOblak201502}.  The first step in studying the representation theory of $\Gslr$ using coadjoint orbits is obtaining a classification of the orbits, a task that has been completed in our previous work \cite{DonnellyFreidelMoosavianSperanza202012}. One then needs to construct the Hilbert space by methods such as geometric quantization \cite{Kirillov196202,Kostant1965,Kostant1970,Souriau1970,Souriau1997} or  brane quantization
\cite{GukovWitten200809,GaiottoWitten202107}. It is thus in principle possible to follow this path. However, complications related to the infinite-dimensionality and topological subtleties of the symmetry group \eqref{eq:the corner symmetry group}, which are already present for its $\Diff$ subgroup, make the implementation of this approach difficult.
    
    \item \textbf{Quantization.} 
    Quantum mechanics allows for a more general class of possibilities than the preceding.
    Namely, it is sufficient to find a group that approaches $\Gslr$ in a suitable classical limit.
    A canonical example of this procedure is provided by the quantum mechanics of a single particle in one dimension.
    The classical phase space of this system, with coordinates $x,p$ (with $\{x,p\} = 1$) carries a representation of the algebra $\mfk{sdiff}(\R^2)$ of area-preserving diffeomorphisms of the phase space plane.
    Quantum-mechanically this symmetry algebra is deformed to an algebra of infinite-dimensional matrices, with the classical symmetry recovered only in the limit $\hbar \to 0$.
    In this case, quantization preserves the Heisenberg subalgebra generated by $x$, $p$, and the constant function $1$, but commutators of more general functions of $x$ and $p$ acquire $\mcal{O}(\hbar)$ corrections due to operator ordering ambiguities.
\end{enumerate} 

Here we will follow the second approach, a choice that requires some physical justification.
The method based on finding an exact representation of the symmetry algebra carries several drawbacks.
The first is that our algebra 
is infinite-dimensional, and although the representation theory
of certain infinite-dimensional algebras such as Virasoro or Kac-Moody
is well-developed, very few tools exist for characterizing the unitary 
representations of the symmetry groups encountered in the present work.
Another drawback of such representations is that the operator product of generators at equal points is ill-defined. 
While this is common in continuum field theory, it is not welcomed for a theory of quantum gravity, as it assumes the continuum structure of spacetime persists even to distances shorter than the Planck scale.
Instead one would expect a fundamental theory to  help resolve such divergences.
Composite operators play a central role in the classification of coadjoint orbits of our symmetry group \cite{DonnellyFreidelMoosavianSperanza202012} and are therefore essential to its quantization.

The introduction of a deformation of the algebra at the Planck scale is strongly suggested by the finite entropy of black holes and other causal horizons.
To see how the symmetry group \eqref{eq:the corner symmetry group} relates to entropy, we first consider the simpler example of Yang-Mills theory.
In Yang-Mills theory with gauge group $g$, the symmetry group analogous to \eqref{eq:the corner symmetry group} is a direct product of the gauge group over points of the surface $S$ \cite{Donnelly:2016auv}:
\begin{equation}
G_{\text{YM}} = g^S:= \prod_{x \in S} g(x),
\end{equation}
where each $g(x)$ is an independent copy of the Lie group $g$ at each point $x$ of the surface $S$.
We will use the shorthand notation $g^S$ for this product here and throughout the paper.
This infinite-dimensional group can be made precise by including a lattice regulator \cite{Buividovich:2008gq,Donnelly:2011hn,Casini:2013rba,Donnelly:2014gva,Soni:2015yga,Lin:2018bud}.
The states of a region bounded by the surface $S$ transform nontrivially under $G$, and when decomposed into irreducible representations we obtain a ``$\log \dim (R)$'' term.
\begin{equation} \label{SYM}
S_{\text{YM}} = - \sum_R p_R \log p_R + \sum_R p_R \log(\rho_R) + \sum_R p_R \log \dim(R),
\end{equation}
where $p_R$ is a probability distribution over all irreducible representations of $G_{\text{YM}}$, $\rho_R$ a set of density matrices and $\dim(R)$ the dimension of each irreducible representation.
All three terms in \eqref{SYM} diverge when the number of points in $S$ is taken to infinity; this is the familiar ultraviolet divergence of the entanglement entropy and comes from the infinite density of ultraviolet degrees of freedom.
These ultraviolet divergences are expected in a continuum field theory, and the contribution of the edge modes plays an important role in the relation between the entanglement entropy and the conformal anomaly \cite{Donnelly:2014gva}.

In quantum gravity, we don't expect the behavior of continuum quantum field theory to persist into the infinite ultraviolet.
At the perturbative level, quantum gravity is much like a gauge theory and there has been much progress in calculating one-loop perturbative quantum-gravitational corrections to the entanglement entropy of gravitons \cite{Benedetti:2019uej,Anninos:2020hfj,David:2022jfd}.
It is however expected that the non-perturbative result for the entropy should be finite and universal \cite{Jacobson:2003wv,Bianchi:2012ev}.
This is in some tension with an infinite-dimensional symmetry group of the form \eqref{eq:the corner symmetry group}, which has a continuum of generators, leading to an ultraviolet divergence similar to that of quantum field theory.
The continuum limit is already present kinematically in the structure of the commutators of the algebra \eqref{gslr} which contain delta functions localized at coincident points of $S$.
A further complication comes from the noncompact $\Sl(2,\R)$ factor in \eqref{gslr} --- unitary representations of noncompact groups are infinite-dimensional and this would lead to infinities in the na\"ive application of \eqref{SYM}.

The quantization method we pursue in this work is motivated by an analogy with single-particle quantum mechanics.
It was shown in \cite{DonnellyFreidelMoosavianSperanza202012} that the coadjoint orbits of $\mfk{g}_{\mfk{sl}(2,\mbb{R})}(S)$ can be reduced to those of the Wigner little group $\sdiff$; the latter is the algebra of Hamiltonian transformations of a two-dimensional phase space,
which can be quantized in much the same way as the familiar case of the standard $(x,p)$ phase space.
This quantization leads to a deformation of the symmetry algebra in
which functions on  $S^2$ are  replaced with 
their fuzzy sphere analogs, which are noncommuting hermitian matrices.  
The resulting deformed algebra is finite-dimensional and isomorphic to 
$\mfk{su}(N)$, where $1/N$ plays the role of a deformation parameter,
which has a physical interpretation as a fundamental unit of area.  
This parameter is analogous to the introduction of Planck's constant $h$, which effectively discretizes the phase space into a fuzzy space with cells of area $h$.
We note that Planck was motivated by understanding Boltzmann's entropy formula,\footnote{For a historical account, see \cite{kuhn}.} and it would be remarkable if the same mechanism responsible for the finite Boltzmann entropy could be responsible for a finite Bekenstein-Hawking entropy.

The deformation of $\sdiff$ we consider is well known: it arises, 
for example, in string theory in the context of matrix models \cite{Hoppe198201,Hoppe198901,PopeStelle198905}. 
There is also some similarity to the holographic spacetime model of Ref.~\cite{Banks:2018ypk}, in which cosmological horizons are replaced with fuzzy spaces.
The key distinction here is that rather than taking a ``bottom-up'' approach and introducing the area-preserving diffeomorphism symmetry by hand, we have derived it ``top-down'' from the symmetries of general relativity. 
This allows us to relate the quantum-mechanical generators to geometric quantities in general relativity.
Moreover, we will show how to incorporate the boost symmetry of the normal plane, which has not appeared previously in the 
aforementioned models.
This boost symmetry plays an especially important role in the context of horizon thermodynamics, where the global boost generator
for Killing horizons is the modular Hamiltonian. 

The quantization procedure we consider is accompanied by a deformation of the classical symmetry algebra.
These deformations are ubiquitous in physics: two classical examples are the deformation of the  Galilean algebra into the semi-simple Lorentz algebra and the deformation of the Poincar\'e algebra to the semi-simple de Sitter algebra.
In both cases, a deformation parameter is needed and the deformation is more stable than the original algebra. In the examples just mentioned these deformation parameters are constants of nature such as the cosmological constant or the inverse speed of light.
In the case of deforming $\sdiff$ to $\su(N)$, the small parameter is $1/N$.
In light of the analogy with quantum mechanics, it is natural to guess that $N$ 
is related to the area of the surface $S$
in Planck units, which is supported by 
the matching conditions for Casimir operators 
described in section
\ref{sec:casmatch}.

The results of section
\ref{sec:casmatch} indicate that once the deformation
has been identified, the quantization procedure
is largely constrained by the matching conditions 
on the Casimirs.  However, these conditions fail 
to fully fix the quantization for two reasons.  First,
although the matching determines the representation
for a folium of the gravitational phase space defined
by fixing the values of all Casimirs,  the full phase space is 
in general a sum of several such folia.  The quantization
is then expected to be a sum of different representations, and
it is a nontrivial problem to determine the multiplicity
of the representations appearing in the quantization.  
The second reason that the resulting quantization 
is not fully determined is that it may not be the 
case that the deformed algebras explored in this work are unique.
In particular, there may be other deformations, or the 
quantization may proceed by representing the original,
undeformed algebras.  The perspective taken in this paper
is that the deformations we identify provide
nontrivial, finite-dimensional deformed algebras that 
can be viewed as regulated versions of the continuum
algebras.  The question of uniqueness of these deformations
remains open.

\subsection{Main results of the paper}
We present here a technical summary of our main results.
To make our way toward the quantization of the full group $\Gslr$, we start 
by analyzing some of its important subgroups. 
In the problem of classification of coadjoint orbits of $\Gslr$, one important subgroup is the group of area-preserving diffeomorphisms of $S$ \cite{DonnellyFreidelMoosavianSperanza202012}. 
Choosing an area form $\nu$ on $S$, this is the subgroup of $\tenofo{Diff}(S)$ that preserves $\nu$
\begin{equation}
    \tenofo{SDiff}(S)_{\nu}\equiv\{f\in\tenofo{Diff}(S) \,|\, f^*\nu=\nu\},
\end{equation}
whose Lie algebra is denoted as $\sdiff_{\nu}$. 
All area forms on a sphere are diffeomorphic up to an overall scaling, hence we can always restrict attention
to the natural volume form on the unit round sphere, and denote the corresponding area-preserving subalgebra simply as $\sdiff$.
It is a celebrated result that $\sdiff$ can be viewed as a large-$N$ limit of $\su(N)$\footnote{
Note that the precise meaning of the large $N$ limit of $\su(N)$
is ambiguous, and different limiting procedures can result in 
non-isomorphic infinite-dimensional Lie algebras
(see e.g.\ \cite{PopeStelle198905, deWit:1989yb, Bordemann:1990pa}). 
The way the limit to $\sdiff$ should be understood is in 
terms of quasi-limits, as defined in \cite{Bordemann:1990pa}.}
\begin{equation} \label{sdiffsun}
    \sdiff \mrln \mfk{su}(N).
\end{equation}
The large-$N$ limit of $\mfk{su}(N)$ has been part of a vast investigation in the past starting with the pioneering work of 't Hooft on QCD \cite{tHooft197404}.

\begin{figure}\hspace*{-1.2cm}\centering
\begin{tikzcd}[row sep=1.5cm,line width=2pt]
\diff \oplus_{\mcal{L}}\, \slr^S \arrow[d,shift left] \arrow[r] & \sdiff \oplus_{\mcal{L}} \slr^S \arrow[r] \arrow[d,shift left] & \sdiff \oplus_{\mcal{L}} \R^S \arrow[r] \arrow[d,shift left] &  \sdiff \arrow[d,shift left, "\substack{\text{Matrix} \\ \text{Regularization}}"] \\
??? \arrow[u,shift left, "\substack{\text{Large-$N$} \\ \text{Limit}}"] & \arrow[l] \su(N,N) \arrow[u,shift left]  & \arrow[l] \Sl(N,\C)\oplus \mathbb{R} \arrow[u,shift left] & \arrow[l] \su(N) \arrow[u,shift left]
\end{tikzcd}
\caption{The corner symmetry algebra $\gslr$, its subalgebras and their regularizations. The regularized algebra whose large-$N$ limit is $\gslr$ is missing in our analysis.}\label{fig:subalgs}
\end{figure}

Our goal is to generalize the 
procedure \eqref{sdiffsun} by which the 
large-$N$ limit of $\su(N)$ 
approaches $\sdiff$
to obtain a new sequence of deformed algebras that limit
to the full  corner symmetry algebra.
This generalization proceeds via the sequence of subalgebras depicted in figure \ref{fig:subalgs}.
From the full algebra $\diff \oplus_{\mcal{L}}\, \slr^S$ in the upper-left corner, we first fix an area form and consider the subalgebra $\sdiff\oplus_{\mcal{L}}\slr^S$ 
which preserves this area form.
Further fixing a hyperbolic generator of $\slr$ at each point on $S$ reduces the algebra to $\sdiff \oplus_{\mcal{L}} \R^S$ consisting of area-preserving diffeomorphisms and pointwise boosts.
Finally, fixing the boost generator to zero we are left with the little group $\sdiff$.
Then starting from the known regularization of $\sdiff$ by $\su(N)$ we proceed leftward along the bottom row of the diagram, finding an increasing sequence of finite-dimensional Lie algebras compatible with the $\su(N)$ regularization of $\sdiff$.

The regularization procedure is straightforward: we write the mode expansion of the generators on the sphere and look for matrices whose commutators agree with the Poisson brackets up to small corrections.
The first subalgebra we consider is:
\begin{equation}\label{eqn:cRalg}
    \mfk{c}_{\mbb{R}}(S)=\mfk{sdiff}(S)\oplus_{\mcal{L}}\mbb{R}^S. 
\end{equation}
Following \cite{DonnellyFreidelMoosavianSperanza202012}, we denote the smeared phase space generators of $\sdiff$ and $\mbb{R}^S$ as $J[\phi]$ and $N[\lambda]$, respectively. 
The smearing parameters $(\phi,\lambda)$ are both real-valued functions on the sphere.  These functions can be expanded in the basis of spherical harmonics
$\{Y_\alpha\}$
where $\alpha$ stands for a pair of indices $(A,a)$ with $A\in\{0,1,\ldots\}$
denoting the total angular momentum quantum number, 
and $a\in\{-A,\ldots,A\}$ is the magnetic quantum number.
The generators in this basis are denoted $J_\alpha \equiv J[Y_\alpha]$,
$N_\alpha = N[Y_\alpha]$, and the Poisson brackets of the generators 
$(J_\alpha, N_\alpha)$ implementing the Hamiltonian action of (\ref{eqn:cRalg})
on the gravitational phase space are given by
\begin{equation} \label{sdiffRpoisson}
   \qquad 
   \begin{aligned}
    \{J_\alpha,J_\beta\}&= C\indices{_{\alpha\beta}^\gamma}J_\gamma,
    \\
    \{J_\alpha,N_{\beta}\}&= C\indices{_{\alpha\beta}^\gamma}N_{\gamma},
    \\
    \{N_{\alpha},N_{\beta}\}&=0.
    \end{aligned}
\end{equation}
The $C\indices{_{\alpha\beta}^\gamma}$ are structure constants of the Poisson bracket on $S^2$,
\begin{equation}
\upb^{AB} (\nabla_A Y_\alpha \nabla_B Y_\beta) = C\indices{_{\alpha\beta}^\gamma} Y_\gamma,
\end{equation}
where $\upb^{AB}$ is the inverse of the standard volume form 
on the unit radius sphere.  
 An explicit expression for $C\indices{_\alpha_\beta^\gamma}$
can be given in terms of Wigner 3j-symbols  (see Appendix \ref{sec:sh} for  details)  \cite{Dowker:1990iy,Dowker:1990ss}.

As discussed in section \ref{sec:sdiffslnc}, the matrix regularization
of \eqref{sdiffRpoisson} is achieved by replacing the generators 
with matrices of dimension $2N \times 2N$.  The generators of the deformed
$\sdiff$ subalgebra correspond to matrices of the form $\wh{Y}_{\id\alpha} = 
\mathbbm{1}_{2}\otimes \wh{Y}_\alpha$, obtained by simply tensoring 
the fuzzy spherical harmonics with the $2\times 2$ identity matrix.  
The remaining generators are of the form $\wh{Y}_{1\alpha} = \rho_1\otimes 
\wh{Y}_\alpha$, where $\rho_1 = \frac12\begin{pmatrix}0&1\\-1&0\end{pmatrix}$.
Together, the commutators of $(\wh{Y}_{\id\alpha}, \wh{Y}_{1\alpha})$
are taken to define the deformed algebra in a $2N$-dimensional representation.
Denoting the corresponding generators of the deformed Lie 
algebra $(\algX_\alpha, \algZ_\alpha)$, appropriately rescaled, the Lie
brackets take the form
\begin{eqaligned}
    [\algX_\alpha, \algX_\beta] 
&= \wh{C}\indices{_\alpha_\beta^\gamma}\algX_\gamma, \\
[\algX_\alpha, \algZ_\beta] 
&= \wh{C}\indices{_\alpha_\beta^\gamma}\algZ_\gamma, \\
[\algZ_\alpha,\algZ_\beta] 
&= -\frac{1}{N^2} \wh{C}\indices{_\alpha_\beta^\gamma}\algX_\gamma,
\end{eqaligned}
where $\widehat C\indices{_{\alpha \beta}^\gamma}$ denote the structure
constants for the fuzzy spherical harmonic commutator, 
$[\widehat Y_\alpha, \widehat Y_\beta] = \frac{2i}{N} \widehat C_{\alpha \beta}{}^\gamma \widehat Y_\gamma$.  Since $\wh{C}\indices{_\alpha_\beta^\gamma}
\rightarrow C\indices{_\alpha_\beta^\gamma}$ as $N\rightarrow\infty$, we 
find that the algebra defined by $(\algX_\alpha,\algZ_\alpha)$
approaches the Poisson bracket algebra 
(\ref{sdiffRpoisson}) in the large $N$ limit.  For finite
$N$, one can show that the algebra is isomorphic to $\mfk{sl}(N,\mathbb{C})
\oplus \mathbb{R}$ (after removing the central generator $\algX_{00}$ which 
does not generate a diffeomorphism in the continuum algebra).  This establishes the following
novel large $N$ limit:
\begin{equation}
    \crr=\mfk{sdiff}(S)\oplus_{\mcal{L}}\mbb{R}^S\mrln\crn\simeq\mfk{sl}(N,\mbb{C})\oplus \mathbb{R}. 
\end{equation}

Going a step further, we then consider an enlargement of the algebra 
by including the full set of pointwise $\slr$ transformations:
\begin{equation}
    \cslr=\mfk{sdiff}(S)\oplus_{\mcal L}\slr^S. 
\end{equation}
The Hamiltonian generators of the Poisson bracket algebra in the spherical harmonic 
basis are now denoted $(J_\alpha, N_{a\alpha})$, with $a = 0,1,2$ an
$\mfk{sl}(2,\mathbb{R})$ index.  
The Poisson brackets are given by
\begin{equation} \label{sdiffslrpoisson}
    \begin{aligned}
    \{J_\alpha,J_\beta\}&= C\indices{_{\alpha\beta}^\gamma}J_\gamma,
    \\
    \{J_\alpha,N_{a \beta}\}&= C\indices{_{\alpha\beta}^\gamma}N_{a \gamma},
    \\
    \{N_{a\alpha},N_{b \beta}\}&=E\indices{_{\alpha\beta}^\gamma}\varepsilon\indices{_{ab}^c}N_{c\gamma}.
    \end{aligned}
\end{equation}
where $\varepsilon_{abc}$ denotes the Levi-Civita symbol, whose index
is raised with the metric $\eta^{cd}  = \text{diag}(-1,+1,+1)$ 
and we have introduced a new set of structure constants $E\indices{_{\alpha\beta}^\gamma}$ associated with the commutative product 
of functions on the sphere, $Y_\alpha Y_\beta = E\indices{_{\alpha\beta}^\gamma} Y_\gamma$.
Like the $C\indices{_{\alpha\beta}^\gamma}$, the $E\indices{_{\alpha\beta}^\gamma}$ can be written explicitly in terms of Wigner 3j symbols (see Appendix \ref{sec:sh} for the details) \cite{Dowker:1990iy,Dowker:1990ss}.

The regularization of $\cslr$ is obtained
in section \ref{sec:sdiffsunn} by a similar procedure 
as the case of $\crr$.  We construct a $2N$-dimensional representation
of the deformed algebra with the matrices $\wh{Y}_{\id\alpha} = \mathbbm{1}_{2}\otimes \wh{Y}_\alpha$, and $\wh{Y}_{a\alpha} = \rho_a\otimes
\wh{Y}_\alpha$, where $\rho_a$ are a basis for $\mfk{sl}(2,\mathbb{R})$,
defined in equation (\ref{matr}).  The corresponding 
Lie algebra generators are denoted $(\algX_\alpha, \algZ_{a\alpha})$,
rescaled appropriately, and their algebra derived from the $2N$-dimensional
representation is given by 
\begin{eqaligned}
    [\algX_\alpha,\algX_\beta] 
&= \wh{C}\indices{_\alpha_\beta^\gamma} \algX_\gamma, \\
[\algX_\alpha, \algZ_{a \beta}]
&=
\wh{C}\indices{_\alpha_\beta^\gamma}\algZ_{a \gamma},  \\
[\algZ_{a\alpha},\algZ_{b \beta}]
&=
\wh{E}\indices{_\alpha_\beta^\gamma}\varepsilon\indices{_a_b^c}\algZ_{c\gamma}
-\frac{1}{N^2} \wh{C}\indices{_\alpha_\beta^\gamma}\eta_{ab}\algX_\gamma,
\end{eqaligned}
where $\wh{E}\indices{_\alpha_\beta^\gamma}$ are the structure 
constants for the Jordan product of the fuzzy spherical 
harmonics, $\wh{Y}_\alpha \circ\wh{Y}_\beta = \frac12(\wh{Y}_\alpha\wh{Y}_\beta
+\wh{Y}_\beta\wh{Y}_\alpha) = \wh{E}\indices{_\alpha_\beta^\gamma}\wh{Y}_\gamma$,
which approach $E\indices{_\alpha_\beta^\gamma}$ in the large $N$ limit.  
It is then readily apparent that the deformed algebra generated by 
$(\algX_\alpha, \algZ_{a\alpha})$ approaches the classical 
algebra \eqref{sdiffslrpoisson} as $N\rightarrow\infty$.
Since the deformed algebra can be shown to be isomorphic
to $\mfk{su}(N,N)$, this establishes the second novel 
large $N$ limit in this work,
\begin{equation}
    \cslr=\mfk{sdiff}(S)\oplus_{\mcal{L}}\slr^S\mrln\cslrn\simeq\mfk{su}(N,N).
\end{equation}

Having established the existence of regularized algebras that approach
the three continuum algebras $\sdiff$, $\crr$, and $\cslr$ at large $N$, we turn in section \ref{sec:casimirs} to the analysis 
of Casimir operators for the deformed and continuum algebras.  
For each large $N$ limit, we demonstrate that the Casimir elements 
for the deformed algebras approach 
corresponding Casimir elements of the continuum algebras.  
We further argue that the matching conditions for the Casimir elements
can be used to determine the representation of the deformed algebra
that appears in the quantization of the gravitational phase space,
and further argue that the matching conditions can also be used to 
determine the value of the deformation parameter $N$.  We outline 
how this procedure can be carried out in detail in the case 
of $\mfk{su}(N)$ in section \ref{sec:casmatch}.  
We leave a detailed calculation
of the matching for the other deformed algebras $\mfk{sl}(N,\mathbb{C})\oplus 
\mathbb{R}$ and $\mfk{su}(N,N)$ for future work.  
Additionally, for the case of $\mfk{su}(N,N)$ we identify an operator
that can be associated with the dynamical area of the surface, and 
argue that while it is a Casimir in the continuum algebra, it becomes 
noncommutative at finite $N$.  

Since the various large $N$ limits considered in this paper
rely on properties of the fuzzy spherical harmonics $\wh{Y}_\alpha$,
we collect a number of formulas and conventions related to them
in appendix \ref{appendix:spherical_harmonics}.  In particular,
the conventions used for the continuum spherical harmonics are 
presented in section \ref{sec:sh}, and conventions for spin-weighted
harmonics, which are used in calculations of structure constants
for various differential operators, are given in section
\ref{sec:swh}.  Following that, we review the presentation
of the fuzzy spherical harmonics developed in \cite{Freidel:2001kb}
in section \ref{sec:fuzH}.  Additionally, we present a novel formula,
derived from an identity due to Nomura \cite{nomura1989description},
for the asymptotic limit of the Wigner 6j-symbol appearing in the 
structure constants for the fuzzy harmonics product, and demonstrate 
that it immediately provides an expansion of this matrix product order
by order in powers of $\frac{1}{N}$.  This asymptotic expansion allows us to  evaluate subleading corrections 
to the matrix product beyond the Poisson bracket term. We develop this expansion in  section \ref{sec:mpexpansion} 
by determining the $\mathcal{O}(\frac{1}{N^2})$ 
contribution to the matrix product, showing that it takes the 
form expected from a valid Moyal product of functions on the sphere. 
In section \ref{app:star} we show that to all orders in $\frac{1}{N}$, the 
Nomura identity yields the expansion of a specific choice of Moyal
product on the sphere.  

The majority of this work has focused on the three subalgebras 
of the full corner symmetry group
appearing in the top line of figure \ref{fig:subalgs}.  Ultimately,
however, we are 
interested in determining the deformation
and quantization of the full symmetry  $\diff \ltimes\Sl(2,\mbb{R})^S$.  While we do not obtain a deformation of this symmetry
algebra due to several conceptual issues related to the form such a 
deformation should take, we initiate the investigation
into such deformations by determining the structure constants
of the  $\diff$ algebra, including diffeomorphisms that 
do not preserve a chosen area form.  These structure constants are 
derived in the spherical harmonic basis in appendix
\ref{app:diffsc}, and they do not appear to have been presented 
previously in the literature.  These expressions will inform 
future work into possible deformations of the full symmetry algebra,
and also will likely be useful in other contexts in which 
$\diff$ algebra appears, such as extended symmetries of asymptotically
flat space and celestial holography \cite{Campiglia:2014yka,
Campiglia:2015yka, Compere:2018ylh, Barnich:2021dta, Freidel:2021fxf, Freidel:2021qpz, Prema:2021sjp,
Raclariu:2021zjz, Pasterski:2021rjz}.
The remaining appendices include calculational details and proofs of formulas appearing in the main text.

\section{Corner symmetries and their Poisson brackets}\label{sec:cornsyms}

In this section, we review some aspects of the corner symmetry algebra $\gslr$ and some of its important subalgebras. To prepare the ground for the matrix regularization of these subalgebras, we introduce an explicit basis of generators and 
give the structure constants of the algebras in this basis.

\subsection{Corner symmetry algebra and its subalgebras}\label{sec:subalgs}

As established in \cite{Donnelly:2016auv}, in the presence of a finite-distance corner $S$, general relativity in the metric formulation enjoys a symmetry group, called the \emph{corner symmetry group}, which acts on the dynamical variables, 
which in the classical analysis correspond to functions on the theory's phase space.
The Lie algebra of the corner symmetry group is
denoted $\gslr$, and consists of two types of transformations: 1)  diffeomorphisms that are tangent to the corner $S$, and 2) generalized boosts that fix $S$ but act on its normal plane. 
Vector fields on $S$ generate the first type of transformations, forming 
a $\mfk{diff}(S)$ Lie algebra under the vector field Lie bracket, while
the second type of transformations are generated by $\mfk{sl}(2,\mbb{R})$-valued functions on $S$, with the Lie bracket computed pointwise. Denoting the coordinates on $S$ by $\si\equiv(\si^1,\si^2)$,  the generators can be packaged together into a pair $(\xi,\lambda)$, where $\xi(\si)=\xi^A(\si)\pa_A$ is a vector field on $S$ and $\lambda(\si)=\lambda^a(\si)\tau_a$ for $a=0,1,2$ belongs to $\mfk{sl}(2,\mbb{R})^S$.
Here $\tau_a$ are $\slr$ generators, whose Lie brackets are given by $[\tau_a, \tau_b] = \varepsilon_{ab}{}^c \tau_c$ where $\varepsilon_{abc}$ is the three-dimensional Levi-Civita symbol $\varepsilon_{012} = 1$ and $\slr$ indices are raised and lowered with the metric $\eta_{ab}=\tenofo{diag}(-1,+1,+1)$. An explicit matrix representation of such generators is
\begin{equation}\label{eqn:taua}
    \tau_0=\frac{1}{2}\begin{pmatrix}
    0 & -1
    \\
    +1 & 0
    \end{pmatrix},\qquad 
    \tau_1=\frac{1}{2}\begin{pmatrix}
    +1 & 0
    \\
    0 & -1
    \end{pmatrix},
    \qquad 
    \tau_2=\frac{1}{2}\begin{pmatrix}
    0 & +1
    \\
    +1 & 0
    \end{pmatrix}.
\end{equation}
The Lie algebra of  $\gslr$ is then given  by
\begin{equation}\label{eq:the Lie bracket of gS}
    \Big[(\xi_1,\lambda_1),(\xi_2,\lambda_2)\Big]=\Big([\xi_1,\xi_2]_{\tenofo{Lie}},\;\mcal{L}_{\xi_1}\lambda_2-\mcal{L}_{\xi_2}\lambda_1+[\lambda_1,\lambda_2]_{\mfk{sl}(2,\mbb{R})} \Big),
\end{equation}
where $[\cdot,\cdot]_{\tenofo{Lie}}$ denotes the Lie bracket of vector fields on $S$, and $[\cdot,\cdot]_{\mfk{sl}(2,\mbb{R})}$ is the $\mfk{sl}(2,\mbb{R})$ Lie bracket. Explicitly, we have
\begin{equation}
    \begin{gathered}
    [\xi_1,\xi_2]_{\tenofo{Lie}}:=(\xi_1^A\pa_A\xi_2^B-\xi_2^A\pa_A\xi_1^B)\partial_B,
    \\
    \mcal{L}_{\xi_1}\lambda_2:=\xi_1^A\pa_A\lambda_2,\qquad \mcal{L}_{\xi_2}\lambda_1:=\xi_2^A\pa_A\lambda_1,
    \\
    [\lambda_1,\lambda_2]_{\mfk{sl}(2,\mbb{R})}:=2\lambda^a_1\lambda^b_2\,\varepsilon\indices{_{ab}^c}\tau_c,
    \end{gathered}
\end{equation}
As is clear from \eqref{eq:the Lie bracket of gS}, $\mfk{diff}(S)$ acts on $\mfk{sl}(2,\mbb{R})^S$ by the Lie derivative and hence the symmetry algebra is
\begin{equation}\label{eq:the algebra gS}
    \mfk{g}_{\mfk{sl}(2,\mbb{R})}(S)=\mfk{diff}(S)\oplus_{\mcal{L}}\mfk{sl}(2,\mbb{R})^S,
\end{equation}
where $\oplus_{\mcal{L}}$ denotes a semidirect sum with an action of the 
first algebra on the second realized by the Lie derivative.

The subalgebras relevant in this work all involve a restriction of the 
$\mfk{diff}(S)$ algebra to an area-preserving subalgebra, which can be explicitly 
constructed as follows.  
Let $\wt{n}$ be a positive density on $S$. Area-preserving diffeomorphisms are generated by divergenceless vector fields $\xi^A$ with respect to 
$\wt{n}$,  which satisfy  $ \pa_A (\wt{n} \xi^A)=0$. Since the Lie bracket of two such vector fields also satisfies this condition, the set of area-preserving diffeomorphisms forms a subalgebra of $\mfk{diff}(S)$, which we denote as $\mfk{sdiff}_\nu(S)$, where $\nu$ is the volume form on the sphere defined as 
\begin{equation}
    \nu:=\frac{1}{2}\nu_{AB}\rd\si^A\wedge\rd \si^B, 
\end{equation}
where $\nu_{AB}:=\wt{n}\,\varepsilon_{AB}$, and $\varepsilon_{AB}$ is the Levi-Civita symbol with $\varepsilon_{12}=1$. 

An alternative presentation of the area-preserving diffeomorphisms can 
be given in terms of functions on the sphere, and will serve more convenient
when comparing to the regularized algebras in later sections.   
Since $\nu^{AB}=\varepsilon^{AB}/\wt{n}$
(again with $\varepsilon^{12} = 1$) defines a Poisson tensor on the sphere, the corresponding
Poisson bracket
\begin{equation}\label{eq:the Poisson bracket of function on sphere with density nu}
    \{\phi_1,\phi_2\}_{\nu}:=\nu^{AB}\pa_A\phi_1\pa_B\phi_2, \qquad \phi_1,\phi_2\in C(S),
\end{equation}
where $C(S)$ denotes the space of functions on sphere, acts as a derivation of the function $\phi_1$ on $\phi_2$. 
Note that this relation implies that 
\beq
d\phi_1\wedge d\phi_2 =\nu\, \{\phi_1,\phi_2\}_\nu.  
\eeq
The vector field $\xi^A$ associated with this derivation is divergenceless, and can
be identified with a function $\phi$,
called the \emph{stream function}, through the 
relation\footnote{This follows from the fact that a divergenceless
vector field satisfies $d(\xi\cdot \nu) = 0$, which 
on the sphere implies that $\xi\cdot \nu = -d\phi$ for some function
$\phi$. 
}
\begin{equation}\label{eqn:stream}
    \xi^B=\nu^{AB}\pa_A\phi. 
\end{equation}
For a given vector field $\xi^A$, this equation determines $\phi$ up to a constant
shift, which can be fixed by requiring that the function $\phi$ integrate to 
zero over the sphere.
The action of the vector field on functions is then reproduced by taking Poisson
brackets with the associated stream function.  
To make this correspondence clear, we denote a vector field  
preserving the area form $\nu$ corresponding to the stream function $\phi$ by $\xi_{\phi}^\nu:= \nu^{AB}\pa_A\phi \pa_B$.
This vector field is such that 
\be \label{eqn:xiphibrack}
[\xi_\phi^\nu,\xi_\psi^\nu]= \xi^\nu_{\{\phi,\psi\}_\nu}, \qquad 
\xi_\phi^\nu[\psi]=\{\phi,\psi\}_\nu,
\ee
demonstrating that the map from an area-preserving vector field to its 
stream function is a Lie algebra homomorphism into the Poisson bracket algebra of 
functions on the sphere.  Note that the relation (\ref{eqn:stream}) implies that 
the constant function on the sphere 
is not the stream function of any nonzero vector field, and this 
function generates the center of the full Poisson algebra.  Hence, the Poisson algebra
can be viewed as a trivial central extension of the algebra $\sdiff$ by this constant 
function.

The area-preserving subalgebra comprises an important component of the main algebra 
studied in this work, which
 is the subalgebra of $\gslr$ that preserves a given volume form
$\nu =\wt{n} d^2\sigma$.  In \cite{DonnellyFreidelMoosavianSperanza202012}, it was 
called the \emph{centralizer subalgebra} $\mfk{c}_{(\mfk{sl}(2,\mbb{R}),\nu)}(S)$, 
since it centralizes the $\Sl(2,\mathbb{R})$
quadratic Casimir operator in the universal enveloping algebra, which defines an area
form on the gravitational phase space.  
All $\Sl(2,\mbb{R})$ transformations preserve the volume form $\nu$, therefore the centralizer subalgebra  is 
simply obtained by restricting the diffeomorphisms appearing in the full algebra 
$\gslr$ to area-preserving ones. 
This fixes  the centralizer subalgebra to be 
\begin{equation}\label{eq:the algebra cS}
    \mfk{c}_{(\mfk{sl}(2,\mbb{R}),\nu)}(S)=\mfk{sdiff}_\nu(S)\oplus_{\mcal{L}}\mfk{sl}(2,\mbb{R})^S.
\end{equation}
Going forward, we will work exclusively with the normalized\footnote{In standard 
spherical coordinates, $\nu_0=\frac1{4\pi} \sin\theta \rd \theta\wedge \rd \phi$ 
and satisfies $\int_S \nu_0=1$.} round sphere volume form, denoted 
$\nu_0$.\footnote{This differs from the ``dynamical'' measure  $\nu=\nu_N$ given by $\nu_N = \sqrt{N_aN^a} \nu_0$
considered in \cite{DonnellyFreidelMoosavianSperanza202012},
with $N_a$ associated with the $\mfk{sl}(2,\mathbb{R})^S$
gravitational Hamiltonian (see section \ref{sec:PJ}).} For simplicity,
we will drop the $\nu_0$ label when working with these algebras, hence we 
denote
\be \label{eqn:sdifffixed}
\mfk{sdiff}(S):= \mfk{sdiff}_{\nu_0}(S),
\qquad 
\mfk{c}_{\mfk{sl}(2,\mbb{R})}(S):=\mfk{c}_{(\mfk{sl}(2,\mbb{R}),\nu_0)}(S).
\ee
Moser's theorem \cite{Moser1965} implies that any volume
form $\nu$ of area $A=\int_S \nu$ is, up to a constant multiple, isomorphic to $\nu_0$,
meaning there exists a diffeomorphism $\Phi$ such that $\Phi^*(\nu) =A\nu_0$.
This implies that the different area-preserving diffeomorphism groups are isomorphic
to each other 
\be 
\mfk{sdiff}_\nu(S)=\Phi( \sdiff).
\ee
A useful analogy is to compare the Diff$(S)$ group to the Lorentz group $\SO(3,1)$ and the area preserving SDiff$_\nu(S)$ group to the rotation subgroup $\SO(3)_p$ preserving the timelike 4-momentum ${p}$
\cite{DonnellyFreidelMoosavianSperanza202012}. The subgroups $\SO(3)_p$ are all isomorphic to the canonical subgroup $\SO(3)$ associated with a reference timelike direction ${p}_0$. The isomorphism is such that $\SO(3)_p= g_p \SO(3)g^{-1}_p $ for a boost $g_pp_0=p$.
 Within this analogy, the full corner symmetry group $\Gslr$ is the analog of the Poincar\'e group, while the centralizer subgroup $C_{\text{SL}(2,\mathbb{R})}$ is analogous to the subgroup $\SO(3)\ltimes \mbb{R}$ preserving the given direction. Finally, the diffeomorphisms that change the area form are analogous to the boost transformations of the Lorentz group which do not preserve $p_0$.

An important subalgebra of  $\gslr$  
is obtained by considering the one-dimensional subalgebra of $\mfk{sl}(2,\mbb{R})$ generated by a single generator. Taking this generator to be a hyperbolic generator --- for example, $\tau_1$ in \eqref{eqn:taua} --- this subalgebra is isomorphic to $\mbb{R}$. Considering functions on $S$ valued in this subalgebra rather than the full $\mfk{sl}(2,\mbb{R})$) yields the subalgebra $\mbb{R}^S\subset\mfk{sl}(2,\mbb{R})^S$, which is just the abelian Lie algebra $C(S)$ of real-valued functions on $S$. Imposing this restriction on the full algebra \eqref{eq:the algebra gS}, we end up with the following subalgebra of $\mfk{g}_{\mfk{sl}(2,\mbb{R})}$
\begin{equation}
    \mfk{g}_{\mbb{R}}(S)=\mfk{diff}(S)\oplus_{\mcal{L}}\mbb{R}^S. 
\end{equation}
This is the hydrodynamical algebra, which is the symmetry algebra of an ideal barotropic fluid \cite{MarsdenRatiuWeinstein1984a}. 
In the present context, this algebra plays an important 
role in the classification of coadjoint orbits of $\mfk{g}_{\mfk{sl}(2,\mbb{R})}$ \cite{DonnellyFreidelMoosavianSperanza202012}. 

 Imposing the same restriction on the centralizer algebra \eqref{eq:the algebra cS}, we arrive at another important subalgebra of $\mfk{g}_{\mfk{sl}(2,\mbb{R})}$
\begin{equation}
    \mfk{c}_{\mbb{R}}(S)=\mfk{sdiff}(S)\oplus_{\mcal{L}}\mbb{R}^S. 
\end{equation}
This algebra appears as the symmetry algebra of a charged particle on a sphere surrounding a magnetic monopole, as recently explored in \cite{AndradeeSilva:2020ofl}.
Turning off the boost generators in \eqref{eq:the algebra gS} and \eqref{eq:the algebra cS}, gives the subalgebras $\mfk{diff}(S)$ and $\mfk{sdiff}(S)$ of $\mfk{g}_{\mfk{sl}(2,\mbb{R})}$, respectively. 
Conversely, turning off the diffeomorphism generators in \eqref{eq:the algebra gS} and \eqref{eq:the algebra cS}, we get the subalgebras $\mfk{sl}(2,\mbb{R})^S$ and $\mbb{R}^S$ of $\mfk{g}_{\mfk{sl}(2,\mbb{R})}$, respectively. 
The Lie bracket of each of these subalgebras is obtained by restriction of \eqref{eq:the Lie bracket of gS} to the corresponding subalgebra. 
Each of these algebras is the Lie algebra of a subgroup of the full corner symmetry group, which we denote
\begin{equation}
    \begin{aligned}
    G_{\tenofo{SL}(2,\mbb{R})}(S)&=\tenofo{Diff}(S)\ltimes\tenofo{SL}(2,\mbb{R})^S, &\qquad C_{\tenofo{SL}(2,\mbb{R})}(S)&=\tenofo{SDiff}(S)\ltimes\tenofo{SL}(2,\mbb{R})^S,
    \\
    G_{\mbb{R}}(S)&=\tenofo{Diff}(S)\ltimes\mbb{R}^S, &\qquad C_{\mbb{R}}(S)&=\tenofo{SDiff}(S)\ltimes\mbb{R}^S.
    \end{aligned}
\end{equation}
The algebra inclusions obtained in this section can be summarized in the following diagram:
\begin{center}
\begin{tikzpicture}[row sep=1.5cm,line width=2pt]
\node at (0,0) {$\sdiff$};
\node at (1.5,0) {$\subset$};
\node at (3.5,0) {$\sdiff \oplus_{\mcal{L}} \R^S$};
\node at (5.5,0) {$\subset$};
\node at (8,0) {$\sdiff \oplus_{\mcal{L}} \slr^S$};
\node at (0,-0.75) {$\cap$};
\node at (3.5,-0.75) {$\cap$};
\node at (8,-0.75) {$\cap$};
\node at (0,-1.5) {$\diff$};
\node at (1.5,-1.5) {$\subset$};
\node at (3.5,-1.5) {$\diff \oplus_{\mcal{L}} \R^S$};
\node at (5.5,-1.5) {$\subset$};
\node at (8,-1.5) {$\diff \oplus_{\mcal{L}} \slr^S$};
\end{tikzpicture}
\end{center}

This concludes our brief synopsis of the relevant subalgebras of $\mfk{g}_{\mfk{sl}(2,\mbb{R})}$. 
In 
later sections, we will focus our attention on the algebras on the top row: $\sdiff$,
$\mfk{c}_{\mbb{R}}(S)$, and $\mfk{c}_{\slr}(S)$ and prove that they can be viewed as large-$N$ limits of the finite-dimensional Lie algebras  $\mfk{su}(N)$, $\mfk{sl}(N,\C)$, and $\mfk{su}(N,N)$, respectively.

\subsection{Poisson bracket representations}\label{sec:PJ}

An important property of the above algebras that will be essential in 
determining their regularizations is that they arise as symmetry algebras of 
classical phase spaces.  Because of this, each algebra deformation considered in section
\ref{sec:matreg} has a natural interpretation in terms of a quantization procedure for the 
associated phase space.  
This section describes how the algebras $\cslr$ and $\crr$
are represented via Poisson brackets on
phase spaces and introduces several quantities related to these representations that 
have direct analogs in the constructions of the deformed algebras.  

The algebra $\gslr$ was identified in \cite{Donnelly:2016auv} as the symmetry algebra
of general relativity restricted to a local subregion bounded by a 2-dimensional surface $S$.
These symmetries arise from diffeomorphisms acting in the vicinity of $S$,
and fail to be pure gauge since the presence of the boundary breaks some 
of the gauge symmetry of the theory.  Instead, these transformations are associated
with nonzero Hamiltonians which generate the action of the transformation on 
the gravitational phase space through Poisson brackets.  Hence, given a
generator $\xi = \xi^A\partial_A$ or $\lambda = \lambda^a\tau_a$ of $\gslr$, the 
corresponding Hamiltonians are given by 
\begin{equation}\label{eqn:Pxi}
    P[\xi]:=\bigintsss_S \xi^A\wt{P}_A, \qquad 
    N[\lambda]:=\bigintsss_S \lambda^a\wt{N}_a,
\end{equation}
where $\wt{P}_A(\sigma)$ and $\wt{N}_a(\sigma)$ are quantities 
related to the geometry of the embedded surface $S$ in spacetime, 
described in detail in \cite{Donnelly:2016auv, DonnellyFreidelMoosavianSperanza202012}. 
Since $\wt{P}_A$ and $\wt{N}_a$ are functions of the dynamical fields in the theory
(namely, the metric), the smeared generators $P[\xi]$, $N[\lambda]$ are 
functions on the phase space.  As such, they obey an algebra defined by the Poisson
bracket on phase space,
which can be shown to satisfy  
\begin{equation}\label{eq:PBgSaction}
    \begin{aligned}
    \{P[\xi_1],P[\xi_2]\}&= P[[\xi_1,\xi_2]_{\tenofo{Lie}}],
    \\
    \{N[\lambda_1],N[\lambda_2]\}&= N[[\lambda_1,\lambda_2]_{\mfk{sl}(2,\mbb{R})}],
    \\
    \{P[\xi],N[\lambda]\}&= N[\mcal{L}_\xi[\lambda]],
    \end{aligned}
\end{equation}
These 
brackets verify that the Hamiltonians $P[\xi]$, $N[\lambda]$ yield a Poisson bracket
representation of the algebra $\gslr$.  

When restricting to the area-preserving diffeomorphisms that appear in 
$\cslr$ and $\crr$, it is more convenient to parameterize the generators in 
terms of their stream functions.  
Since it will be convenient for the vectors to 
reproduce the Poisson brackets $\{,\}_{\upb}$ associated
with the unit radius volume form $\upb = 4\pi \nu_0$, 
we will define the stream function so that $\xi_\phi^B = \upb^{AB}
\partial_A\phi = \frac{1}{4\pi} \nu_0^{AB}\partial_A\phi$.
We can then write the Hamiltonian for an area-preserving diffeomorphism corresponding 
to the vector field $\xi_\phi$ as
\begin{equation} \label{eqn:Jphi}
    \begin{aligned}
    P[\xi_{\phi}]&=\bigintsss_S \xi_{\phi}^{B}\widetilde{P}_B=\bigintsss_S \upb^{AB}\pa_A\phi\wt{P}_B 
    = \bigintsss_S \phi \left(-\frac{1}{4\pi}\varepsilon^{AB}\partial_A P_{B}\right) \\
    &= \bigintsss_S \phi\,\wt{J}:=J[\phi],
    \end{aligned}
\end{equation}
where\footnote{The generator $J$ defined here is $-J$ in \cite{DonnellyFreidelMoosavianSperanza202012}.} 
\begin{equation} \label{eqn:Jdef}
    P_{A}:=\frac{\wt{P}_A}{\wt{n}_0}, \qquad J:=-\upb^{AB}\partial_AP_{ B}=\frac{\wt{J}}{\wt{n}_0}.
\end{equation}
Using \eqref{eq:PBgSaction} and employing the relation (\ref{eqn:xiphibrack}), we have
\begin{equation}\label{eqn:JNPB}
    \begin{aligned}
    \{J[\phi_1],J[\phi_2]\}&=J[\{\phi_1,\phi_2\}_{\upb}],
    \\
    \{J[\phi],N[\lambda]\}&=N[\{\phi,\lambda\}_{\upb}],
    \\
    \{N[\lambda_1],N[\lambda_2]\}&= N[[\lambda_1,\lambda_2]_{\mfk{sl}(2,\mbb{R})}].
    \end{aligned}
\end{equation}
The Poisson bracket of $\crr$ are obtained by simply restricting $\lambda$ to be 
proportional to a single $\Sl(2,\mathbb{R})$ generator, in which case 
the last Poisson bracket in \eqref{eqn:JNPB} vanishes. 

While the above discussion focused on the specific example of the algebras acting 
on the gravitational phase space of \cite{Donnelly:2016auv}, the various objects that 
appear in the description have interpretations in term of natural quantities 
arising for a generic phase space admitting an action of $\cslr$.
Given any such phase space $\phsp$, there exists a unique $\emph{moment map}$ $\mu:\phsp
\rightarrow \cslr^*$ that sends the phase space to the dual of the Lie algebra, which 
is a Poisson manifold 
foliated by the coadjoint orbits
\cite{kirillov2004lectures}.\footnote{Generically, 
the image of this map in $\cslr^*$ will include many
different coadjoint orbits; this is implied by the existence
of nontrivial Casimir functions on the phase space $\phsp$.}  
One can therefore construct the pullback map $\mu^*$ which sends a function on 
$\cslr^*$ to a function on $\phsp$.  
Since any element of $\cslr$ is naturally associated with a linear 
function on $\cslr^*$, the pullback map restricts to an action 
on $\cslr$, and defines a linear
map $\mu^*:\cslr\rightarrow C^\infty(\phsp)$; this is just the map that 
sends a Lie algebra element to its corresponding Hamiltonian on phase space.
This map is explicitly described by a quantity 
$\ham \in  \cslr^*\otimes C^\infty(\phsp)$,
i.e.\ a linear form on the Lie algebra valued in functions on the phase space.   The split
in $\cslr$ and its dual into $\sdiff$ and $\mfk{sl}(2,\mathbb{R})^S$ generators 
leads to a decomposition of $\ham$ into two components, $\ham = (J, N_a)$, with each component
coinciding with the functions $J(\sigma)$ and $N_a(\sigma)$ appearing in (\ref{eqn:Pxi})
and (\ref{eqn:Jdef}).\footnote{The fact that elements of $\cslr^*$ can be identified 
with functions on the sphere comes from the existence of a trace provided by 
the integral over the sphere with respect to the fixed volume form $\nu_0$.  This 
trace gives a canonical identification of $\cslr^*$ with $\cslr$, the 
latter of which is parameterized by functions $\phi$ and $\mfk{sl}(2,\mathbb{R})$-valued 
functions $\lambda_a$.}  Hence, $J$ and $N_a$ should be viewed as $C^\infty(\phsp)$-valued
functions on $S$.

The above discussion of the moment map can be clarified with a simple
example.  
Let $P=\mathbb{R}^3 \times \mathbb{R}^3$ be the phase space of a non-relativistic particle,
with coordinates $(q^i, p_j)$. It is acted upon by the rotation group with generators $X= X_i \sigma^i$ where $\sigma^i$ is the three dimensional matrix
\be 
(\sigma_i)_j{}^k=\varepsilon\indices{_i^k_j},
\ee 
with indices raised by the standard Euclidean metric.  
They satisfy the algebra
\be 
[\sigma_i,\sigma_j]= \varepsilon_{ij}{}^k\sigma_k.
\ee 
The matrices $\sigma_i$ can be taken as a basis for the 
$\mfk{so}(3)$ Lie algebra, and the moment map pullback $\mu^*$
sends each matrix to a function on phase space.  Defining these functions
as $L_i$, we have that
\beq
L_i = \mu^*\sigma_i = \varepsilon_{ijk}q^j p^k,
\eeq
and they satisfy
\beq
\{L_i, L_j\} = \varepsilon\indices{_i_j^k}L_k.
\eeq
If we parameterize the dual of the Lie algebra $\mfk{so}(3)^*$ with the same 
matrices $\sigma^i$, with a pairing defined by $\langle \sigma^i, \sigma_j\rangle
=\delta^i_j$, the moment map $\mu$ sending a point $(q^i, p_j)$
in $P$ to a point in $\mfk{so}(3)^*$ is therefore given by
\beq
\mu(q^i, p_j) = \varepsilon_{ijk}\sigma^i q^j p^k.
\eeq
The object $\ham$ in this case is an element of $\mfk{so}(3)^*\otimes C^\infty(P)$
given by
\beq
\ham = \sigma^i\otimes \varepsilon_{ijk}q^j p^k, 
\eeq
and we easily verify that it satisfies the defining property
\beq
\langle H,\sigma_i\rangle = \mu^*\sigma_i = L_i.
\eeq

\subsection{Mode expansion of Hamiltonian generators}\label{sec:modes}

The determination of the regularized algebras is most easily achieved in an explicit
basis for the generators, so in this section we construct such a basis for the Lie algebra
$\cslr$ and the corresponding phase space generators $J[\phi]$ and $N[\lambda]$ defined 
in (\ref{eqn:Jphi}) and (\ref{eqn:Pxi}). The Lie algebra is parametrized by a 
pair of functions $(\phi, \lambda^a)$ on $S$, where $\phi$ is real-valued and 
$\lambda^a$ is $\mfk{sl}(2,\mathbb{R})$-valued; therefore, we need a basis for these
spaces of functions.  A good choice is the spherical harmonic functions,
which we denote as $Y_\alpha$, where $\alpha = (A,a)$ denotes a pair of integers with 
$A\in\mathbb{N}$ the total angular momentum and $a\in\{-A,\ldots, +A\}$ is the magnetic 
spherical harmonic number.  The conventions employed in this work for the spherical
harmonics are detailed in appendix \ref{sec:sh}.

The pointwise product and Poisson bracket of spherical harmonic functions yield two
types of structure constants, 
\begin{equation}\label{eqn:ECclassical}
   Y_\alpha Y_\beta=E\indices{_{\alpha\beta}^\gamma}Y_\gamma, \qquad\qquad \{Y_\alpha,Y_\beta\}_{\upb}=C\indices{_{\alpha\beta}^\gamma}Y_\gamma.  
\end{equation}
The explicit form of $E\indices{_{\alpha\beta}^\gamma}$ and $C\indices{_{\alpha\beta}^\gamma}$ in terms of Wigner 
$3j$ symbols are given in equations (\ref{eqn:Eabccontra}) and (\ref{eqn:Cabc}).
These structure constants are directly used to construct the structure constants
of $\cslr$.  A basis for this algebra is provided by the $Y_\alpha$ and the quantities 
$Y_{a\alpha} = \tau_a\otimes Y_\alpha $ with 
$\tau_a$ given by \eqref{eqn:taua}.
Applying the relationship between a function and its associated vector field (\ref{eqn:stream}) as well as the identities (\ref{eqn:xiphibrack}) and (\ref{eqn:ECclassical}), the Lie brackets are given by
\beq \label{eqn:cslrsc}
    \begin{aligned}
     [Y_\alpha, Y_\beta] &= C\indices{_\alpha_\beta^\gamma}Y_\gamma, \\   
     [Y_\alpha, Y_{a \beta}] &= C\indices{_\alpha_\beta^\gamma}Y_{a \gamma}, \\
     [Y_{a\alpha}, Y_{b \beta}] &= E\indices{_\alpha_\beta^\gamma} 
     \varepsilon\indices{_a_b^c} Y_{c\gamma}.
    \end{aligned}
\eeq
This basis for $\cslr$ immediately leads to a basis for the Hamiltonian generators, 
given by
\begin{equation} \label{eqn:Jalpha}
    \begin{aligned}
    J_\alpha:= J[Y_\alpha]=\bigintsss_S \nu_0 Y_\alpha J, \qquad\qquad N_{a\alpha}:= N[Y_\alpha \tau_a]=\bigintsss_S \nu_0 Y_\alpha N_{a}.
    \end{aligned}
\end{equation}
The Poisson bracket relations (\ref{eqn:JNPB}) then imply that these basis 
generators satisfy
\begin{equation}\label{NNcommutatorslr}
    \begin{aligned}
    \{J_\alpha,J_\beta\}&= C\indices{_{\alpha\beta}^\gamma}J_\gamma,
    \\
    \{J_\alpha,N_{a\beta}\}&= C\indices{_{\alpha\beta}^\gamma}N_{a\gamma},
    \\
    \{N_{a\alpha},N_{b \beta}\}&=E\indices{_{\alpha\beta}^\gamma}\varepsilon\indices{_{ab}^c}N_{c\gamma},
    \end{aligned}
\end{equation}
which reproduce the structure constants (\ref{eqn:cslrsc}) of the $\cslr$ Lie algebra,
as expected.  Finally, note that in this basis, the $C^\infty(\phsp)$-valued functions
$J(\sigma)$ and $N_a(\sigma)$ discussed in section \ref{sec:PJ} can be written
\begin{equation}\label{eqn:Jsigma}
    \begin{aligned}
    J(\sigma)&=\sum_{\alpha} J_\alpha Y^\alpha(\sigma), 
    \\
    N_a(\sigma)&= \sum_{\alpha} N_{a\alpha} Y^\alpha(\sigma),
    \end{aligned}
\end{equation}
where $Y^\alpha = \delta^{\alpha\beta}Y_\beta$, with $\delta^{\alpha\beta}$ defined 
as the inverse of the spherical-harmonic metric defined in equation (\ref{eqn:shmetric}).

The generators associated with the subalgebra $\crr$ is obtained by restricting the normal generator to be $N_\alpha\equiv N_{1\alpha}$.
The expansion in modes simplifies to:
\begin{equation}\label{NNcommutatorR}
    \begin{aligned}
    \{J_\alpha,J_\beta\}&= C\indices{_{\alpha\beta}^\gamma}J_\gamma,
    \\
    \{J_\alpha,N_{\beta }\}&= C\indices{_{\alpha\beta}^\gamma}N_{\gamma},
    \\
    \{N_{\alpha },N_{\beta }\}&=0.
    \end{aligned}
\end{equation}

Finally, we mention that although it is not the focus of the present
work, one would also like to understand how to lift the algebra deformations
identified in section \ref{sec:matreg} for area-preserving algebras to the full
corner symmetry algebra $\gslr$.  The main obstacle in doing so 
lies in the identification of a suitable deformation of the full 
$\diff$ algebra compatible with the  deformation of its $\sdiff$
subalgebra.  As 
a first step toward investigating this question, we derive in 
appendix \ref{app:diffsc} the structure constants of $\diff$, and discuss
some ideas and challenges in using these to obtain a suitable deformation
of $\diff$ in section \ref{sec:fulldiff}.

\section{Matrix regularizations of classical symmetry algebras}\label{sec:matreg}

Having reviewed the classical symmetry algebras, we are now 
interested in exploring their quantization, in the sense 
described in section \ref{section:introduction}.  
We restrict attention to the centralizer algebra $\cslr$ and its 
subalgebras, since these all possess natural candidates for their
deformation in terms of finite-dimensional matrix algebras.  These 
matrix algebras arise from promoting the sphere on which the 
diffeomorphism groups act to a fuzzy sphere and appealing the 
well-known correspondence between the large $N$ limit of the 
$\mfk{su}(N)$ Lie algebra and $\sdiff$ \cite{Hoppe198201,Hoppe198901,PopeStelle198905}. 
Using the mode expansions of the classical algebras obtained in section \ref{sec:modes}, we obtain quantization maps between the classical generators and corresponding sets of matrices and show that the structure
constants for the matrix product approach the classical structure 
constants in the large $N$ limit.  The final result is that the 
respective matrix regularizations of $\sdiff$, $\sdiff\oplus_{\mcal L}\mbb{R}^S$, and  $\sdiff\oplus_{\mcal L}\slr^S$ are found to be 
$\mfk{su}(N)$, $\mfk{sl}(N,\mathbb{C})\oplus\mathbb{R}$, and $\mfk{su}(N,N)$.

\subsection{From functions on phase space to linear operators on Hilbert space}
\label{sec:fxns2ops}

Quantization is a procedure that seeks to replace the algebra of functions 
on a phase space with the algebra of linear operators on a Hilbert space.  
The quantization map sends each function on phase space to a Hilbert space 
operator, and the commutators of the quantized observables are required 
to reproduce the classical Poisson bracket algebra only up to order $\hbar^2$ 
corrections. 
These higher-order corrections indicate that the Poisson bracket algebra has been deformed.  While it is possible that certain subalgebras remain undeformed
by the quantization procedure, generically one expects a deformation to occur
whenever one is available.  In this case, the classical symmetry algebras 
discussed in section \ref{sec:cornsyms} 
are modified in the quantum theory.

It is possible to identify at the semiclassical level whether an algebra deformation exists or not. 
The presence of a quantum deformation implies the existence of a one-parameter family of deformations of the classical Poisson algebra.
The deformation of the algebra is encoded in the existence of Poisson 2-cocycles \cite{Flato:1995vm}, a notion which is intimately related to Lie algebra 2-cocycles \cite{Fuks} and Hochschild 2-cocycles \cite{Gerstenhaber} which 
parameterize deformations of Lie algebras and algebras respectively. 
As explained in appendix  \ref{appsec:algebra deformation}, a  Poisson 2-cocycle 
for a Poisson manifold $M$ 
is a map   $D :C(M)\times C(M)\to C(M)$, where $C(M)$ denotes the space of
smooth functions
on $M$, which is skew-symmetric, is a bi-derivation, and satisfies the Poisson 2-cocycle  identity.
 Explicitly, this means that
 \bea
D(f,g)&=&-D(g,f),\cr
D(f,gh)&=& D(f,g)h + g D(f,h),\label{cocycle}\\
\{f,D(g,h)\}+\{g,D(h,f)\}+\{h,D(f,g)\} &=&-[D(f,\{g,h\})+D(g,\{h,f\})+D(h,\{f,g\})]\nonumber .
\eea
These identities simply imply that the bracket $\{f,g\}_\lambda:= \{f,g\}+ \lambda D(f,g)$ satisfies the Jacobi identity to first order in $\lambda$.
As described in appendix \ref{appsec:algebra deformation}, in each algebra that we study, there exists a Poisson 2-cocycle that controls the quantum deformation. Ultimately we will find that our Lie algebra deformations have a non-perturbative completion at finite $\lambda \sim 1/N^2$ which satisfies the Jacobi identity exactly.
After having found such a non-perturbative quantization, the identities \eqref{cocycle} can be derived by expanding the Jacobi identity in $\lambda$.

The quantization procedure also requires that the object $\ham$ defined 
in section \ref{sec:PJ} be replaced by its quantum analog, $\wh{\ham}$.  
The classical object $H$ is valued in functions on the phase space, $C^\infty(\phsp)$, 
which is the space of classical observables.  The quantized object
$\wh\ham$ should therefore be valued in the space $\opalg$ of linear operators 
on a Hilbert space $\mathcal{H}$, which serves as the space of observables in the 
quantum theory.  Furthermore, since the quantum theory deals with a deformed 
algebra, $\wh{\ham}$ should be a linear map from this deformed algebra into
$\opalg$, as opposed to a map from the classical algebra.  Therefore, 
we see that $\wh{\ham}\in \wh{\mfk{g}}^*\otimes \opalg$, where $\wh{\mfk{g}}$ 
is the specific deformed algebra under consideration.  
Generically, the deformed algebra $\wh{\mfk{g}}$ depends 
on a deformation parameter
$N$, which will be taken to be large in the semiclassical limit.  
Similar to the classical
object, the map defined by 
$\wh{\ham}$ is required to be a homomorphism, up to a constant rescaling,
from the deformed Lie algebra $\wh{\mfk{g}}$
into $\opalg$, which is simply the statement that the image of this map in $\opalg$ 
furnishes a linear representation of the deformed algebra.  
Note that the generators of this representation are taken to be
$i\hbar\,\pi(X)$, with $X\in\wh{\mfk{g}}$ and $\pi$ denoting a representation.
The factor of $i$ ensures that the generators are Hermitian
in a unitary representation of the algebra, and the factor of $\hbar$ 
is included to give the correct proportionality constant 
between the commutator and the classical Poisson bracket.
The value of $N$ and the specific representation $\pi$ 
of the deformed algebra
that occurs 
depends on the phase space 
being quantized: different phase spaces correspond to different 
deformations and representations.  Both $N$ and $\pi$
can be determined by 
requiring that the generators reproduce the symmetric product
of the classical phase space to leading order in $\hbar$.  
This is most straightforwardly done by matching 
the Casimir functions on the classical phase space to
the values of corresponding Casimir operators in the representation.
This matching procedure is discussed in section \ref{sec:casimirs}.

\subsection{Matrix regularization of $\sdiff$}\label{sec:sdiffsun}
We begin by providing some details on the matrix regularization of 
$\sdiff$, the algebra of vector fields preserving a fixed volume form $\nu_0$ 
 discussed in equation (\ref{eqn:sdifffixed}).  It is well-known that 
the regularized algebra is $\mfk{su}(N)$ \cite{Hoppe198201,Hoppe198901,PopeStelle198905}, and we use this section to 
illustrate the method for obtaining the large $N$ limit of a matrix algebra that will
be subsequently applied to the algebras $\crr$ and $\cslr$.  The results on the limits of the $\mfk{su}(N)$ structure
constants obtained in this section will also play a key role in obtaining
the matrix regularizations of the other algebras of interest in
sections \ref{sec:sdiffslnc} and \ref{sec:sdiffsunn}.

As discussed in section \ref{sec:subalgs}, 
a standard presentation of the Lie algebra $\sdiff$ is in terms of the Poisson
brackets of functions on the sphere.  
In section \ref{sec:modes}, we found that the spherical harmonics $Y_\alpha$
provide a convenient basis for this space of functions.  In terms 
of this basis, 
a generic function $\phi$ can  be expanded as
\begin{equation}\label{eq:the exansion of a generic function on sphere in terms of spherical harmonics}
    \phi=\sum_\alpha \phi^\alpha Y_\alpha, 
\end{equation}
where $\phi^\alpha$ are complex constants. Since $\phi$ is real-valued,
the coefficients must satisfy the reality condition 
\beq
(\phi^\alpha)^* = (-1)^a \phi^{\lmconj{\alpha}},
\eeq
in direct correspondence to the condition (\ref{eqn:YAareality}) satisfied by 
the $Y_\alpha$, recalling that $\lmconj{\alpha} = (A,-a)$ for $\alpha =(A,a)$.  

In the matrix regularization, the functions $Y_\alpha$ are replaced with  fuzzy spherical harmonics $\wh Y_\alpha$, which are  $N \times N$ matrices \cite{Madore199109}. 
Our conventions for fuzzy spherical harmonics are spelled out in Appendix \ref{sec:fuzH}. These matrices obey a multiplication law with structure 
constants $\wh M\ind{_\alpha_\beta^\gamma}$, and satisfy additional normalization
and reality conditions:
\begin{equation}\label{TrM}
\wh Y_\alpha \wh Y_\beta = \wh M_{\alpha \beta}{}^\gamma \wh Y_\gamma, \qquad \frac{1}{N}\tenofo{Tr}_{\bf N}(\wh{Y}_\alpha\wh{Y}_\beta)=\delta_{\alpha\beta}, \qquad \wh Y_\alpha^\dagger=(-1)^a \wh Y_{\lmconj{\alpha}}. 
\end{equation}
with the metric $\delta_{\alpha\beta}$ defined in (\ref{eqn:shmetric}).
As discussed in appendix \ref{sec:fuzH}, the multiplication 
structure constants $\wh{M}_{\alpha\beta\gamma} \equiv \wh{M}\indices{_\alpha
_\beta^\mu}\delta_{\gamma\mu}$ can be expressed explicitly in terms 
of Wigner 3j and 6j symbols according to \cite{Freidel:2001kb}
\beq
\wh{M}_{\alpha\beta\gamma} = \frac{\sqrt{N}}{(-1)^{2J}}\tj{A&B&C\\a&b&c}
\sj{A&B&C\\J&J&J},
\eeq
with $N = 2J+1$. 
The \emph{quantization map} from a function $\phi$ on the sphere to an
$N\times N$ matrix $\wh\phi$ is achieved by expressing $\wh\phi$ in terms of 
fuzzy spherical harmonics with the same coefficients $\phi^\alpha$, 
\begin{equation}\label{eq:the expansion of a Hermitian matrix in the basis of fuzzy spherical harmonics}
    \widehat \phi = \sum_{\alpha\in I_N} \phi^\alpha \widehat Y_\alpha,
\end{equation}
where the sum runs over the index set $I_N$ consisting of all 
$\alpha=(A,a)$ with
$A\leq 2J$. 
Note that this is a finite sum in which all spherical harmonics with $A > 2J$  are truncated. The reality condition in (\ref{TrM}) for the fuzzy
harmonics implies that the quantization map preserves the star structure,
$\wh{\phi^*} = \big(\wh\phi\,\big)^\dagger$, and since  real functions 
satisfy $\phi^* = \phi$, we see that the quantization map sends them to Hermitian
matrices $\wh\phi^\dagger = \wh \phi$.

It is straightforward to see that all possible Hermitian $N\times N$ matrices 
are obtained as the quantization of some function on the sphere, and hence the 
full quantized algebra coincides with the algebra of all $N\times N$ Hermitian 
matrices.  The associated Lie algebra obtained by taking commutators is just the 
standard presentation of the algebra $\mfk{u}(N)$.  This Lie algebra has a trivial
center generated by the   matrix $\wh{Y}_{(0,0)}=\mathbb{I}_N$, which is the quantization of the constant
function $Y_{(0,0)}$ on the sphere.  Since, as discussed in section \ref{sec:subalgs},
the constant function does not generate an area-preserving diffeomorphism, we see that 
the algebra $\sdiff$ quantizes to the space of matrices with vanishing $\wh{Y}_{(0,0)}$ 
component.
These are precisely the traceless Hermitian matrices, and hence 
the quantized Lie algebra
is  $\mfk{su}(N)$.

The classical structure constant relations \eqref{eqn:ECclassical}
possess corresponding relations for the quantized algebra,
coinciding with the symmetric and antisymmetric parts of $\wh{M}\indices{_\alpha_\beta^\gamma}$, 
\begin{align}
    \wh{E}\indices{_\alpha_\beta^\gamma}&=\wh{M}\indices{_{(\alpha}_{\beta)}^\gamma}, \label{MN} 
    \\ \wh{C}\indices{_\alpha_\beta^\gamma}&=\frac{N}{i}\wh M\indices{_{[\alpha}_{\beta]}^\gamma}. \label{MC}
\end{align}
It is shown in Appendix \ref{sec:fuzH} that in the large-$N$ limit, they approach
the classical structure constants
\begin{align}
\wh E_{\alpha\beta}{}^\gamma &= E_{\alpha\beta}{}^\gamma + \mcal{O}\left(N^{-2}\right), \label{ElargeN}\\
\wh C_{\alpha\beta}{}^\gamma &=  C_{\alpha\beta}{}^\gamma + \mcal{O}\left(N^{-2}\right).
\label{ClargeN}
\end{align}
This implies that the quantization map preserves the 
symmetric product and bracket to order $N^{-2}$ and $N^{-3}$, respectively
\begin{equation} \label{eqn:hphihpsi}
    \begin{aligned}
    \wh \phi \circ \wh \psi &=\widehat{ \phi \psi } + \mcal{O}(N^{-2}), \\
[\wh\phi, \wh\psi] &= \frac{2i}{N} \widehat{\{ \phi, \psi \}}_{\upb} + \mcal{O}(N^{-3}).
    \end{aligned}
\end{equation}
where $\wh\phi \circ \wh \psi =\frac12 (\wh\phi \wh \psi+\wh\psi \wh\phi)$ is the symmetrized product.

In this relation, we see that the quantity $2/N$ is playing the 
role of $\hbar$ in the relation between the commutator and Poisson
bracket in (\ref{eqn:hphihpsi}).  However, it is not quite correct to 
equate $\hbar$ with $2/N$, since such a relationship only holds in the special
case of the fuzzy sphere, and will not hold for the quantizations of
the gravitational phase spaces considered in this work.  Instead,
recalling that the Poisson bracket $\{\cdot,\cdot\}_{\upb}$ is defined 
for a spherical phase space with area $A=4\pi$, the correct relation
is  $\hbar_\text{fs} = \frac{A}{2\pi N_\text{fs}}$, or equivalently 
\beq \label{eqn:Nfs}
N_\text{fs} = \frac{A}{2\pi\hbar_\text{fs}}, \qquad (\text{fuzzy sphere}).
\eeq
Here, we have added subscripts ``$\text{fs}$'' to $N$ and $\hbar$ to 
emphasize that this relation only holds for the fuzzy sphere, and 
for other phase spaces (such as the gravitational phase space
that is the primary focus of this work), 
the relation between the two will be different.
This relation should be viewed as determining the deformation parameter
$N_\text{fs}$ in terms of the phase 
space area $A$ and Planck's constant
$\hbar_\text{fs}$.\footnote{Dimensionally, this requires that 
the phase space area $A$ has the same units as $\hbar_\text{fs}$. 
This can be made explicit by defining the symplectic form
for the phase space to be $\Omega_{\text{fs}} = \gamma\frac{A}{4\pi}
\epsilon$, with $\epsilon$ the unit-radius spherical volume form,
$A$ the area of the sphere in standard units, and $\gamma$ a parameter
with dimensions $[\gamma] = \hbar/(\text{length})^2$.
In this case, the relation between $\hbar_\text{fs}$ and $N_\text{fs}$
is
\begin{equation*}
    N_\text{fs} = \frac{\gamma A}{2\pi \hbar_\text{fs}}.
\end{equation*}
The relation (\ref{eqn:Nfs}) holds in units where $\gamma = 1$.
}  
One might worry that this relation is ambiguous since by rescaling the generators $\wh\phi$, $\wh\psi$,
one would obtain a similar relation between the commutator and Poisson
bracket, but with a rescaled value of $\hbar_\text{fs}$.  
Note however that such
a rescaling is not possible, as it spoils the first relation 
in (\ref{eqn:hphihpsi}) for the symmetric product.\footnote{In more detail,
if we work instead with $\widetilde{\phi}\equiv \lambda\wh{\phi}$, $\widetilde{\psi} \equiv
\lambda\wh{\psi}$, we would instead find $\wt{\phi} \circ\wt{\psi}
= \lambda\widetilde{\phi\psi}+\mathcal{O}(N^{-2}) \neq
\widetilde{\phi\psi}+\mathcal{O}(N^{-2})$. }
Hence, the relationship between
$N_\text{fs}$ and $\hbar_\text{fs}$ 
is fully determined for a given phase space by 
requiring that the quantized generators reproduce the symmetric
product at leading order in $N_\text{fs}$, and that the commutator 
equal the Poisson bracket rescaled by $i\hbar_\text{fs}$ to leading order 
in $N_\text{fs}$.  

Finally, we note that the relation $\hbar_\text{fs} = \frac{2}{N_\text{fs}}$
for the unit-radius fuzzy sphere allows us to identify the standard normalization
for the $\mfk{u}(N)$ Lie algebra generators $\algX_\alpha$.  The fuzzy
spherical harmonics $\wh{Y}_\alpha$ occur in the defining representation
$\pi_{\bf N}$ 
of $\mfk{u}(N)$, i.e.\ the representation in terms of $N\times N$ Hermitian
matrices.  Using the relation 
\beq \label{eqn:piNX}
\pi_{\bf N}(\algX_\alpha) = \frac{1}{i\hbar_\text{fs}} \wh{Y}_\alpha = \frac{N_\text{fs}}{2i} \wh{Y}_\alpha,
\eeq
we see that the structure constants for the $\mfk{u}(N)$ Lie algebra
in the basis $\algX_\alpha$ are simply $\wh{C}\indices{_\alpha_\beta^\gamma}$:
\beq \label{eqn:XaXb}
[\algX_\alpha, \algX_\beta] = \wh{C}\indices{_\alpha_\beta^\gamma} \algX_\gamma.
\eeq
The relation (\ref{ClargeN}) then confirms that the large $N$ limit
of the $\mfk{su}(N)$ Lie algebra in this basis coincides with 
$\sdiff$ (recalling that the central generator $\algX_{(0,0)}$ does not
generate a diffeomorphism).  

Note that because $\algX_\alpha$ 
defines a complex basis for $\mfk{u}(N)$, we need to specify
a reality condition to identify the real form of the Lie algebra
under consideration.  This reality condition
 is an antilinear involution $*$ on the Lie algebra, with the real
form determined by the set of generators fixed under the involution.  
This involution acts on the $\algX_\alpha$ basis according to
\beq \label{eqn:Xareality}
\algX_\alpha^* = (-1)^a \algX_{\bar\alpha},
\eeq
and ensures that in a unitary representation $\pi$, the operators 
$\pi(\algX_\alpha)$ satisfy 
$i\pi(\algX_\alpha^*) = (i\pi(\algX_\alpha))^\dagger$.  One easily verifies
that the relation (\ref{TrM}) for $\wh{Y}_\alpha^\dagger$ shows that 
the fuzzy spherical harmonics define a unitary representation
of the algebra.  

It is also useful to relate the $\algX_\alpha$ basis to the standard 
basis of $\mfk{u}(N)$ in terms of elementary matrices 
$E\indices{^i_j}$.  The relation is given by 
\beq\label{eqn:XaEij}
\algX_\alpha  = \frac{N}{2i} \big(\wh{Y}_\alpha\big)
\indices{_i^j}E\indices{^i_j},
\eeq
where $(\wh{Y}_{\alpha})_{i}{}^{j}$ denotes the $ij$\textsuperscript{th} component of the matrix $\wh{Y}_\alpha$,
and one can show that the commutation relations 
(\ref{eqn:XaXb}) and reality condition
(\ref{eqn:Xareality}) imply the standard $\mfk{u}(N)$ brackets 
and involution
in the $E\indices{^i_j}$ basis
\beq \label{eqn:EijEkl}
[E\indices{^i_j}, E\indices{^k_l}] = \delta^k_j E\indices{^i_l}-
\delta^i_l E\indices{^k_j}, \qquad (E\indices{^i_j})^*
 = -\delta_{jk} \delta^{il} E\indices{^k_l}.
\eeq
The equivalence of \eqref{eqn:XaXb} and \eqref{eqn:EijEkl} follows from the identity (see Appendix \ref{app:SUNrelns})
\begin{equation}\label{eq:the identity for product of Yalpha and Ybeta}
    (\wh Y_\gamma)_i{}^j \delta_k^l - \delta_i^j (\wh Y_\gamma)_k{}^l 
= \frac{2i}{N^2} \sum_{\alpha, \beta } \wh{C}_{\gamma }{}^{\alpha \beta } (\wh Y_\alpha)_i{}^l(\wh Y_\beta)_k{}^j.
\end{equation}
Following the discussion of section \ref{sec:fxns2ops}, the deformed algebra appears when quantizing a classical phase space, and the quantum theory yields a 
linear representation of the deformed algebra.  This representation is characterized
by the quantity $ \wh{\ham}\in \mfk{su}(N)^*\otimes\opalg$, 
which we instead call $\wh{J}$
in this section since we are dealing only with the deformation of the $\sdiff$ algebra, 
as opposed to the extended algebras $\crr$ and $\cslr$, which have additional generators.  
Up to rescaling by $i\hbar$, 
the fuzzy spherical harmonics $\wh{Y}_\alpha$ furnish a 
representation for the (complexification of) the $\mfk{su}(N)$ Lie
algebra, and hence can be used as an explicit realization of the 
abstract Lie algebra.  These same matrices can be used to parameterize
the dual $\mfk{su}(N)^*$ by utilizing the trace relation appearing in
(\ref{TrM}).  This allows $\wh{J}$ to instead be viewed as an element of $
\text{Mat}(N)\otimes \opalg$, where $\tenofo{Mat}(N)$ is the space of $N \times N$ matrices,
and the generators of the $\mfk{su}(N)$ algebra on the quantum
Hilbert space are given by
\beq
\wh{J}_\alpha = \frac1{N}\tr_{\mbf{N}}\left(\wh{J} \wh{Y}_\alpha\right),
\eeq
where the product and trace refer to the $\text{Mat}(N)$ factor of $\wh{J}$.  This 
relation is the precise analog of the equation (\ref{eqn:Jalpha}) for the classical 
generators, yielding the correspondence
\begin{equation}\label{Matreg}
    J_\alpha=\bigintsss_S \nu_0 {J}Y_\alpha, \qquad \mrln\qquad  \wh{J}_\alpha=\frac{1}{N}\tenofo{Tr}_{\mbf{N}}\left( \wh{{J}} \wh{Y}_\alpha\right).
\end{equation}
The map $\wh{J}$ is required to be normalized such that the 
generators $\wh{J}_\alpha$ satisfy 
\beq\label{AJa}
[\wh{J}_\alpha, \wh{J}_\beta] = i\hbar\, \wh{C}\indices{_\alpha_\beta^\gamma}
\wh{J}_\gamma,
\eeq
since we recall that $\wh{J}_\alpha = i\hbar \,\pi_{\cal{P}}(\algX_\alpha)$,
where $\pi_{\cal{P}}$ is the representation of the Lie algebra 
corresponding to the phase space $\cal{P}$. Since $\wh{J}_\alpha$
are a complex basis for the $\mfk{u}(N)$ generators, there must be a 
reality condition imposed to ensure a unitary representation of the 
Lie algebra.  This condition is 
\beq\label{eqn:Jdagger}
\wh{J}_\alpha^\dagger = (-1)^a\wh{J}_{\bar\alpha},
\eeq
in direct analogy with the condition
(\ref{TrM}) satisfied by the fuzzy spherical harmonics.  
Conversely, the $\opalg$-valued matrix elements of $\wh{J}$ can be recovered by
summing over the $\wh{Y}_\alpha$ basis according to
\beq \label{eqn:hJij}
\wh{J}\indices{_i^j} = \sum_{\alpha\in I_N}\wh J_\alpha  (\wh{Y}^\alpha)\indices{_i^j}.
\eeq
This is the analog of the classical relation (\ref{eqn:Jsigma}), with the correspondence 
being given by
\be 
J(\sigma)=\sum_{\alpha} J_\alpha Y^\alpha(\sigma) 
\qquad \mrln\qquad \wh{J}_{i}{}^j=\sum_{\alpha\in I_N}\wh{J}_\alpha (\wh{Y}^{\alpha})_{i}{}^{j}.
\ee 
With this choice of normalization and the Hermiticity condition
(\ref{eqn:Jdagger}), the operators comprising the 
matrix elements $\wh{J}\indices{_i^j}$ satisfy
\beq\label{eq:Jij and Jki commutaror} 
[\wh{J}\indices{_i^j},\wh{J}\indices{_k^l}]
=
\frac{\hbar N^2}{2}\left(\delta_i^l \wh{J}\indices{_k^j} - \delta_k^j
\wh{J}\indices{_i^l}\right), 
\qquad (\wh{J}\indices{_i^j})^\dagger = -\delta_{ik}\delta^{jl}\wh{J}\indices{_l^k}.
\eeq

\subsection{Matrix regularization of $\crr$}\label{sec:sdiffslnc}
Having reviewed the matrix regularization of $\sdiff$, we turn now
to a related deformation of the extended algebra 
$\crr = \sdiff \oplus_{\mathcal{L}} \mathbb{R}^S$.  To motivate this 
deformation, we  recall the explicit parameterization of 
this algebra given in section \ref{sec:modes}.  There, the 
$\sdiff$ generators were given in terms of spherical 
harmonics $Y_\alpha$ as before, while the generators 
of the $\mathbb{R}^S$ algebra were written in terms of $Y_{1\alpha}
 = \tau_1\otimes Y_\alpha $, i.e.\ a tensor product between a function
 on the sphere and a $2\times 2$ matrix.  
The idea behind the deformation of the $Y_{1\alpha }$ generators is that 
we should seek to replace the $Y_{\alpha}$ appearing in it with 
its fuzzy version $\wh{Y}_\alpha$, and simply compute the commutators 
of the resulting $2N\times 2N$ matrices.\footnote{A possibly
related discussion of the quantization of matrix-valued functions on 
a fuzzy space is given in \cite{Nair:2020xzn}.}  

There are two subtleties to implementing this idea in practice.  
First, the usual quantization of the $\sdiff$ generators to $N\times N$
matrices clearly will not define a consistent algebra with a set 
of $2N\times 2N$ matrices.  This is easily remedied by simply tensoring
with the $2\times 2$ identity matrix $\mathbbm{1}_{2}$, 
so that the generators 
of the deformed $\sdiff$ subalgebra are now 
$2N\times 2N$ matrices of the form 
$\wh{Y}_{\id\alpha}=\mathbbm{1}_{2}\otimes\wh{Y}_\alpha$.

The second subtlety relates to the relative factor of $i$ in 
the structure constants for the commutator of the fuzzy harmonics 
$\wh{Y}_\alpha$ relative to the Poisson brackets of the classical 
functions $Y_\alpha$ (see, e.g., (\ref{eqn:hphihpsi})). 
This factor of $i$ is simply the ``physicist's''
convention for parameterizing the $\mfk{u}(N)$ Lie algebra in 
terms of Hermitian matrices, as opposed to the ``mathematician's'' convention
which uses anti-Hermitian matrices, and is necessary because 
the commutator of two Hermitian matrices is anti-Hermitian. 
On the other hand, the $\tau_a$ basis (\ref{eqn:taua})
for $\mfk{sl}(2,\mathbb{R})$ 
uses the mathematician's convention in which the structure constants
are real.  Taking tensor products of a set of matrices in the physicist's
convention with a set in the mathematician's convention
yields an algebra in the mathematician's convention.  To obtain 
a tensor product algebra in the physicist's convention, both
algebras in the tensor product should use this convention.  For that
reason, we should instead consider a basis $\rho_a$
of $\mfk{sl}(2,\mathbb{R})$ in which the structure constants are
imaginary.  This basis is described in detail in equation (\ref{matr}) in the 
following section, but for the present construction we simply need 
the form of one of the hyperbolic generators,
\beq
\rho_1 = \frac12\begin{pmatrix} 0&1\\-1&0 \end{pmatrix}.
\eeq
The proposal for the  deformation of the $Y_{1\alpha}$ generators 
is then simply 
$\wh{Y}_{1\alpha} = \rho_1\otimes \wh{Y}_\alpha$.

The Lie algebra obtained from the commutators of these matrices
can be computed directly by applying an identity for the 
commutator of a tensor product of matrices,
\beq \label{eqn:otimesbrack}
[A\otimes C, B\otimes D] = (A\circ B)\otimes[C,D] + [A,B]\otimes
(C\circ D),
\eeq
recalling that $A\circ B = \frac12(AB + BA)$.  Along with the 
expression (\ref{eqn:hphihpsi}) for the structure 
constants of the fuzzy spherical harmonics, this immediately
yields the algebra
\begin{align}
[\wh{Y}_{\id\alpha}, \wh{Y}_{\id\beta}] &= \frac{2i}{N} \wh{C}\indices{_\alpha_\beta^\gamma}\wh{Y}_{\id     \gamma }, \\
[\wh{Y}_{\id\alpha},\wh{Y}_{1\beta}] &=
\frac{2i}{N} \wh{C}\indices{_\alpha_\beta^\gamma}\wh{Y}_{1\gamma},\\
[\wh{Y}_{1\alpha}, \wh{Y}_{1\beta}]&=
-\frac14 \cdot \frac{2i}{N} \wh{C}\indices{_\alpha_\beta^\gamma}
\wh{Y}_{\id \gamma}.
\end{align}

Matching this algebra to the classical algebra $\crr$ is slightly
more subtle than in the case of $\sdiff$.  A reason for the subtlety
is the fact that the representation provided by the matrices
$(\wh{Y}_{\id\alpha}, \wh{Y}_{1\alpha})$ is not unitary, since
matrices of the form $\phi^\alpha \wh{Y}_{1\alpha}$ with
$(\phi^\alpha)^* = (-1)^a \phi^{\bar\alpha}$ are not 
Hermitian.  Hence, this representation does not show up as a
quantization of a classical phase space, unlike the example
provided by the ordinary fuzzy sphere.  Because of this,
the matching to the classical phase space generators need not involve
a universal rescaling by $i\hbar$; instead, different generators
may be scaled by prefactors with different parametric dependence 
on the deformation parameter $N$.  

We denote the generators of the deformed Lie algebra $(\algX_\alpha,
\algZ_\alpha)$, and the $2N$-dimensional representation in which the 
matrices $(\wh{Y}_{\id\alpha},\wh{Y}_{1\alpha})$ live as $\pi_{\bf{2N}}$.
To obtain the correct large-$N$ limit the matrices $\wh{Y}_{\id\alpha}$ generating the $\mfk{su}(N)$ 
subalgebra should be rescaled as in (\ref{eqn:piNX}),
\beq
\pi_{\bf{2N}}(\algX_\alpha) = \frac{N}{2i} \wh{Y}_{\id\alpha},
\eeq
to obtain the bracket
\beq \label{eqn:XX}
[\algX_\alpha, \algX_\beta] = \wh{C}\indices{_\alpha_\beta^\gamma}\algX_\gamma,
\eeq
as before, which matches the first bracket in the classical algebra
(\ref{NNcommutatorR}) as $N\rightarrow\infty$.  
The second bracket in (\ref{NNcommutatorR}) can be matched 
for any choice of scaling for the $Z_\alpha$ generators.  This 
freedom can be parameterized by a quantity $\lambda$ defined so that
\beq
\pi_{\bf{2N}}(\algZ_\alpha) = \frac{N\lambda}{i} \wh{Y}_{1\alpha}.
\eeq
The remaining brackets for the Lie algebra are then fully 
determined to be
\begin{align}
[\algX_\alpha, \algZ_\beta] &= \wh{C}\indices{_\alpha_\beta^\gamma}\algZ_\gamma, \label{eqn:XZ}\\
[\algZ_\alpha,\algZ_\beta]&= -\lambda^2 \wh{C}\indices{_\alpha_\beta^\gamma} X_\gamma. \label{eqn:ZZ}
\end{align}

Here we see that in order to reproduce the final bracket in 
(\ref{NNcommutatorR}), $\lambda$ must go to zero as $N\rightarrow\infty$.
While this still leaves some choice in the precise value of $\lambda$,
the choice $\lambda = \frac{1}{N}$ is most convenient, as it is the value
required when determining the deformation of the larger algebra 
$\cslr$ and determining how the deformation of $\crr$ embeds into
the larger deformed algebra. 

As in the case of $\mfk{u}(N)$, the generators 
$(\algX_\alpha, \algZ_\alpha)$ 
yield a complex basis of the deformed Lie algebra, and hence 
are naturally associated with the complexification of the Lie algebra.
The real Lie algebra is obtained by specifying an antilinear involution
$*$ and restricting to elements that are fixed under the involution.  
The resulting reality condition on $(\algX_\alpha, \algZ_\alpha)$ is
given by
\beq \label{eqn:glncreality}
\algX_\alpha^* = (-1)^a \algX_{\bar\alpha}, 
\qquad \algZ_\alpha^* = (-1)^a\algZ_{\bar\alpha}.
\eeq
The reality condition leads to a criterion for specifying whether
a given representation of the Lie algebra is unitary, namely, that
$i\pi(\algX_\alpha^*) = (i\pi(\algX_\alpha))^\dagger$, and similarly
for $\algZ_\alpha$.  Note that because the generators $\wh{Y}_{1\alpha}$
in the representation $\pi_{\bf{2N}}$  do not satisfy this condition,
we see once again that this representation is not unitary. 

The Lie algebra defined by the brackets
(\ref{eqn:XX}), (\ref{eqn:XZ}), and (\ref{eqn:ZZ}) along 
with the reality condition (\ref{eqn:glncreality}) in fact 
coincides with $\mfk{gl}(N,\mathbb{C})$,  viewed 
as a real Lie algebra, which is the complexification
of $\mfk{u}(N)$.  This can be seen by noting that $\algX_\alpha$
generate a $\mfk{u}(N)$ algebra, and its complexification
is obtained by adding generators $\ii \otimes\algX_\alpha$, where 
$\ii$ is an imaginary unit satisfying $\ii^2 = -1$.\footnote{The
symbol $\ii$ is used to distinguish this imaginary 
unit from the factors of $i$ appearing when using the complex
basis $\algX_\alpha$ for the $\mfk{u}(N)$ Lie algebra.  The  distinction is important, since, for example, the reality
condition $(\ii\otimes\algX_\alpha)^* = \ii\otimes\algX_\alpha^*
 = (-1)^a\ii\otimes\algX_{\bar\alpha}$ is essentially
 equivalent to assuming $\ii^* = +\ii$.}  The brackets of the new
generators are fixed by assuming $\ii$ commutes with the original
generators, so 
$[\algX_\alpha, \ii\otimes \algX_\beta] = 
\wh{C}\indices{_\alpha_\beta^\gamma}\ii\otimes \algX_\gamma$ and 
$[\ii\otimes\algX_\alpha,\ii\otimes\algX_\beta] = 
\wh{C}\indices{_\alpha_\beta^\gamma}\ii^2\otimes\algX_\gamma
=-\wh{C}\indices{_\alpha_\beta^\gamma}\algX_\gamma$, 
which precisely match the brackets (\ref{eqn:XZ}) and (\ref{eqn:ZZ})
upon identifying $\algZ_\alpha = \lambda \ii\otimes \algX_\alpha$.  
This verifies that $(\algX_\alpha, \algZ_\alpha)$ generate 
the Lie algebra $\mfk{gl}(N,\mathbb{C})$.\footnote{This can further be verified by constructing generators 
$E^\pm_{\alpha} = \frac12(\algX_\alpha \pm \frac{i}{\lambda} \algZ_\alpha)$,
which can be shown to satisfy $[E_\alpha^\pm, E_\beta^\pm] 
= \wh{C}\indices{_\alpha_\beta^\gamma}E_\gamma^\pm$, $[E_\alpha^\pm,
E_\beta^\mp] = 0$.}   It is worth
pointing out that this procedure involving tensoring 
with the imaginary unit $\ii$ is more or less equivalent to the 
construction of the $\wh{Y}_{1\alpha}$ generators in the representation $\pi_{\bf{2N}}$, with $2\rho_1$ 
serving as the new imaginary unit $\ii$.\footnote{Note that for this algebra,
since $\lambda$ and $N$ may in principle be chosen independently,
we could instead take the limit $\lambda\rightarrow 0$ before taking 
$N\rightarrow\infty$.  This implements an In\"on\"u-Wigner
contraction \cite{InonuWigner195306} of the algebra $\mfk{gl}(N,\mathbb{C})$ to 
$\mfk{u}(N)\ltimes \mathbb{R}^{N^2}$.  This contraction is effectively
still happening in the large $N$ limit when $\lambda$ is identified
with $\frac{1}{N}$, and explains why
the large $N$ limit of the semisimple Lie algebra results in
an algebra with instead a semidirect product structure.}

The final step in making contact with the continuum algebra $\crr$
is to determine which central generators in $\mfk{gl}(N,\mathbb{C})$
have classical counterparts on the gravitational phase space.  
As noted before, the generators 
$\algX_{00}$
(i.e.\ the generator with $\alpha = (A,a)  = (0,0)$)
do not generate a diffeomorphism
of the sphere, and hence should be discarded when matching the 
continuum algebra.  On the other hand, the generator $\algZ_{00}$
coincides with the global boost in the normal plane to the 
codimension-$2$ surface in spacetime, and remains an important part of the 
continuum algebra.  The remaining generators $(\algX_\alpha,\algZ_\alpha)$
with $A\geq1$ produce the simple subalgebra $\mfk{sl}(N,\mathbb{C})$.  
Hence, we can conclude that the deformation of the continuum
algebra $\crr$ is $\mfk{sl}(N,\mathbb{C})\oplus \mathbb{R}$,
with the generator of the central factor 
$\mathbb{R}$ coinciding with $\algZ_{00}$.

While the representation $\pi_{\bf{2N}}$ is useful in obtaining the 
deformed algebra, it has the property that is it not an irreducible
representation of $\mfk{sl}(N,\mathbb{C})\oplus\mathbb{R}$,
as can be seen from the fact that the central generator $\wh{Y}_{1,00}$
is not proportional to the identity.  
In section \ref{sec:cRlim} when evaluating the Casimir operators for this
algebra, it will be useful to instead have an irreducible faithful
representation of this algebra.  This is given by the standard 
$N$-dimensional
vector representation $\pi_{\bf N}$ of $\mfk{gl}(N,\mathbb{C})$, in 
which
\beq
i\pi_{\bf N}(\algX_\alpha) = \frac{N}{2} \wh{Y}_\alpha,\qquad
i\pi_{\bf N}(\algZ_\alpha) = iN\lambda\wh{Y}_\alpha.
\eeq
Here it is clear that this representation is just the complexification
of the representation of $\mfk{u}(N)$ provided by the matrices
$\wh{Y}_\alpha$.

Finally, to relate this algebra to the classical phase space, 
we should exhibit the map $\wh{H}\in \big(\mfk{sl}(N,\mathbb{C})\oplus
\mathbb{R}\big)^*\otimes \opalg$, with $\opalg$ chosen to be the space
of operators in which the representation of the algebra is valued.
Similar to the case of $\mfk{u}(N)$, 
this is equivalent to defining a pair $\wh{J}, \wh{N}\in \text{Mat}(N)\otimes
\opalg$, which produce a set of  generators 
$(\wh{J}_\alpha, \wh{N}_\alpha)$ in $\opalg$ labeled 
by the fuzzy spherical harmonics according to 
\begin{align}
\wh{J}_\alpha &= \frac1N \tr_{\mbf{N}}\left(\wh{J} \wh{Y}_\alpha\right),
\\
\wh{N}_\alpha &= \frac1N \tr_{\mbf{N}} \left(\wh{N}\wh{Y}_\alpha\right).
\end{align}
The maps $(\wh{J}, \wh{N})$ must be normalized so that 
$(\wh{J}_\alpha, \wh{N}_\beta)
 = \Big(i\hbar\, \pi_{\cal{P}} (\algX_\alpha), i\hbar\,\pi_{\cal{P}}(\algZ_\beta) \Big)$,
yielding the following algebra satisfied by the generators:
\begin{align}
[\wh{J}_\alpha, \wh{J}_\beta] &
= i\hbar \,\wh{C}\indices{_\alpha_\beta^\gamma}\wh{J}_\gamma, \\
[\wh{J}_\alpha, \wh{N}_\beta] &
= i\hbar\, \wh{C}\indices{_\alpha_\beta^\gamma}\wh{N}_\gamma, \\
[\wh{N}_\alpha, \wh{N}_\beta] &
= -i\hbar\, \lambda^2 \wh{C}\indices{_\alpha_\beta^\gamma}\wh{J}_\gamma.
\end{align}
The correspondence between the classical and quantum generators of the 
algebra is therefore given by
\begin{equation}
    \begin{aligned}
    J_\alpha&=\bigintsss_S \nu_0 {J}Y_\alpha,
    \\
    N_\alpha&=\bigintsss_S \nu_0 N Y_\alpha,
    \end{aligned}
    \qquad\mrln\qquad
    \begin{aligned}
    \wh{J}_\alpha&=\frac{1}{N}\tenofo{Tr}_{\mbf{N}}\left(\wh{J}\wh{Y}_\alpha\right),
    \\
    \wh{N}_\alpha&=\frac{1}{N}\tenofo{Tr}_{\mbf{N}}\left(\wh{N}  \wh Y_\alpha\right).
    \end{aligned}
\end{equation}
The inverse of this relation expresses the $\opalg$-valued matrix elements 
$(\wh{J}\indices{_i^j}, \wh{N}\indices{_i^j})$ as a sum over the generators 
$(\wh{J}_\alpha,\wh{N}_\beta)$, 
\beq
\wh{J}\indices{_i^j} = \sum_{\alpha\in I_N} \wh{J}_\alpha
\big(\wh{Y}^\alpha\big)\indices{_i^j},\qquad
\wh{N}\indices{_i^j} = \sum_{\alpha\in I_N} \wh{N}_\alpha 
\big(\wh{Y}^\alpha \big)\indices{_i^j}.
\eeq
As before, the matrix elements $(\wh{J}\indices{_i^j}, \wh{N}\indices{_i^j})$
are the regularized version of the phase space functions $(J(\sigma), N(\sigma))$. 
Finally, the fact that the representation $\pi_{\cal{P}}$ corresponding to
the quantization of the phase space $\cal{P}$ should be unitary implies 
that the generators satisfy
\beq
\wh{J}_\alpha^\dagger = \wh{J_\alpha^*} = (-1)^a \wh{J}_\alpha, \qquad
\wh{N}_\alpha^\dagger = \wh{N_\alpha^*} = (-1)^a \wh{N}_\alpha.
\eeq
A final comment is in order on the uniqueness of the deformation  of $\crr$
obtained in this section.  
Instead of using the generators $\wh{Y}_{1\alpha}$
to arrive at the deformation, one could instead 
work with $\wh{Y}_{0\alpha} = \rho_0\otimes \wh{Y}_\alpha$, where $\rho_0 =
\displaystyle\frac12\begin{pmatrix}  1&0\\0&-1\end{pmatrix}$
is an elliptic generator of $\mfk{sl}(2,\mathbb{R})$ 
satisfying $\rho_0^2 = +\frac14$.
The entire discussion goes through as before, with the only change 
being that the bracket (\ref{eqn:ZZ}) now comes with a coefficient
$+\lambda^2$.  Since $\lambda\rightarrow0$ in the large $N$ limit, this 
gives the same classical algebra in the limit.  The deformed algebra
in this case is the compact algebra $\mfk{su}(N)\oplus\mfk{su}(N)\oplus
\mfk{u}(1)$, as opposed to $\mfk{sl}(N,\mathbb{C}) \oplus \mathbb{R}$.
The fact that one can obtain the same algebra as a contraction in different ways is not surprising. A familiar similar example is that of the Euclidean group $\mfk{so}(3) \ltimes \R^3$, which can be obtained as an In\"on\"u-Wigner contraction of either the noncompact Lie algebra $\mathfrak{so}(1,3) = \Sl(2, \C) $ or  the compact Lie algebra $\mathfrak{so}(4) = \su(2) \oplus \su(2)$.
While the limiting algebra is the same, the groups are quite different: in particular, the choice of contraction can determine whether the deformed 
group is compact or noncompact. 
In the present context, however, the noncompact deformation 
better matches the nature of the classical algebra in which the additional
generators correspond to boosts, as opposed to rotations, of the normal plane 
of the codimension-2 sphere. 
Additionally, we will see that the noncompact
deformation is the correct choice
when embedding into the deformation of the larger algebra $\cslr$ considered in 
the next section.

\subsection{Matrix regularization of $\cslr$}\label{sec:sdiffsunn}

The determination of 
the matrix regularization of $\cslr$ follows a similar procedure to the 
case of $\crr$ considered in section \ref{sec:sdiffslnc}.  Using the 
explicit parameterization of the algebra given in section \ref{sec:modes},
the $\sdiff$ generators are again labeled by the spherical harmonics 
$Y_\alpha$, and the generators of the $\mfk{sl}(2,\mathbb{R})^S$ algebra
take the form $Y_{a\alpha} = \tau_a\otimes Y_\alpha$, where $\tau_a$ 
are the basis of $\mfk{sl}(2,\mathbb{R})$ given in 
(\ref{eqn:taua}).  Once again, we determine the deformed algebra 
by promoting the spherical harmonics appearing in these generators
to fuzzy spherical harmonics and computing the matrix commutators.
As before, the deformed $\sdiff$ generators are  obtained by 
tensoring with the $2\times 2$ identity $\wh{Y}_{\id\alpha} = \mathbbm{1}_{2}
\otimes \wh{Y}_\alpha$.  To arrive at the deformed $\mfk{sl}(2,\mathbb{R})^S$
generators, we recall the discussion in section \ref{sec:sdiffslnc} regarding
properties of the tensor product of matrices using physicist's versus 
mathematician's conventions for the algebra.  The conclusion is that 
in order to obtain a consistent algebra after taking the tensor product, 
we must use a basis for $\mfk{sl}(2,\mathbb{R})$ in which 
the structure constants are purely imaginary.  This basis is given by
\begin{equation} \label{matr}
    \rho_0 =  \frac12 \begin{pmatrix} 1 & 0 \\ 0 & -1 \end{pmatrix}, \qquad
    \rho_1 = \frac12 \begin{pmatrix} 0 & 1 \\ -1 & 0 \end{pmatrix}, \qquad
    \rho_2 = \frac12 \begin{pmatrix} 0 & -i \\ -i & 0 \end{pmatrix},
\end{equation}
whose product satisfies
\beq \label{eqn:rhoarhob}
\rho_a \rho_b = -\frac14\eta_{ab}\mathbbm{1}_{2} +\frac{i}{2}\varepsilon\indices{_a_b^c}\rho_c.
\eeq

It is useful to recall that this basis arises naturally in the 
presentation of $\mfk{sl}(2,\mathbb{R})$ in terms of the isomorphic
algebra $\mfk{su}(1,1)$.  The latter is the Lie algebra
of $\SU(1,1)$, consisting of $2\times 2$ complex 
matrices $g$ of unit determinant preserving an indefinite Hermitian
form $h$, 
\beq\label{eqn:h}
g^\dagger h g = h, \qquad h = \begin{pmatrix}1&0\\0&-1\end{pmatrix}.
\eeq
Expressing $g$ as the exponential of a Lie algebra generator $g= \exp(iT)$, 
preservation of $h$ translates to the condition 
\beq \label{eqn:Th}
T^\dagger h = h T,
\eeq
which indeed is satisfied by the matrices $\rho_a$.

We now take the deformed $\mfk{sl}(2,\mathbb{R})^S$ generators to be of the form
$\wh{Y}_{a\alpha} = \rho_a\otimes \wh{Y}_\alpha$.  Again employing the 
identity (\ref{eqn:otimesbrack}) for the bracket of the tensor
product of matrices, we find that the commutators of the 
matrices $(\wh{Y}_{\id\alpha}, \wh{Y}_{a\alpha})$ satisfy
\begin{align}
[\wh{Y}_{\id\alpha}, \wh{Y}_{\id\beta}] 
&= \frac{2i}{N}\wh{C}\indices{_\alpha_\beta^\gamma}\wh{Y}_{\id\gamma}, \\
[\wh{Y}_{\id\alpha},\wh{Y}_{a\beta}] 
&= \frac{2i}{N} \wh{C}\indices{_\alpha_\beta^\gamma}\wh{Y}_{a\gamma}, \\
[\wh{Y}_{a\alpha}, \wh{Y}_{b \beta}] 
&=  i\varepsilon\indices{_a_b^c}\wh{E}\indices{_\alpha_\beta^\gamma} \wh{Y}_{c\gamma} 
-\frac{i}{2N}\eta_{ab}\wh{C}\indices{_\alpha_\beta^\gamma}\wh{Y}_{\id\gamma},
\end{align}
where the last bracket applies equations (\ref{MN}), (\ref{MC}), and 
(\ref{eqn:rhoarhob}) for the structure constants of the symmetric
and antisymmetric products of $\wh{Y}_\alpha$ and of $\rho_a$.

As in the case of the deformation of $\crr$, these matrices do not 
provide a unitary representation of a Lie algebra, and hence should not 
be viewed as a quantization of a phase space.  Because of this, when matching
to the classical algebra, we are again free to rescale the generators
by prefactors with different parametric dependence on $N$.  Unlike the case
of $\crr$, in the present context matching to the classical 
algebra (\ref{eqn:cslrsc})
fully determines the choice of prefactor.  Denoting the basis for the 
deformed Lie algebra as $(\algX_\alpha, \algZ_{a\alpha})$ and $\pi_{\bf{2N}}$
the representation in which the matrices $(\wh{Y}_{\id\alpha},\wh{Y}_{a\alpha})$
live, the required scaling between the algebra generators and the 
matrices is given by
\begin{align}
\pi_{\bf{2N}}(\algX_\alpha) &= \frac{N}{2i}\wh{Y}_{\id\alpha},
\label{eqn:pi2NX}\\
\pi_{\bf{2N}}(\algZ_{a\alpha}) &=  \frac{1}{i}\wh{Y}_{a\alpha}. \label{eqn:pi2NZ}
\end{align}
This implies the following brackets for the deformed Lie algebra generators
\begin{align}
[\algX_\alpha, \algX_\beta] 
&= \wh{C}\indices{_\alpha_\beta^\gamma}\algX_\gamma, \label{eqn:XXsunn}\\
[\algX_\alpha, \algZ_{a\beta}]
&= \wh{C}\indices{_\alpha_\beta^\gamma}\algZ_{a\gamma}, \label{eqn:XZsunn} \\
[\algZ_{a\alpha},\algZ_{b \beta}] 
&=\varepsilon\indices{_a_b^c}\wh{E}\indices{_\alpha_\beta^\gamma}\algZ_{c\gamma}
-\frac{1}{N^2}\eta_{ab}\wh{C}\indices{_\alpha_\beta^\gamma}\algX_\gamma.
\label{eqn:ZZsunn}
\end{align}
Comparing to (\ref{eqn:cslrsc}), we see that these brackets match
the $\cslr$ algebra in the limit $N\rightarrow\infty$.  Note that the 
scaling of the generators $\algZ_{a\alpha}$ in
(\ref{eqn:pi2NZ}) corresponds to the preferred 
choice $\lambda = \frac{1}{N}$ discussed below (\ref{eqn:ZZ}) for the 
similar case of the $\crr$ deformation.

As before, $(\algX_\alpha, \algZ_{a\alpha})$ define a complex basis 
for the deformed Lie algebra.  The reality condition to specify the real
Lie algebra again descends from the reality condition for the spherical
harmonics, and is given by
\beq \label{eqn:unninv}
\algX^*_\alpha = (-1)^a \algX_{\bar\alpha},\qquad \algZ_{a\alpha}^* = (-1)^a
\algZ_{a\bar\alpha}.
\eeq
This reality condition determines whether a given representation is unitary
by the requirement that $i\pi(\algX_\alpha^*) = (i\pi(\algX_\alpha) )^\dagger$ and 
similarly for $\algZ_{a\alpha}$.  The fact that the matrices 
$\wh{Y}_{1\alpha}$ and 
$\wh{Y}_{2\alpha}$ do not satisfy this condition verifies that the 
representation $\pi_{\bf{2N}}$ is not unitary.

In the $\pi_{\bf{2N}}$ representation, the combination 
of generators $(\wh{Y}_{\id\alpha}, \wh{Y}_{a\alpha})$ that are fixed under
the involution $*$ are all of the form $(\mathbbm{1}_{2}\otimes \wh{A}, 
\rho_a\otimes \wh{B})$, with $\wh{A}$, $\wh{B}$ 
 Hermitian $N\times N$ matrices.  This characterization of the 
generators allows us to identify the Lie algebra defined 
by the brackets (\ref{eqn:XXsunn}), (\ref{eqn:XZsunn}), and (\ref{eqn:ZZsunn}).
Defining a Hermitian form $\wh{h}$ of signature $(N,N)$ given by
\beq \label{eqn:hath}
\wh{h} = h\otimes \mathbbm{1}_N =
\begin{pmatrix}\mathbbm{1}_N&0\\0&-\mathbbm{1}_N \end{pmatrix},
\eeq
with $h$ the $2\times 2$ mixed signature Hermitian form from (\ref{eqn:h}),
we find that the generators preserve $\wh{h}$ in the sense of 
satisfying the analogous condition to (\ref{eqn:Th}):
\begin{align}
(\mathbbm{1}_{2}\otimes\wh{A})^\dagger \wh h &=
h\otimes \wh{A}^\dagger = h\otimes \wh{A} = \wh{h}(\mathbbm{1}_{2}\otimes\wh{A}), \\
(\rho_a\otimes\wh{B})^\dagger \wh h &=
\rho_a^\dagger h \otimes \wh{B}^\dagger
= h\rho_a \otimes\wh{B}
= \wh{h}(\rho_a\otimes \wh{B}).
\end{align}
As preservation of $\wh{h}$ is the defining property of the Lie algebra 
$\mfk{u}(N,N)$, we immediately conclude that the algebra 
defined by the generators $(\algX_\alpha, \algZ_{a\alpha})$ is $\mfk{u}(N,N)$.

Just as in the case of $\mfk{u}(N)$, the Lie algebra for $\mfk{u}(N,N)$
can be parameterized in a basis of elementary matrices 
$E\indices{^{\msf m}_{\msf n}}$, where $\msf m, \msf n = 1,\ldots, 2N$.  
In terms of these, the $(\algX_\alpha, \algZ_{a\alpha})$
generators are given by
\begin{align}
\algX_\alpha &= \frac{N}{2i}\big(\wh{Y}_{\id\alpha}\big)\indices{_{\msf m}
^{\msf n}} E\indices{^{\msf m}_{\msf n}},
\\
\algZ_{a\alpha}&= \frac{1}{i} \big(\wh{Y}_{a\alpha}\big)\indices{_{\msf m}
^{\msf n}} E\indices{^{\msf m}_{\msf n}}.
\end{align}
As usual, the Lie brackets in the $E\indices{^{\msf m}_{\msf n}}$ basis are 
given by (see Appendix \ref{sec:su(N,N) idesntities} for the proof)
\beq \label{eqn:Emnbrack}
[E\indices{^{\msf m}_{\msf n}}, E\indices{^{\msf p}_{\msf q}}] =
\delta^{\msf p}_{\msf n}E\indices{^{\msf m}_{\msf q}}
-\delta^{\msf m}_{\msf q}E\indices{^{\msf p}_{\msf n}},
\eeq
and the involution (\ref{eqn:unninv}) becomes
\beq
(E\indices{^{\msf m}_{\msf n}})^* = - \wh{h}_{\msf{np}} \wh{h}^{\msf{mq}}
E\indices{^{\msf p}_{\msf q}},
\eeq
where $\wh{h}_{\msf{np}}$ is the Hermitian form defined by (\ref{eqn:hath}),
and $\wh{h}^{\msf{mq}}$ is its inverse.  The inverse relation between
the two bases is 
\beq
E\indices{^{\msf m}_{\msf n}} = 
\frac{i}{N}\left(\frac{1}{N} 
\algX^\alpha\big(\wh{Y}_{\id\alpha}\big)\indices{_{\msf{n}}^{\msf m}}
-2 \algZ^{a\alpha}\big(\wh{Y}_{a\alpha}\big)\indices{_{\msf n}^{\msf m}}
\right).
\eeq

Finally, we recall that the central generator $\algX_{00}$ 
is not included when matching to 
the continuum algebra since the constant function
on the sphere does not generate a diffeomorphism.  The algebra obtained
by excluding this generator from the deformed algebra is then the 
simple Lie algebra $\mfk{su}(N,N)$.  We, therefore, arrive at one of
our main results, that $\mfk{su}(N,N)$ defines a finite-dimensional
deformation of the continuum algebra $\cslr$, or, equivalently, that
the large $N$ limit of $\mfk{su}(N,N)$ can be identified
with $\cslr$.

Just as $\crr$ is the subalgebra of $\cslr$ obtained by 
restricting to $\tau_1$ generators of $\slr^S$, the 
deformation $\mfk{sl}(N,\mathbb{C}) \oplus \mathbb{R}$ occurs as a 
subalgebra of $\su(N,N)$ by including only the $\algZ_{1\alpha}$ generators.  
In the $\pi_{\bf{2N}}$ representation, this subalgebra can equivalently
be characterized as the collection of generators that commute with 
$2\wh{Y}_{1,00} = \begin{pmatrix}0&\mathbbm{1}_N\\-\mathbbm{1}_N&0\end{pmatrix}$.
This is interesting because $(2\wh{Y}_{1,00})^2 = -\mathbbm{1}_{\bf{2N}}$, and 
hence defines a complex structure in this representation.  Hence we see
that $\mfk{sl}(N,\mathbb{C})\oplus\mathbb{R}$ can be viewed as the subalgebra
of $\mfk{su}(N,N)$ preserving a complex structure in the defining representation.

Having identified the deformed algebra, we can now relate this 
algebra to the quantization of the classical phase space.  This 
requires specifying the map $\wh{H}\in \mfk{su}(N,N)^*\otimes
\opalg$, with $\opalg$ the space of operators in which the representation
$\pi_{\cal{P}}$ defining the quantization is valued.
Again, this map can be specified by the quantities $\wh{J}, \wh{N}_a \in
\text{Mat}(N)\otimes \opalg$ which yield generators in the 
representation $\pi_{\cal{P}}$ by the relation 
\begin{align}
\wh{J}_\alpha &= \frac{1}{N}
\tr_{\mathbf{N}}\left(\wh{J}\wh{Y}_{\alpha}\right), \\
\wh{N}_{a\alpha}&= \frac{1}{N}
\tr_{\mathbf{N}}\left(\wh{N}_a\wh{Y}_{\alpha}\right).
\end{align}
The normalization condition for the maps $(\wh{J}, \wh{N}_a)$ is again 
chosen so that $(\wh{J}_\alpha, \wh{N}_{a\alpha}) = 
\Big(i\hbar \pi_{\cal{P}}(\algX_\alpha), i\hbar \pi_{\cal{P}}(\algZ_{a\alpha})
\Big)$, so that the generators satisfy the algebra
\begin{align}
[\wh{J}_\alpha, \wh{J}_\beta] &= i\hbar \wh{C}\indices{_\alpha_\beta^\gamma}
\wh{J}_\gamma, \\
[\wh{J}_\alpha, \wh{N}_{a\beta}]&=
i\hbar \wh{C}\indices{_\alpha_\beta^\gamma} \wh{N}_{a\gamma}, \\
[\wh{N}_{a\alpha}, \wh{N}_{b \beta}] &=
i\hbar\left(\varepsilon\indices{_a_b^c}\wh{E}\indices{_\alpha_\beta^\gamma}
\wh{N}_{c\gamma}
-\frac{1}{N^2} \eta_{ab}\wh{C}\indices{_\alpha_\beta^\gamma}\wh{J}_\gamma\right).
\end{align}
This again produces the correspondence between classical and quantum
generators of the algebra,
\begin{equation}\label{eq:the matrix regularization procedure of cslr}
    \begin{aligned}
    J_\alpha&=\bigintsss_S \nu_0 {J}Y_\alpha,
    \\
    N_{a\alpha}&=\bigintsss_S \nu_0 N_a Y_\alpha,
    \end{aligned}
    \qquad\mrln\qquad
    \begin{aligned}
    \wh{J}_\alpha&=\frac{1}{N}
    \tenofo{Tr}_{N}\left(\wh{J}\wh{Y}_\alpha\right),
    \\
    \wh{N}_{a\alpha}&=\frac{1}{N}\tenofo{Tr}_{N}\left(\wh{N}_a 
    \wh Y_\alpha\right).
    \end{aligned}
\end{equation}
The requirement that the representation be unitary 
follows from the involution (\ref{eqn:unninv}), and implies that 
the generators satisfy
\beq
\wh{J}_\alpha^\dagger = \wh{J_\alpha^*} = (-1)^a\wh{J}_{\bar\alpha},\qquad
\wh{N}_{a\alpha}^\dagger = \wh{N_{a\alpha}^*}
 = (-1)^a\wh{N}_{a\bar{\alpha}}.
\eeq

It is also convenient to introduce a set of generators tied 
to the elementary matrix basis.  Using $\wh{J}$ and $\wh{N}_a$,
we can construct a quantity $\opH \in \text{Mat}_{2N}\otimes \opalg$ 
whose matrix elements are
\beq \label{eqn:hHmn}
\opH\indices{_{\msf m}^{\msf n}} 
= \frac{1}{N}\wh{J}^\alpha
\big(\wh{Y}_{\id\alpha}\big)\indices{_{\msf{m}}^{\msf{n}}}
-2\wh{N}^{a\alpha}
\big(\wh{Y}_{a\alpha}\big)\indices{_{\msf{m}}^{\msf n}}.
\eeq
By making the split $\msf m = (M,i)$, $\msf n = (N,j)$ where $M,N = 1,2$
are indices in the 2D representation of $\mfk{sl}(2,\mathbb{R})$ and 
$i,j = 1,\ldots, N$ are $\mfk{su}(N)$ indices, this relation can equivalently 
be expressed in block diagonal form,
\beq
\wh{H} = \begin{pmatrix}\frac1N\wh J +\wh{N}_0 & -\wh{N}_1+i\wh{N}_2 \\
\wh{N}_1+i\wh{N}_2 & \frac1N\wh{J}-\wh{N}_0\end{pmatrix}.
\eeq
The commutators of the operators $\opH\indices{_{\msf m} ^{\msf n}}$ 
are rescaled relative to the bracket (\ref{eqn:Emnbrack}) according to
\beq
[\opH\indices{_{\msf{m}}^{\msf{n}}}, \opH\indices{_{\msf p}^{\msf q}}]
=\hbar N\left(\delta_{\msf m}^{\msf q} \opH\indices{_{\msf p}^{\msf{n}}} 
-\delta_{\msf p}^{\msf n}\opH\indices{_{\msf m}^{\msf q}}\right),
\eeq
and the Hermiticity condition they satisfy is 
\beq
\big(\opH\indices{_{\msf m}^{\msf n}}\big)^\dagger
= - \wh{h}_{\msf{mp}}\wh{h}^{\msf{nq}} \opH\indices{_{\msf{q}}^{\msf{p}}}.
\eeq

\section{The large-$N$ correspondence of Casimirs}
\label{sec:casimirs}

The previous section established the existence of deformations 
of three infinite-dimensional symmetry algebras appearing 
in gravity into finite-dimensional, semisimple Lie algebras.  
The quantum theory, however, contains information
beyond that  in the deformed Lie algebra.  In particular,
the generators of the deformed symmetry are operators on a Hilbert
space, and while the Lie algebra determines the commutators of these
operators, the quantum theory depends on the
full associative product of the operators, i.e., on anticommutators
as well as commutators.\footnote{Consider, for example,
the spin-$\frac12$ and spin-$1$ representations of $SO(3)$.  
In the former, the anticommutator of two different 
Pauli matrices is zero, while in the latter the anticommutator of 
two orthogonal 
$\mfk{so}(3)$ generators is a nonzero
symmetric matrix with zeros on the 
diagonal.}  The structure of the full operator product depends on the representation of the deformed algebra 
in which the quantum theory is defined.  Hence, in order to understand 
the quantization of the gravitational phase spaces admitting 
actions of these algebras, we need a means for determining 
the appropriate representation of the deformed algebra.

As mentioned in section
\ref{sec:fxns2ops}, 
the representation is constrained by matching to the classical
algebra of functions on the gravitational phase space.  
In the limit $\hbar\rightarrow 0$, the symmetric product of 
anticommutators of generators of the algebra is required to 
reproduce  the abelian, associative product of the corresponding
functions on the phase space.  This matching was already discussed
in the simplest example of the fuzzy sphere at the beginning 
of section  \ref{sec:sdiffsun}.  In that case, the symmetric
product of the fuzzy spherical harmonics $\wh{Y}_\alpha$ was 
given by $\wh{Y}_\alpha\circ \wh{Y}_\beta =
\wh{E}\indices{_\alpha_\beta^\gamma}\wh{Y}_\gamma$, with the normalization
of $\wh{Y}_\alpha$ chosen so that $\wh{E}\indices{_\alpha_\beta^\gamma}$
approaches the expression for the classical structure constants
$E\indices{_\alpha_\beta^\gamma}$ in the limit
$N\rightarrow\infty$, as indicated in equation
(\ref{ElargeN}).  This equation in fact determines the representation
of $\mfk{u}(N)$ associated with the quantization of the sphere due 
to the observation that the classical spherical harmonics $Y_\alpha$
form a complete basis for the algebra of functions on the sphere, and 
hence their quantization should also share this property, namely, that
the operator product of two generators closes on the space of generators.
The only representation of $\mfk{u}(N)$ possessing this property
is the defining representation in terms of $N\times N$ matrices, 
leading to the conclusion that this is the appropriate
representation appearing in the quantization of the sphere.  Furthermore,
as discussed around equation (\ref{eqn:Nfs}), consistently
matching the commutators of the generators to the Poisson  
bracket fixes the deformation parameter 
$N_\text{fs}$ to be $\frac{A}{2\pi\hbar_\text{fs}}$,
where $A$ is the area of the phase space computed from the symplectic
form.  Hence, in this case, we see that the quantized algebra is 
fully determined by matching the classical limits of the 
symmetric and antisymmetric products of operators.  

A subtlety arises when applying this reasoning to the gravitational
phase space, since the classical generators of the algebra 
are far from forming a complete basis for functions on the phase space.
Generic products of generators become complicated multilocal integrals 
over the 2-sphere in spacetime, all of which represent
independent functions on the phase space.  This suggests that the 
representation yielding the quantization of the phase space will be 
large, in the sense of containing many operators beyond those 
corresponding to the Lie algebra generators.  These additional
operators would then be assigned to the multilocal 
observables of the classical theory.  Determining a representation
from properties of these multilocal observables  appears daunting;
however, the task is drastically simplified by focusing on invariant
functionals of the classical symmetry algebra, which are associated 
with Casimir operators in the quantum theory.  These invariant
functionals arise from the pullback via the moment 
map of Casimir functions on the classical coadjoint orbits, and we will
find that they reduce, nontrivially, to expressions involving single
integrals over the $2$-sphere in spacetime.  Each such function is 
shown to coincide uniquely with a Casimir element of the deformed algebra,
which are represented as matrices proportional to the identity in 
an irreducible representation.  The c-number proportionality constants
largely determine the representation, and hence by matching these 
c-numbers to the values of the corresponding classical phase space
functions, we arrive at a procedure for determining the representation
associated with the quantization of the phase space.\footnote{This 
argument assumes that the representation is irreducible, and 
requires one to consider a subspace of the classical phase space
defined by fixing the value of the Casimir functions.  More generally,
we expect the full phase space to be foliated by several such
subspaces, which suggests the full quantum theory will occur
in a reducible representation, with each irreducible component
coinciding, roughly, with a single leaf of the foliation in 
the classical phase space.  Determining the multiplicity of the 
representations occurring in this quantization appears to be more 
challenging.  One needs either a natural measure on the space of 
Casimir functions, possibly arising from the phase space symplectic
form itself, or otherwise to find a larger symmetry group that 
acts transitively on the phase space, whose irreducible
representations will occur as reducible representations 
of the smaller algebras considered here.}  
This matching procedure should also determine the value of the 
deformation parameter $N$.

Given their importance for determining the representation
of the deformed algebra, in this section, we characterize the Casimir
elements of each of the deformed algebras, as well as the Casimir 
functions on the coadjoint orbits of the classical algebras.  Furthermore,
we derive the appropriate correspondence between classical 
and deformed Casimir elements, which then facilitates the matching 
procedure needed to determine the representation for the quantization
of the gravitational phase space.  In the case of $\sdiff$, we 
carry out the matching in somewhat more detail to argue that 
the value of $N$ and the associated representation of $\mfk{su}(N)$
are both determined by this procedure.

\subsection{$\sdiff$ and $\su(N)$}\label{sec:sdiffcas}

We begin with the application of the above procedure to 
the algebra $\sdiff$ and its deformation $\mfk{su}(N)$. 
The key step is to classify the invariants of the two 
algebras, and to determine the correspondence between
the invariants in the large-$N$ limit.  

Phase space functions that are invariant under the action of 
$\sdiff$ generically arise as pullbacks of Casimir elements of the 
$\sdiff$ Lie algebra via the moment map.  
We recall that the moment map $\mu$ for a given phase space 
admitting an action of $\sdiff$ sends the phase space $\phsp$ to the
dual of the Lie algebra $\sdiff^*$, which is itself a phase space 
admitting a Hamiltonian action of $\sdiff$ via the coadjoint action
\cite{kirillov2004lectures}.
Hence, a classification of the invariant functions for this 
action on $\sdiff^*$ leads to a corresponding classification of 
invariants on the phase space $\phsp$.  Casimir elements of the 
universal enveloping algebra of $\sdiff$ define functions 
on $\sdiff^*$ via the natural pairing between the Lie algebra and 
its dual, and the fact that the Casimir elements commute with the Lie
algebra translates to the statement that the corresponding 
functions on $\sdiff^*$ are invariant under the coadjoint action of 
$\sdiff$.  This, therefore, gives the link between Casimir elements 
of the Lie algebra and $\sdiff$ invariants in the gravitational phase 
space, thus reducing the problem to determining the Casimir elements 
of $\sdiff$.  

Before doing so, we first describe the space $\sdiff^*$ and the 
coadjoint action in more detail.  
As discussed in section \ref{sec:subalgs}, the Lie algebra $\sdiff$
can be parameterized in terms of stream functions, with each
function $\phi$ on the sphere coinciding with an infinitesimal 
diffeomorphism.  The Lie bracket between two functions
$\phi$ and $\psi$ is defined via the 
Poisson bracket $\{\phi,\psi\}_{\upb}$ associated with the 
unit-radius volume form $\upb$.  
It will be convenient in this section to 
consider the space 
$C^\infty(S)$ of 
all smooth functions on the sphere, including the constant 
function which generates the trivial center of the Poisson algebra
of functions, and hence is equivalent to working with 
the trivially extended algebra $\mfk{g} = \sdiff\oplus \mathbb{R}$.  
The dual Lie algebra $\mfk{g}^*$ can also 
be parameterized by functions on the sphere due to the natural 
pairing provided by integration over the sphere.  Specifically, for 
$\phi \in \mfk{g}$, 
$f\in\mfk{g}^*$, the pairing is given by
\beq\label{eqn:pairing}
\langle f, \phi\rangle = \int_S\nu_0 f\,\phi.
\eeq
Note that because $\sdiff$ is associated with the 
quotient space of all functions modulo constant shifts, the 
natural dual $\sdiff^*$ is given by all functions which integrate
to zero, in order to have a consistent pairing by integrating 
over the sphere.  

Since the Poisson bracket of functions defines the Lie algebra 
on $\mfk{g}$, the adjoint action is given in terms of 
this bracket: $\ad_\phi \psi = \{\phi,\psi\}_{\upb}$. 
Throughout this section, we will always employ the Poisson bracket
$\{\cdot,\cdot\}_{\upb}$, and hence will drop the $\upb$ subscript.
The coadjoint 
action $\ad_\phi^*$ on $\mfk{g}^*$ is defined 
by 
\beq
\langle \ad^*_\phi f, \psi\rangle = -\langle f,\ad_\phi \psi\rangle.
\eeq
Applying the definition (\ref{eqn:pairing}) of the pairing, this 
implies that 
\beq
\langle \ad^*_\phi f,\psi\rangle = -\int_S \nu_0 f\{\phi,\psi\}
 = \int_S \nu_0 \{\phi,f\}\psi = \langle \{\phi,f\},\psi\rangle.
\eeq
We, therefore, conclude that coadjoint action is given by 
\beq
\ad_\phi^* f = \{\phi,f\},
\eeq
and hence agrees with the adjoint action when both spaces 
$\mfk{g}$ and $\mfk{g}^*$
are realized as $C^\infty(S)$.

The coadjoint-invariant functions on
$\mfk{g}^*$ are expressible in terms of the 
Casimir elements of the $\mfk{g}$ universal enveloping
algebra.  The latter can be constructed as follows.  We begin by 
noting that the identity map $\Id$ on $\mfk{g}$ is an element of 
$\mfk{g}^*\otimes\mfk{g}$, and using the isomorphism between $\mfk{g}^*$
and $C^\infty(S)$, we see that it is naturally associated with a
Lie-algebra valued function 
$\jfxn\in C^\infty(S)\otimes \mfk{g}$.  
The fact that $\jfxn$ arises from the identity map implies the relation
\beq \label{eqn:fj}
\langle f,\jfxn\rangle = f,
\eeq
where the pairing is taken between the $\mfk{g}$ tensor factor
of $\jfxn$ and $f\in \mfk{g}^*$, and the output is the function on 
$S$ corresponding to $f$.  Similarly, we have that 
\beq\label{eqn:intj}
\int_S \nu_0\, \jfxn\,\phi = \phi,
\eeq
where again the output on the right-hand side is the element $\phi$
of $\mfk{g}$ associated with the function $\phi$ on the sphere.  This 
latter relation implies that, if we instead use the isomorphism
between $\mfk{g}$ and $C^\infty(S)$, we can  view $\jfxn$ as a bilocal
function, $\jfxn\in C^\infty(S)\otimes C^\infty(S)$, the relation
(\ref{eqn:intj}) implies that $\jfxn(\sigma, \sigma') 
= \delta(\sigma-\sigma')$.\footnote{Since this is a distribution
rather than a function, $\jfxn$ actually lies in a larger space than
$C^\infty(S)\otimes C^\infty(S)$ that includes distributions, but 
this technicality does not affect the arguments of this section.  }
Additionally, it allows us to obtain a basis $\ja_\alpha$
of $\mfk{g}$ from the spherical harmonics $Y_\alpha$,
\beq \label{eqn:ja}
\ja_\alpha = \int_S \nu_0 \, \jfxn Y_\alpha,
\eeq
which, conversely, leads to a mode decomposition of the 
function $\jfxn(\sigma)$,
\beq \label{eqn:jmodes}
\jfxn(\sigma) = \sum_{\alpha} \ja_\alpha Y^\alpha(\sigma).
\eeq

Equation (\ref{eqn:ja}) states that the Lie algebra element associated
with the spherical harmonic $Y_\alpha$ is given by $\jfxn_\alpha$.  It
further implies that the Lie bracket with $\phi\in\mfk{g}$ should 
be determined by the Poisson bracket between $\phi$ and the spherical
harmonic $Y_\alpha$; specifically,
\beq
[\phi, \ja_\alpha] = \int_S \nu_0\,\jfxn\{\phi,Y_\alpha\}.
\eeq
This relation can then be used to determine the Lie
bracket of $\phi$ with the $\mfk{g}$ factor of $\jfxn$,
(see appendix \ref{app:casimirs})
\beq \label{eqn:phij}
[\phi, \jfxn] = -\{\phi,\jfxn\},
\eeq
where the Poisson bracket is evaluated on the $C^\infty(S)$ factor 
of $\jfxn$. 

The Casimir elements of $\mfk{g}$ can then be obtained
straightforwardly by taking products of $\jfxn$ 
with itself.  The expression
$\jfxn^n$ is interpreted as a $\mfk{g}^{\otimes n}$-valued function
on the sphere, with the product taken within the $C^\infty(S)$ factor
of each $\jfxn$.\footnote{If we instead identify
each factor of $\mfk{g}$ with a function on the 
sphere, the expression $\jfxn^n$ would be interpreted 
as an $(n+1)$-local function on the sphere consisting 
of products of delta functions, i.e.\ 
$\jfxn^n(\sigma_1,\ldots\sigma_n;\sigma) = 
\displaystyle\prod_{i=1}^n \delta(\sigma_i-\sigma)$.  However, 
in matching to the Casimirs of the deformed algebra, 
it is more convenient to use the 
abstract Lie algebra as opposed to the representation 
in terms of functions on the sphere.}    
To arrive at the Casimir elements,
we integrate this object over the sphere,
\beq
c_n = \int_S \nu_0\, \jfxn^n.
\eeq
Verifying that $c_n$ commutes with every element of $\mfk{g}$ comes 
from a straightforward application of (\ref{eqn:phij}):
\begin{equation}
    \begin{aligned}
[\phi, c_n]  
&= \int_S \nu_0 \sum_{k=1}^n \jfxn^{k-1}[\phi,\jfxn] \jfxn^{n-k}
\\
&= -\int_S\nu_0\sum_{k=1}^n \jfxn^{k-1}\{\phi,\jfxn\} \jfxn^{n-k}
\\
&= - \int_S\nu_0 \{\phi, \jfxn^n\}
=0,
    \end{aligned}
\end{equation}
since any Poisson bracket integrates to zero over the sphere.  

When matching to the Casimirs of the deformed algebra $\mfk{u}(N)$,
it is useful to have an expression of $c_n$ in a specific 
basis.  This can be obtained immediately from the mode
decomposition (\ref{eqn:jmodes}) of $\jfxn$,
\beq \label{eqn:cn}
c_n = \ja_{\alpha_1} \ja_{\alpha_2}\ldots \ja_{\alpha_n}
\int_S \nu_0 Y^{\alpha_1}\ldots Y^{\alpha_n}
\equiv \ja_{\alpha_1}\ldots\ja_{\alpha_n} d^{\alpha_1\ldots\alpha_n},
\eeq
where the second equation defines the totally symmetric
tensor $d^{\alpha_1\ldots \alpha_n}$.
We will later see that this tensor matches a corresponding tensor
$\wh{d}^{\alpha_1\ldots\alpha_n}$ defined for $\mfk{u}(N)$ as 
$N\rightarrow\infty$.

The  functions on $\mfk{g}^*$ associated with the Casimir
elements $c_n$ are obtained by the natural pairing between $\mfk{g}$ and
$\mfk{g}^*$.  For $f\in\mfk{g}^*$, we have that 
\beq
c_n[f] = \int_S \nu_0 \langle\jfxn,f\rangle^n = \int_S \nu_0 \, f^n,
\eeq
where we have applied the relation (\ref{eqn:fj}).  
In the current context in which our Lie algebra involves $\sdiff$,
these Casimir functions are called enstrophies, due to a close 
analogy between the $\sdiff$ coadjoint orbits and 2D fluid dynamics
on the sphere 
\cite{DonnellyFreidelMoosavianSperanza202012, arnold1999topological, izosimov2016coadjoint}.
These Casimir functions can be pulled back to the gravitational 
phase space via the moment map $\mu$.  This pullback is readily
obtained from the relation
$\mu^*\jfxn = J$, where $J$ is the function on $S$ defined in section 
\ref{sec:PJ}.  The result of the pullback of the Casimir functions is 
a set of $\sdiff$ invariants on the gravitational phase space, coinciding
with the gravitational enstrophies discussed in
\cite{DonnellyFreidelMoosavianSperanza202012}.  These invariants 
are given explicitly by
\beq\label{eqn:Cn}
C_n = \int_S \nu_0 J^n.  
\eeq

Note that the invariants can equivalently be expressed 
by pulling back the mode decomposition (\ref{eqn:cn}) to the 
gravitational phase space.  Since $\ja_\alpha$ pulls back to
the generator $J_\alpha$ on the phase space, we see that  
$C_n$ is equivalently expressed as 
\beq\label{eqn:Cnmodes}
C_n = J_{\alpha_1}\ldots J_{\alpha_n} d^{\alpha_1\ldots\alpha_n}.
\eeq
Since each $J_\alpha$ is given by an integral over $S$,
the expression (\ref{eqn:Cnmodes}) naively appears to be a 
complicated object involving multiple integrals over the 
sphere.  The fact that it localizes to a single integral
as in (\ref{eqn:Cn}) comes from special properties of the 
tensor $d^{\alpha_1\ldots\alpha_n}$, which produces delta functions
when contracted into the $J_\alpha$
in the expression for $C_n$, 
resulting in a single integral expression.\footnote{As an example,
since $J_\alpha = \int_S d\sigma Y_\alpha(\sigma) J(\sigma)$
and $d^{\alpha\beta} = \int_S d\sigma Y^\alpha(\sigma)Y^\beta(\sigma)$,
evaluating (\ref{eqn:Cnmodes}) for $C_2$, we get
\begin{equation*}
    \begin{aligned}
        C_2 
&= 
\int d\sigma_1 \int d\sigma_2\int d\sigma_3
J(\sigma_1)J(\sigma_2)
Y_\alpha(\sigma_1)Y_\beta(\sigma_2) Y^\alpha(\sigma_3)Y^\beta(\sigma_3)
\nonumber\\
&=
\int d\sigma_1 \int d\sigma_2\int d\sigma_3
J(\sigma_1)J(\sigma_2)
\delta(\sigma_1-\sigma_3)\delta(\sigma_2-\sigma_3)
\nonumber \\
&=
\int d\sigma_1 J^2(\sigma_1).
    \end{aligned}
\end{equation*}
Computations for the higher Casimirs $C_n$ 
show that delta functions appear in a similar manner,
always leading to a single integral expression.
}  

The Casimir elements for the deformed algebra 
$\wh{\mfk{g}} = \mfk{u}(N)$
can be obtained in an analogous manner.  We 
now define an object 
$\defj\in \Mat_{N\times N}\otimes \wh{\mfk{g}}$ normalized
so that
\beq
\frac{1}{N}\tr_{\bf{N}}\left(\defj \cdot \wh{Y}_\alpha\right) 
= \algX_\alpha,
\eeq
where $\algX_\alpha$ are the basis elements 
for the $\mfk{u}(N)$
Lie algebra introduced in section \ref{sec:sdiffsun}.
The mode decompositions of the $\defj$ matrix elements 
are therefore given by
\beq\label{eqn:defjmodes}
\defj\indices{_i^j} = \sum_{\alpha} X_\alpha (\wh{Y}^\alpha)\indices{_i^j}.
\eeq
Similar to the classical relation
(\ref{eqn:phij}), the Lie bracket between
$\defj$ and a Lie algebra element $\algX_\alpha$ can be
expressed as (see appendix \ref{app:casimirs})
\beq \label{eqn:Xaj}
[\algX_\alpha, \defj]_{\wh{\mfk{g}}} =
-\frac{N}{2i}[\wh{Y}_\alpha, \defj],
\eeq
where the bracket on the right hand side is the matrix 
commutator evaluated on the $\Mat_{N\times N}$ factor
of $\defj$.

Invariant elements of the $\wh{\mfk{g}}$ tensor algebra
arise from products of $\defj$ with itself, $\defj^n$, 
where the product is taken within the $\Mat_{N\times N}$ 
factor of $\defj$ and the resulting matrix is valued 
in  $\wh{\mfk{g}}^{\otimes n}$, which
therefore defines an element of the universal enveloping
algebra.  Taking a trace over the 
matrix factor yields the Casimir element,
\beq
\wh{c}_n = \frac1N\tr_{\bf{N}} \defj^n,
\eeq
which can be shown to commute with $\wh{\mfk{g}}$ using 
(\ref{eqn:Xaj}):
\begin{equation}
    \begin{aligned}
[\algX_\alpha, \wh{c}_n]_{\wh{\mfk{g}}}
&=\frac1N\sum_{k=1}^n \tr\left(\defj^{k-1} [\algX_\alpha, \defj\,]_{\wh{\mfk{g}}}\,
\defj^{n-k} \right) \\
&=-\frac{N}{2i}\frac{1}{N}\sum_{k=1}^n\tr\left(
\defj^{k-1}[\wh{Y}_\alpha \defj\,]\,\defj^{n-k}\right)\\
&= -\frac{1}{2i} \tr\left([\wh{Y}_\alpha, \defj^n]\right) = 0.
\label{eqn:Xhatc}
\end{aligned}
\end{equation}
This can also be expressed in the $\algX_\alpha$ basis 
for $\wh{\mfk{g}}$ by applying the mode decomposition
(\ref{eqn:defjmodes})
\beq \label{eqn:hatcn}
\wh{c}_n = \algX_{\alpha_1}\ldots\algX_{\alpha_n}
\frac{1}{N}\tr\left(\wh{Y}^{\alpha_1}\ldots\wh{Y}^{\alpha_n}
\right)
\equiv \algX_{\alpha_1}\ldots\algX_{\alpha_n}
\wh{d}^{\alpha_1\ldots \alpha_n},
\eeq
where the second equality defines the coefficients
$\wh{d}^{\alpha_1\ldots \alpha_n}$.  

As demonstrated in
appendix \ref{app:casimirs}, in the large-$N$ limit the deformed 
Casimir coefficients $\wh{d}^{\alpha_1\ldots \alpha_n}$
approach the classical Casimir coefficients $d^{\alpha_1\ldots
\alpha_n}$ for the Lie algebra $\sdiff\oplus\mathbb{R}$,
\beq \label{eqn:dhatd}
\wh{d}^{\alpha_1\ldots \alpha_n} =  d^{\alpha_1\ldots\alpha_n}
+\mathcal{O}(N^{-1}).
\eeq
In this sense, the Casimir elements of the deformed 
algebra approach those of the classical algebra in the 
large-$N$ limit.\footnote{This agreement between the Casimir 
elements $c_n$ and $\wh{c}_n$ requires that $n$ is held fixed
as $N\rightarrow \infty$.}  In particular, it implies that in the 
representation $\pi_{\phsp}$ associated with the gravitational
phase space, the quantization of the deformed Casimir 
elements
\beq \label{eqn:hatCn}
\wh{C}_n = \wh{J}_{\alpha_1} \ldots \wh{J}_{\alpha_n}\wh{d}^{\alpha_1\ldots\alpha_n}
 = (i\hbar)^n \pi_{\phsp}(\wh{c}_n),
\eeq
must match the value of the classical invariants $C_n$,
given by (\ref{eqn:Cn}), up to $\mathcal{O}(\hbar^2)$ and 
$\mathcal{O}(N^{-2})$ corrections.  The prefactor of 
$(i\hbar)^n$ appears due to the normalization
condition $\wh{J}_\alpha = i\hbar\pi_{\phsp}(\algX_\alpha)$.

\subsection{Matching Casimirs}\label{sec:casmatch}

Having determined the correspondence between the Casimir operators 
$\wh{C}_n$ and the classical gravitational invariants $C_n$, we 
next show that this correspondence can be used to determine the 
deformation parameter $N$ and the appropriate 
representation of $\mfk{u}(N)$ associated with 
the quantization of the gravitational phase space.  This matching
makes use of the explicit characterization of large-$N$ representations
of $\mfk{u}(N)$ and the associated Casimir operators that 
has been developed in previous investigations on 
matrix models (see, e.g.\, \cite{Cordes:1994fc}).

In order to take advantage of these results,
we first need to express the Casimir elements $\wh{c}_n$
given in (\ref{eqn:hatcn}) in terms of the standard expressions
for the Casimirs in the elementary matrix basis $E\indices{^i_j}$for $\mfk{u}(N)$, described in (\ref{eqn:XaEij}) and (\ref{eqn:EijEkl}).  
Using the identity satisfied by the fuzzy spherical harmonics
(derived in appendix \ref{app:SUNrelns})
\beq
\delta^{\alpha\beta}\big(\wh{Y}_\alpha\big)\indices{_i^j}
\big(\wh{Y}_\beta\big)\indices{_k^l} = N \delta^l_i\delta^j_k,
\eeq
we find the expression for $\wh{c}_k$ in the $E\indices{^i_j}$ basis,
\begin{equation}
    \begin{aligned}
\wh{c}_k &=
\left(\frac{N}{2i}\right)^k\frac{1}{N}
\big(\wh{Y}^{\alpha_1}\big)\indices{_{i_1}^{i_2}}
\big(\wh{Y}_{\alpha_1}\big)\indices{_{m_1}^{n_1}}
E\indices{^{m_1}_{n_1}} \ldots
\big(\wh{Y}^{\alpha_k}\big)\indices{_{i_k}^{i_1}}
\big(\wh{Y}_{\alpha_k}\big)\indices{_{m_k}^{n_k}}
E\indices{^{m_k}_{n_k}} 
\\
&=
\frac{N^{2k-1}}{(2i)^k}E\indices{^{i_2}_{i_1}}E\indices{^{i_3}_{i_2}}
\ldots E\indices{^{i_1}_{i_k}}. \label{eqn:hatctildec}
\end{aligned}
\end{equation}
Up to a permutation of the order of the $E\indices{^i_j}$ 
generators,\footnote{This reordering will affect a detailed matching 
for the Casimirs including subleading corrections in $\frac{1}{N}$, but
should not affect the large-$N$ scaling derived in this section.}
this shows that the Casimir elements $\wh{c}_k$ are rescaled by a 
factor of $\frac{N^{2k-1} }{(2i)^k}$ relative to the standard 
$\mfk{u}(N)$ Casimirs
\beq
\cstd_k = E\indices{^{i_1}_{i_2}}E\indices{^{i_2}_{i_3}} \ldots
E\indices{^{i_k}_{i_1}}.
\eeq

In a given irreducible representation $R$ of $\mfk{u}(N)$, the Casimir
element $\cstd_k$ is given by a number $\cstd(R)$ times the 
identity.  This number is matched to the corresponding invariant 
functional on the gravitational phase space in order to determine 
the representation $R$ and deformation parameter $N$.  Noting
the rescaling by $(i\hbar)^k$ implied by equation (\ref{eqn:hatCn})
and the additional prefactor in (\ref{eqn:hatctildec})
relating $\wh{c}_k$ and $\cstd_k$, the matching between 
$\wh{C}_k$ and $C_k$ implies that 
\beq\label{eqn:Ckmatch}
C_k = \left(\frac{\hbar N}{2}\right)^k N^{k-1} \cstd_k(R),
\eeq
showing that the gravitational enstrophies $C_k$ directly determine
the values of the Casimirs $\cstd(R)$ in the representation
associated with the quantization of the phase space.  

While a detailed determination of the representation from these
matching relations depends on the precise values of the enstrophies
$C_k$, we can use generic properties of the Casimirs at large $N$ 
to determine a scaling relation for the deformation parameter
$N$.  The relation is that \cite{Cordes:1994fc}
\beq \label{eqn:ckRscaling}
\cstd_k(R) \sim N^{k-1}n,
\eeq
where $n$ denotes the number of boxes in the Young diagram for the 
representation $R$, with the precise
coefficient and subleading corrections depending on the 
shape of the Young diagram. 

To arrive at the desired relation for $N$, we would like to 
determine how the gravitational enstrophies $C_k$ scale with the 
area of the surface $S$.  This requires relating the 
normalization conventions for the $J_\alpha$ generators 
given in section \ref{sec:modes} to the convention employed in
reference \cite{DonnellyFreidelMoosavianSperanza202012}.  
Consistently relating the normalization conventions
(see appendix \ref{app:casimirs}) leads to the 
relation 
\beq \label{eqn:CkA}
C_k = \left(\frac{A}{16\pi G}\right)^k
\int_S \nu_0\left(\frac{-A W}{4\pi}\right)^k,
\eeq
where $A$ is the area of the surface, and $W$ is the outer curvature
scalar associated with curvature of the normal bundle of $S$
\cite{DonnellyFreidelMoosavianSperanza202012, Carter1992}. 
Since the quantity $A W$ is a dimensionless, order $1$ function on the 
sphere, we see that the integral in (\ref{eqn:CkA}) only contributes
an order $1$ coefficient to each $C_k$. 
Hence,
the scaling relation for $C_k$ with area is
\beq
C_k\sim\left(\frac{A}{16\pi G}\right)^k.
\eeq
Together with the matching equation (\ref{eqn:Ckmatch}) and the large $N$
scaling of the Casimirs (\ref{eqn:ckRscaling}), this implies that 
\beq
\left(\frac{A}{16\pi G}\right)^k\sim\left(\frac{\hbar N^3}{2}\right)^k
\frac{n}{N^2}.
\eeq
In order to satisfy this scaling for all values of $k$, it must be
that the number of boxes $n$ in the Young diagram of the representation
scales like $N^2$, and further that $N$ scales as
\beq \label{eqn:NsimA}
N\sim \left(\frac{A}{8\pi G\hbar}\right)^{\frac13}.
\eeq
Since $\frac{A}{8\pi G\hbar}$ is associated with the entropy of the 
codimension-2 surface $S$, we find that this relation says that 
$S \sim N^3$.  This stands in contrast with standard holographic
examples, where typically the entropy of a black hole scales with 
$N^2$.\footnote{However, this scaling is far from universal. A counterexample is provided by ABJM theory, which is dual to quantum gravity in AdS${}_4$ \cite{Aharony:2008ug} and for which the entropy
scales as $S= N^{\frac32}$.  There are also examples 
of brane configurations in string theory with triple intersections
in which the number of states can scale as $N^3$
\cite{Berenstein:1998rr}, 
reminiscent of the scaling found here.  
\label{ftn:Nscale}}
However, we note that this computation should be taken with a 
grain of salt, since we are only analyzing the $\sdiff$ symmetry algebra,
which is a subalgebra of the full gravitational symmetry algebra.
In particular, the $\sdiff$ subalgebra does not include boosts,
whose Noether charge in gravity is typically associated with 
the entropy of black holes.  Hence, although the 
calculations of this section give a proof of principle for how
the Casimir matching should work, we should not immediately
draw any conclusions from these computations in relation to entropy
in gravitational applications.  Instead, we should look to complete 
the matching conditions in the extended algebras $\crr$ or $\cslr$, 
or even the full gravitational algebra $\gslr$, which may yield 
a more sensible relation between the entropy and deformation parameter
$N$.  

The results of this section has demonstrated how the Casimir matching 
can be done in principle to determine the representation; however,
it would be interesting to carry out this matching in more detail.
Doing so would yield a precise relation between the 
entropy and deformation parameter $N$.  Furthermore, we should 
expect to be able to relate the function $W$ on the sphere to the 
shape of the Young diagram of the representation in the 
large $N$ limit.  We leave this more detailed matching as an
interesting direction for future work.

\subsection{$\crr$ and $\mfk{sl}(N,\mathbb{C})\oplus \mathbb{R}$}\label{sec:cRlim}

We now turn to the first extended algebra appearing in the gravitational
phase space, $\crr$, which was shown in section \ref{sec:sdiffslnc}
to arise as a large $N$ limit of the finite-dimensional
algebra $\mfk{sl}(N,\mathbb{C})\oplus\mathbb{R}$.
Following the same procedure as in section \ref{sec:sdiffcas}, we 
begin by describing the coadjoint orbits of the classical
algebra $\crr$ and use them to determine the Casimir elements of the 
algebra.  We then identify the Casimir elements of the deformed algebra,
and determine the appropriate matching condition between these 
Casimirs and their classical analogs.  

As in section \ref{sec:sdiffcas}, it is convenient to work with
the algebra $\mfk{g} = \crr\oplus \mathbb{R}$, where the additional 
central generator corresponds to a constant function on the sphere.
The Lie algebra $\mfk{g}$ is then parameterized by a pair of 
functions $(\phi,\alpha)$ on the sphere, and  the 
dual of the Lie algebra is similarly parameterized by a pair
of functions $(f,a)$.  The pairing is given by the integral
over the sphere,
\beq \label{eqn:crrpairing}
\langle(f,a),(\phi,\alpha)\rangle = 
\int_S \nu_0 \left(f\phi + a \alpha\right).
\eeq
The adjoint action of the Lie algebra on itself is 
given by 
$\ad_{(\phi,\alpha)} (\psi,\beta) =
(\{\phi,\psi\},\{\phi,\beta\}-\{\psi,\alpha\})$, which,
along with the pairing (\ref{eqn:crrpairing}) determines the coadjoint
action to be (see appendix \ref{app:casimirs})
\beq \label{eqn:ad*pha}
\ad^*_{(\phi,\alpha)}(f,a) = (\{\phi,f\} + \{\alpha,a\},\{\phi,a\}).
\eeq
This equation indicates that both $f$ and $a$ transform as scalars 
under $\sdiff$ transformations, but $f$ has a nontrivial transformation
law under the $\mathbb{R}^S$ subalgebra, which leads to some subtleties
in obtaining a full set of Casimir invariants.  

To construct the Casimir elements of $\mfk{g}$, it is convenient
to introduce a pair of $\mfk{g}$-valued functions on the sphere
$(\jfxn, \nfxn)$ in analogy with the construction of section
\ref{sec:sdiffcas}, satisfying
\begin{align}
\langle (f,a), \jfxn\rangle &= f, \\
\langle (f,a), \nfxn\rangle &= a,
\end{align}
where the left-hand side evaluates the pairing between $(f,a)$ and 
the $\mfk{g}$ factors of $\jfxn$ and $\nfxn$, and the right hand 
side returns the functions on the sphere associated with $f$ and $a$.
These pairing relations can then be used to determine the Lie bracket
between an element $(\phi,\alpha)$ of $\mfk{g}$ and the $\mfk{g}$
factors of $\jfxn$ and $\nfxn$ (see appendix \ref{app:casimirs}).  
The result is 
\begin{align}
[(\phi,\alpha),\jfxn] &= -\{\phi,\jfxn\} - \{\alpha,\nfxn\} 
\label{eqn:phialphajfxn},\\
[(\phi,\alpha),\nfxn] &= -\{\phi,\nfxn\}.
\label{eqn:phialphanfxn}
\end{align}

The Casimir elements are now obtained by examining products of the form
$\jfxn^m \nfxn^n$, interpreted as a $\mfk{g}^{\otimes(m+n)}$-valued 
function on the sphere, again with the product taken within the 
$C^\infty(S)$ factor of each $\jfxn, \nfxn$.  Taking integrals
of these over the sphere gives a set of candidate Casimir elements,
\beq
c_{mn} = \int_S \nu_0 \jfxn^m \nfxn^n.
\eeq
From the relations (\ref{eqn:phialphajfxn}) and (\ref{eqn:phialphanfxn}),
one can verify that $c_{mn}$ commute with all $\sdiff$ generators 
in $\mfk{g}$:
\beq
[(\phi,0),c_{mn}] = -\int_S \nu_0 \{\phi, \jfxn^m \nfxn^n\} = 0.  
\eeq
However, there is an additional constraint coming from
demanding invariance with respect to the $\mathbb{R}^S$ 
generators $(0,\alpha)$:\footnote{These steps require that we 
employ the identities $\jfxn \nfxn = \nfxn \jfxn$, $\{\nfxn,\nfxn\} = 0$,
and $\{\nfxn,\jfxn\}\propto\nfxn$, none of which are 
immediately obvious due to $\jfxn$ and $\nfxn$ being Lie algebra valued.
These identities are derived in appendix \ref{app:casimirs}.}
\begin{equation}
    \begin{aligned}
[(0,\alpha),c_{mn}] &= -m \int_S \nu_0\,\jfxn^{m-1}\{\alpha,\nfxn\}
\nfxn^n
\\
&=m(m-1)\int_S \nu_0\,\jfxn^{m-2}\nfxn^n \{\jfxn,\nfxn\} \alpha.
\label{eqn:alphacmn}
\end{aligned}
\end{equation}
Since $\{\jfxn,\nfxn\}\neq 0$, this quantity will vanish for all 
choices of the function $\alpha$ only if $m=0$ or $m=1$.  Hence,
these define two sets of Casimir elements for $\mfk{g}$, 
\beq \label{eqn:c0nc1n}
c_{0n} = \int_S \nu_0\, \nfxn^n,\qquad c_{1n} = \int_S\nu_0\,\jfxn
\nfxn^n, \qquad n = 1,2,\ldots.
\eeq
The associated functions on $\mfk{g}^*$ that are invariant
under the coadjoint action are given by
\beq
c_{0n}[(f,a)] = \int_S\nu_0\, a^n,
\qquad c_{1n}[(f,a)] = \int_S\nu_0\, f a^n.
\eeq
These pull back to invariant functions on the gravitational phase space
is given by
\beq
C_{0n} = \int_S \nu_0\, N^n ,\qquad C_{1n}=\int_{S}\nu_0\,J N^{n}.
\eeq

We now would like to relate the classical Casimirs to the 
Casimir elements of the deformed algebra, which we take to
be $\wh{\mfk{g}} = \mfk{gl}(N,\mathbb{C})$, which is the 
appropriate algebra to limit to the classical algebra
$\crr\oplus\mathbb{R}$. To identify the Casimirs of $\wh{\mfk{g}}$,
we can proceed analogously to section \ref{sec:sdiffcas}
and define $\defj,\defn\in\Mat_{N\times N}\otimes \wh{\mfk{g}}$,
normalized such that
\begin{align}
\frac{1}{N}\tr_{\bf{N}}\big(\,\defj\cdot \wh{Y}_\alpha\big) &= \algX_\alpha,
\\
\frac{1}{N}\tr_{\bf{N}}\big(\,\defn\cdot\wh{Y}_\alpha \big) &= \algZ_\alpha,
\end{align}
where $\algX_\alpha, \algZ_\alpha$ are the generators of the 
$\mfk{gl}(N,\mathbb{C})$ algebra defined in section
\ref{sec:sdiffslnc}.  This implies the following mode
decomposition of the matrix elements of $\defj, \defn$,
\begin{align}
\defj\indices{_i^j}
&= \sum_\alpha \algX_\alpha (\wh{Y}^\alpha)\indices{_i^j}, 
\\
\defn\indices{_i^j}
&= \sum_\alpha \algZ_\alpha (\wh{Y}^\alpha)\indices{_i^j}.
\end{align}

The Casimir elements are most easily identified by forming complex
combinations of $\defj$ and $\defn$.  These arise naturally by
noting that the complexified generators 
$E_\alpha^\pm = \frac12\left(\algX_\alpha\pm i N\algZ_\alpha\right)$
satisfy
\begin{align}
[E_\alpha^\pm, E_\beta^\pm] 
&= \wh{C}\indices{_\alpha_\beta^\gamma}E_\gamma^\pm, 
\\
[E_\alpha^+, E_\beta^-]&=0.
\end{align}
The associated Lie-algebra-valued matrices $\defe^\pm$ defined by
the condition
\beq
\frac{1}{N}\tr_{\bf{N}}\big(\,\defe^\pm\cdot\wh{Y}_\alpha\big) = E_\alpha^\pm. 
\eeq
are then related to $\defj, \defn$ by
\beq
\defe^\pm = \frac12\left(\defj \pm iN\defn\right).
\eeq
The Lie brackets between $\defe^\pm$ and the Lie algebra elements 
$E_{\alpha}^\pm$ can be shown to satisfy 
\begin{align}
[E^\pm_\alpha, \defe^\pm]_{\wh{\mfk{g}}} &=
-\frac{N}{2i}[\wh{Y}_\alpha,\defe^\pm], \\
[E^\mp_\alpha, \defe^\pm]_{\wh{\mfk{g}}} &= 0.
\end{align}
From this, it follows that two sets of Casimir elements can be 
formed according to 
\begin{align}
\wh{c}_n^\pm = \frac{1}{N}\tr_{\bf{N}}\left[\left(\frac{2\defe^\pm}{N}\right)^n\right].
\end{align}
Demonstrating that $\wh{c}_n^\pm$ commute with $\wh{\mfk{g}}$
proceeds analogously to the computation leading to (\ref{eqn:Xhatc}).  

The scaling with $N$ chosen for the normalization of $\wh{c}_n^\pm$
is needed in order to obtain a good large $N$ limit.  
Expanding out the expression for $\wh{c}_n^\pm$ in terms of $\defj, \defn$,
we find that 
\beq
\wh{c}_n^\pm = \frac{(\pm i)^n}{N}\tr_{\bf{N}}(\defn^n) + 
\frac{(\pm i)^{n-1} n}{N}\tr_{\bf{N}}(\defj\defn^{n-1}) 
+\mathcal{O}(N^{-2}).
\eeq
The appropriate objects to match to the classical Casimir elements 
(\ref{eqn:c0nc1n}) are the linear combinations
\begin{align}
\wh{c}_{0n} &= \frac{1}{2i^n}\left(\wh{c}_n^+ + (-1)^n\wh{c}_n^-\right)
=\frac{1}{N}\tr_{\bf{N}}(\defn^n) + \mathcal{O}(N^{-2}),\\
\wh{c}_{1n} &= \frac{N}{2ni^{n-1}}\left(\wh{c}_n^+-(-1)^n\wh{c}_n^-\right)
= \frac{1}{N}\tr_{\bf{N}}(\defj\defn^{n-1}) + \mathcal{O}(N^{-2}).
\end{align}
Just as in section \ref{sec:sdiffcas}, one can show that the 
deformed Casimirs $\wh{c}_{0n}$, $\wh{c}_{1n}$ approach their classical
counterparts $c_{0n}, c_{1n}$, in the sense that their coefficients 
when expressed in the $(\algX_\alpha, \algZ_\alpha)$ basis approach
the classical coefficients.  Once again, this is a consequence
of the relation (\ref{eqn:dhatd}).  
Furthermore, the corresponding quantum operators 
$\wh{C}_{0n},\wh{C}_{1n}$ obtained in the representation
$\pi_\phsp$ corresponding to the quantization of the phase space
are given by
\beq
\wh{C}_{0n} = (i\hbar)^n\pi_\phsp(\wh{c}_{0n}),\qquad
\wh{C}_{1n} = (i\hbar)^n\pi_{\phsp}(\wh{c}_{1n}),\qquad
\eeq
and these should be matching to the classical invariants $C_{0n}$,
$C_{1n}$ defined on the gravitational phase space.  Since 
$\wh{C}_{0n, 1n}$ are proportional to the identity in an irreducible
representation of $\wh{\mfk{g}}$, they can be matched as c-numbers 
according to 
\beq
C_{0n, 1n} = \wh{C}_{0n,1n} +\mathcal{O}(\hbar^2) + \mathcal{O}(N^{-2}).
\eeq
Just as in section \ref{sec:casmatch}, this matching relation
should determine the representation of $\wh{\mfk{g}}$ as well as 
the value of the deformation parameter $N$.  Carrying out the
matching in detail would require an in-depth enumeration of 
the unitary irreducible representations of $\mfk{gl}(N,\mathbb{C})$,
which is beyond the scope of the present work, but would nevertheless
be a fruitful direction for future investigations.
In carrying out this matching, the results of 
\cite{Vogan1986} are likely relevant.

\subsection{$\cslr$ and $\su(N,N)$}\label{SuNN}

Finally, we consider the largest extended algebra $\cslr$,
which was shown in section \ref{sec:sdiffsunn} 
to appear in the large
$N$ limit of the semisimple, finite-dimensional algebra 
$\mfk{su}(N,N)$.  As in previous sections, we begin the 
analysis by describing the coadjoint orbits of the classical 
algebra $\cslr$, and use these to identify the Casimir elements.
We then show that these Casimirs naturally match onto 
corresponding Casimirs of the deformed algebra, and this
matching condition can once again be used to determine the 
representation appearing in the quantization of the 
classical phase space.  

As before, we work with the trivially extended algebra
$\mfk{g}= \cslr\oplus \mathbb{R}$ for convenience,
which is naturally associated with the large $N$ 
limit of $\mfk{u}(N,N) = \mfk{su}(N,N)\oplus\mathbb{R}$.
The Lie algebra is parameterized by a pair of 
functions on the sphere $(\phi, \alpha^a)$,
with $\phi$ scalar valued and $\alpha^a$ valued in 
$\mfk{sl}(2,\mathbb{R})$, with the index $a = 0,1,2$ 
denoting the components of the function in a basis.  We will
utilize the $\tau_a$ basis for 
$\mfk{sl}(2,\mathbb{R})$ given in equation (\ref{eqn:taua})
in which the structure constants are real.  
The dual lie algebra $\mfk{g}^*$ is similarly
parameterized by a pair of functions $(f,a_a)$, again with $f$
scalar-valued and $a_a$ $\mfk{sl}(2,\mathbb{R})$-valued, 
and the pairing between $\mfk{g}$ and $\mfk{g}^*$ is given
by 
\beq
\langle (f,a_b), (\phi,\alpha^a)\rangle = \int_S \nu_0\big(
f\phi + a_a \alpha^a\big).
\eeq
Given the expression for the adjoint action
of the Lie algebra on itself, $\ad_{(\phi,\alpha^a)}
(\psi,\beta^b) = \big(\{\phi,\psi\},
\{\phi,\beta^b\}-\{\psi,\alpha^b\} +[\alpha,\beta]^b_{\mfk{sl(2,\mathbb{R})}} \big)$,
where $[\alpha,\beta]^c_{\mfk{sl}(2,\mathbb{R})}
=\alpha^a\beta^b\varepsilon\indices{_a_b^c}$,
the coadjoint action is given by
\beq\label{eqn:cslrcoad}
\ad^*_{(\phi,\alpha^a)} (f,a_b)
= \left(\{\phi,f\} + \{\alpha^b, a_b\}, \{\phi, a_b\} 
+ [\alpha,a]^{\mfk{sl}(2,\mathbb{R})}_b \right).
\eeq
Note that this coadjoint action for $\mfk{g}$ 
is closely related to the action for the larger symmetry
group 
$\mfk{diff}(S)\oplus_{\mathcal{L}}\mfk{sl}(2,\mathbb{R})^S$
examined in \cite{DonnellyFreidelMoosavianSperanza202012},
upon replacing Lie derivatives with Poisson brackets.  
The action (\ref{eqn:cslrcoad})
indicates that $f$ and $a_b$ transform as scalars under 
diffeomorphisms of the sphere, and $a_b$ transforms
in the adjoint representation under $\mfk{sl}(2,\mathbb{R})$
transformations.  However, $f$ transforms inhomogeneously
under $\mfk{sl}(2,\mathbb{R})$ transformations, and this 
is the main challenge to deal with when looking for 
invariant functions under the coadjoint action.  

Rather than working with $\mfk{g}$-valued functions 
$(\jfxn, \nfxn_a)$ to construct Casimir elements as 
in previous sections, here it will be more convenient to 
look directly for invariant functions on the orbits,
after which expressions for the Casimir elements can 
be determined.  A first set of invariants is readily obtained
by noting that the $\mfk{sl}(2,\mathbb{R})$ quadratic 
Casimir $a^2 = a_b a^b$ transforms as a scalar function on
the sphere, and hence the moments of this function will be
fully invariant under $\sdiff$ and $\mfk{sl}(2,\mathbb{R})$ 
transformations.  This leads to the first set of 
Casimir functions
\beq \label{eqn:c2n}
c_{2n}[(f,a_b)] = \int_S \nu_0 \left(a^2\right)^n.
\eeq
These invariants are  the analogs of the $c_{0n}$
Casimirs of the algebra $\crr$ defined in (\ref{eqn:c0nc1n}),
since both are independent of $f$.  Note that these 
Casimirs  have no analog in the larger algebra 
$\diff\oplus_{\mathcal{L}}\mfk{sl}(2,\mathbb{R})^S$
examined in \cite{DonnellyFreidelMoosavianSperanza202012},
since in that case, there is no fixed volume form
$\nu_0$, and hence the only natural volume form on the 
sphere comes from the $\mfk{sl}(2,\mathbb{R})$ quadratic
Casimir itself.  Because of this, there is no meaningful
way to construct moments of the quadratic Casimir 
when working with the larger algebra, since in that 
case it transforms as a density as opposed to a scalar.  

On the other hand, we should expect a 
second set of Casimirs that are the analogs 
of the $c_{1n}$ Casimirs of $\crr$ in equation
(\ref{eqn:c0nc1n}).  Additionally, the 
construction of Casimir functions for the larger
symmetry algebra $\diff\oplus_{\mathcal{L}}\mfk{sl}(2,\mathbb{R})^S$ 
in \cite{DonnellyFreidelMoosavianSperanza202012} 
lead to 
a set of generalized enstrophies constructed from moments 
of a scalar vorticity $w$ which contains a cubic term in the 
$\mfk{sl}(2,\mathbb{R})$ generators.  Since the vorticity
arose naturally as an object constructed from $f$ 
that is invariant under $\mfk{sl}(2,\mathbb{R})$
transformations, the expectation is that 
a similar object should arise in the classification
of invariants of $\cslr$.  By examining to what
extent such a vorticity can be defined from the 
$\cslr$ orbit data, we will obtain a prescription for 
constructing the second set of Casimirs for this algebra.

In the larger algebra, the vorticity is constructed from
the orbit data, which consists of a densitized 
1-form $\wt p_A$ and an $\mfk{sl}(2,\mathbb{R})$-valued
density $\wt a_b$.  The $\mfk{sl}(2,\mathbb{R})$ quadratic 
Casimir constructed from $\wt{a}_a$ determines a 
dynamical volume form $\nu$, which is related to the fixed 
volume form $\nu_0$ by the relation
\beq \label{eqn:nudyn}
\nu = \nu_0 \sqrt{a^2},
\eeq
where $a_b$ is the $\mfk{sl}(2,\mathbb{R})$-valued scalar
appearing in the $\cslr$ orbit data.  We assume
throughout this section that $a^2 >0$, which defines 
the positive area orbits, as are relevant for 
gravitational applications.  This volume form then 
allows us to construct a de-densitized one-form $p_A$
satisfying $|\nu| p_A = \wt p_A $.  Note that 
when specializing to $\cslr$ orbits, $\wt p_A$ is related 
to the associated scalar stream function $f$ by a similar
relation as in equation (\ref{eqn:Jdef}), which,
taking into account the relation (\ref{eqn:nudyn})
between the fixed and dynamical volume forms, is given by
\beq
f = -\upb^{BA}\partial_B(|a|p_A),
\eeq
where $|a|\equiv \sqrt{a^2}$.  Defining $p_A^0 = |a|p_A$,
this relation equivalently can be expressed as 
$*_{\upb} f = -dp^0$.  Similarly, we can define 
a de-densitized $\mfk{sl}(2,\mathbb{R})$ function
by the equation $\wt a_b = \wh a_b |\nu|$, which 
is related to the $\cslr$ 
orbit data by $a_b = |a|\wh a_b$.

With all this in hand, we can examine the expression
for the vorticity in terms of $\cslr$ orbit data.  
Using the results of section 4.2 of \cite{DonnellyFreidelMoosavianSperanza202012},
the vorticity 2-form $\bar{w}$ is given by
\begin{equation}
    \begin{aligned}
        \bar w &= dp -\frac12 \varepsilon_{abc}\wh a^a d\wh a^b\wedge d\wh a^c
        \\
        &=
    -\frac{*_{\upb}f}{|a|} -\frac12\varepsilon_{abc}\frac{a^a}{|a|}
    d\left(\frac{a^b}{|a|}\right)\wedge d\left(\frac{a^c}{|a|}\right)
    -\frac{d|a|}{|a|^2}\wedge p^0.\label{eqn:barworbit}
    \end{aligned}
\end{equation}

While the first two terms in this expression are well-defined 
functions of the $\cslr$ orbit data $(f,a_b)$, the third term is 
not, since it depends explicitly on the one-form $p^0$.  Although
$p^0$ is related to $f$ by the equation $*_{\upb}f = -dp^0$, this 
expression only determines $p^0$ up to shifts by exact forms,
$p^0\rightarrow p^0 + dA$.  This means that 
under an $\sdiff$ transformation
generated by $\xi^A = \upb^{BA} \nabla_B \phi$, 
$p^0$ will transform anomalously 
as $\delta_\xi p^0 = \lie_\xi p^0 +d A_\xi,$
where $A_\xi$ is a scalar function depending on the precise procedure
employed to construct a unique $p^0$ from a given $f$.\footnote{An
example of such a procedure is to select a fixed metric 
on the sphere, and impose that $d\star p^0=0$, where $\star$ is the natural
dualization associated with this metric.  Such a condition fixes
the shift ambiguity in $p^0$, but introduces dependence on the 
fixed background metric.}  This similarly implies an anomalous 
transformation of $\bar w$ under $\sdiff$ transformations:
\beq\label{eqn:barwanom}
\delta_\xi \bar w = \lie_\xi \bar w -\frac{d|a|}{|a|^2}
\wedge d A_\xi.
\eeq
Nevertheless, $\bar{w}$ retains the important property of being 
invariant under $\mfk{sl}(2,\mathbb{R})^S$ transformations

Because of the anomalous transformation property (\ref{eqn:barwanom}),
arbitrary moments of the vorticity scalar $*_{\upb} \bar{w}$ will
not yield invariant functions on the $\cslr$ coadjoint orbits.  However,
a set of invariants analogous to the $c_{1n}$ Casimirs for 
$\crr$ described in (\ref{eqn:c0nc1n}) 
can be obtained when integrating a single factor 
of $\bar{w}$ against a function of $|a|$.  Under $\sdiff$ transformations,
we have that 
\begin{equation}
    \begin{aligned}
\delta_\xi\Big(\bar w |a|^2 F'(|a|)\Big) &= \lie_\xi\Big(\bar w |a|^2 F'(|a|)\Big)
+ F'(|a|) d|a| \wedge dA_\xi \\
&= 
\lie_\xi\Big(\bar w |a|^2 F'(|a|)\Big) + d\Big(F(|a|) dA_\xi\Big).
\end{aligned}
\end{equation}
Since this is an exact form, integrating $\bar{w} |a|^2 F'(|a|)$
over the sphere will yield an invariant for the orbit:
\beq
c_F[(f,a)] = \int_S\bar w |a|^2 F'(|a|).
\eeq
Applying the definition (\ref{eqn:barworbit}) and using that 
$*_{\upb} f = 4\pi\nu_0 f$, $da^b\wedge da^c 
= 4\pi\nu_0\{a^b, a^c\}_{\upb}$, and $dp^0 = -*_{\upb} f$, 
this can be reexpressed as 
\beq
c_F[(f,a_b)] = -4\pi \int_S \nu_0\Big(
f\big[|a|F'(|a|)+F(|a|)\big] 
+\frac{F'(|a|)}{|a|} \varepsilon_{abc}a^a \{a^b,a^c\}\Big),
\eeq
which is now manifestly a function of the orbit data $(f,a_b)$. 
These can be expressed as a set of polynomial invariants by 
choosing $F(|a|) =\frac{-1}{4\pi} |a|^{2n}$, producing
\beq\label{eqn:c2n+1}
c_{2n+1}[(f,a_b)] = \int_S\nu_0\Big((2n+1)f a^2
+n\varepsilon_{abc}a^a\{a^b,a^c\}\Big)(a^2)^{n-1},
\eeq
which are the desired analogs of the $c_{1n}$ invariants 
from (\ref{eqn:c0nc1n}) for the $\crr$ algebra. 
It is possible to check directly that this expression is invariant under the coadjoint action. In the appendix \ref{app:casimirs} it is shown that, if we call the integrand
of (\ref{eqn:c2n+1}) $w_{2n+1}$, we obtained that the coadjoint action gives 
\be
ad^*_{(\phi,\alpha^a)} w_{2n+1} = \{\phi, x_{2n+1}\} + \{\alpha^a,
a_a a^{2n}\}.  \label{casint}
\ee 
which implies the invariance of its sphere integral \eqref{eqn:c2n+1}.

A somewhat strange feature 
is that only the even $n$ values of $c_{0,n}$ from the $\crr$
algebra match onto the Casimirs $c_{2n}$ for $\cslr$, and similarly only the 
odd $n$ values of the $c_{1,n}$ Casimirs match onto the $c_{2n+1}$ Casimirs
of $\cslr$.  
This discrepancy occurs due to the requirement that only integer powers of $a^2$
appear in (\ref{eqn:c2n}) and (\ref{eqn:c2n+1}), as is necessary to obtain
Casimirs that are polynomial in the generators.  An additional set of 
non-polynomial Casimirs for $\cslr$ can be obtained by allowing odd powers of
$|a| = \sqrt{a_a a^a}$ to appear in these expressions, and they 
would give analogs of the remaining $\crr$ Casimirs $c_{0,n}$ with $n$ odd
and $c_{1,n}$ with $n$ even.
Such square roots are relevant in the discussion of the area operator for 
the deformed algebra at the end of this section.

The Casimir functions $c_{2n}[(f,a_b)], c_{2n+1}[(f,a_b)]$ on the 
coadjoint orbits arise from Casimir elements of the algebra $\mfk{g}$.
These elements can be obtained from the functional expressions
by constructing the $\mfk{g}$-valued functions on $S$ 
$(\jfxn, \nfxn_b)$, normalized such that 
\begin{align}
\langle (f,a_b), \jfxn\rangle &= f, \\
\langle (f,a_b), \nfxn_a\rangle &= a_a,
\end{align}
where, as before, the left-hand side evaluates the pairing
between $(f,a_b)$ and the $\mfk{g}$ factors of $\jfxn$ and $\nfxn_a$,
and the right-hand side returns the functions on the sphere associated
with $f$ and $a_a$.  The Casimir elements of $\mfk{g}$ are then
obtained by replacing $f$ with $\jfxn$ and $a_b$ with $\nfxn_b$
in the expressions (\ref{eqn:c2n}) and (\ref{eqn:c2n+1}), giving\footnote{There are subtleties 
related to the ordering of the Lie algebra elements in these 
expressions, due to the fact, derived in appendix \ref{app:casimirs},
that $\nfxn_a(\sigma) 
\nfxn_b(\sigma') = \nfxn_b(\sigma')\nfxn_a(\sigma)
+\delta(\sigma-\sigma')\varepsilon\indices{_a_b^c}
\nfxn_c(\sigma')$, and hence, for example, the quantities
$\nfxn_a \nfxn^a \nfxn_b\nfxn^b$ and $\nfxn_a \nfxn_b
\nfxn^a \nfxn^b$ differ by divergent coefficients.
However, any choice of ordering for the Lie algebra 
elements define the same Casimir function on the coadjoint
orbits, and furthermore any choice of ordering for $c_{2n}$ and
$c_{2n+1}$ yields objects in the center of the universal
enveloping algebra.  We will not worry too much about this 
ordering for the remainder of this section since
the Casimir functions on the orbits are the important
quantities to work with to determine the representation
for the quantization of the phase space.  However, 
as we will see, it is interesting that the large $N$ 
limit of the $\mfk{u}(N,N)$ Casimirs picks out a preferred
ordering, and it would be interesting to understand how this 
preferred ordering could be obtained directly from the 
classical algebra.}
\begin{align}
c_{2n} &= \int_S \nu_0\, (\nfxn_a \nfxn^a)^n,  \label{eqn:c2nint}
\\
c_{2n+1} &= \int_S \nu_0\left((2n+1)\jfxn \nfxn_a \nfxn^a
+ n \varepsilon_{abc}\nfxn^a\{\nfxn^b,\nfxn^c\}\right)\left(\nfxn_d
\nfxn^d\right)^{n-1}.
\end{align}

With the expressions for the classical Casimirs in hand, we can 
now turn to matching these to the large $N$ limit of the Casimirs 
of the deformed algebra $\wh{\mfk{g}} = \mfk{u}(N,N)$.
To obtain convenient expressions for the deformed Casimirs, 
we begin by constructing the objects $(\defj, \defn_a)$ which are 
elements of $\Mat_{N\times N}\otimes\, \wh{\mfk{g}}$, normalized
according to 
\begin{align}
\frac{1}{N}\tr_{\bf{N}}\left(\,\defj\cdot\wh{Y}_\alpha\right) &= \algX_\alpha, 
\\
\frac{1}{N}\tr_{\bf{N}}\left(\,\defn_a\cdot\wh{Y}_\alpha\right) &=
\algZ_{a\alpha}.
\end{align}
where $(\algX_\alpha, \algZ_{a\alpha})$ are the basis of 
$\wh{\mfk{g}}$ introduced in section \ref{sec:sdiffsunn}.
These objects can be assembled into a single $2N\times2N$ 
matrix valued in $\wh{\mfk{g}}$ by tensoring with 
the $2\times 2$ matrices $(\mathbbm{1}_{2}, \rho_a)$.
The resulting object $\defph$ given by
\beq \label{eqn:hhat}
\defph = \left(\frac{1}{N} \mathbbm{1}_{2}\otimes \defj
- 2\rho^a\otimes\defn_a\right)
=\begin{pmatrix} \frac1N \defj +\defn_0 &-\defn_1+i\defn_2\\
\defn_1+ i\defn_2 & \frac1N\defj-\defn_0
\end{pmatrix},
\eeq
can then be shown to satisfy the key relations
\begin{align}
\left[\algX_\alpha, \defph \right]_{\wh{\mfk{g}}} 
&= -\frac{N}{2i}\left[\mathbbm{1}_{2}\otimes
\wh{Y}_\alpha, \defph \right], \\
\left[\algZ_{a\alpha},\defph\right]_{\wh{\mfk{g}}}
&= i\left[\rho_a\otimes \wh{Y}_\alpha, \defph\right],
\label{eqn:Zaahhat}
\end{align}
where the brackets on the right-hand side denote a matrix commutator
with the $\Mat_{2N\times 2N}$ factor of $\defph$.
Note that $\defph$ is closely related to the quantity 
$\opH$ defined in equation (\ref{eqn:hHmn}), 
since the latter is an operator-valued matrix
acting in a representation $\pi_\phsp$ of 
$\mfk{u}(N,N)$ coinciding with a quantization
of the phase space.  The exact relation between the two 
quantities is $\opH = i\hbar \pi_\phsp(\defph)$.
The relative 
coefficients of the $\defj$ and $\defn$ terms in (\ref{eqn:hhat})
are chosen to ensure the relation (\ref{eqn:Zaahhat}) holds.

Since the action of any Lie algebra element on $\defph$
can be expressed as a matrix commutator acting on $\defph$, 
the Casimir elements can be formed as in previous sections by 
taking powers $\defph^n$, where the product refers to the matrix
product on the $\Mat_{2N\times 2N}$ factor, thus producing 
an element of $\Mat_{2N\times 2N}\otimes \,\wh{\mfk{g}}^{\,\otimes n}$.
It is then straightforward to verify that 
\beq \label{eqn:hatcnunn}
\wh{c}_n = \frac{1}{2N}\tr_{\bf{2N}}\defph^n.
\eeq
commute with the action of $\wh{\mfk{g}}$, and hence define the Casimir
elements of the deformed algebra.

In order to match these Casimirs to those of the classical algebra,
we make use of the following relations at large $N$
(see appendix \ref{app:casimirs}):
\begin{align}
[\defj, \defn_a] &\equiv \defj\defn_a - \defn_a \defj = \op(N^{-1}),
\label{eqn:jna}\\
[\defn_a, \defn_b]&\equiv \defn_a \defn_b - \defn_b \defn_a
=\frac{2i}{N}\wh{\{\nfxn_a, \nfxn_b\} }+ [\defn_a,\defn_b]_{\wh{\mfk{g}}}
+ \op(N^{-3}). \label{eqn:nanb}
\end{align}
We can then expand (\ref{eqn:hatcnunn}) in terms of $\defj$ and $\defn_a$
using (\ref{eqn:hhat}).  Beginning with the even case, all terms
involving $\defj$ are suppressed by $\frac{1}{N}$, hence to leading order we have that
\begin{equation}\label{eqn:c2ntrace}
\begin{aligned}
\wh{c}_{2n} &=
\frac1{2N}\tr_{\bf{2N}}\left[(-2\rho^{a_1}\otimes \defn_{a_1})
(-2\rho^{b_1}\otimes \defn_{b_1})
(-2\rho^{a_2}\otimes \defn_{a_2}) \ldots
(-2\rho^{b_n}\otimes \defn_{b_n})\right]
\\
&=
\frac{1}{2N}(-2)^{2n}
\tr_{\bf{2}}\left(\rho^{a_1}\rho^{b_1}\ldots\rho^{a_n}\rho^{b_n} \right) 
\tr_{\bf{N}}\Big(\defn_{a_1}\defn_{b_1}\ldots \defn_{a_n}\defn_{b_n}\Big).
\end{aligned}
\end{equation}
In the trace over the $\rho^a$ matrices, we can expand
the products pairwise
using (\ref{eqn:rhoarhob}) to get $\rho^{a_i}\rho^{b_i}
 = -\frac14 \eta^{a_i b_i}\mathbbm{1}_{2} +\frac{i}{2}
\varepsilon\indices{^{a_i}^{b_i}^{c_i}}\rho_{c_i}$.  
Each term involving $\varepsilon\indices{^{a_i}^{b_i}^{c_i}}$
produces a commutator in $\defn_{a_i} \defn_{b_i}$ appearing in the 
second trace in (\ref{eqn:c2ntrace}).  We can then apply 
(\ref{eqn:nanb}) to find that this commutator can be replaced 
with the $\wh{\mfk{g}}$ Lie bracket, up to subleading terms in
$\frac{1}{N}$.  While these Lie brackets survive in the 
large $N$ limit, they simply produce elements of the universal 
enveloping algebra with different orderings of
the Lie algebra elements, all of which map to the same classical 
function on the coadjoint orbits.  Hence, we can drop these 
terms when matching to the classical Casimir (\ref{eqn:c2nint}).
Thus, keeping only the term proportional to the identity in each 
pairwise product of the $\rho^a$ matrices, (\ref{eqn:c2ntrace})
evaluates to
\beq
\wh{c}_{2n} =
(-1)^n \eta^{a_1b_1}\ldots\eta^{a_n b_n}\frac1N\tr_{\bf{N}}\Big(\defn_{a_1}
\ldots \defn_{a_n}\Big) +\mathcal{O}(N^{-1}).
\eeq
Demonstrating that this Casimir matches the classical expression
(\ref{eqn:c2nint}) follows immediately from the same argument
as in section \ref{sec:sdiffcas}, by expressing each Casimir in the 
Lie algebra basis, and applying the large $N$ relation (\ref{eqn:dhatd}).
This results in 
\beq
\wh{c}_{2n}\rightarrow (-1)^n c_{2n}+\mathcal{O}(N^{-1}).
\eeq

For the odd Casimirs, the leading piece in the contribution
involving only $\defn_a$ terms will be suppressed in the large $N$
limit, and hence we need to keep the first-order terms in the 
$\frac{1}{N}$ expansion.  The leading behavior at large $N$
is then given by
\begin{equation}
    \begin{aligned}
\wh{c}_{2n+1} &=
\frac{1}{2N}\left[ \vphantom{\frac{2n+1}{N}}
(-2)^{2n+1}\tr_{\bf{2}}\left(\rho^c\rho^{a_1}\rho^{b_1}
\ldots \rho^{a_n}\rho^{b_n}\right) \tr_{\bf{N}}\Big(\defn_c\defn_{a_1}
\defn_{b_1} \ldots \defn_{a_n}\defn_{b_n}\Big) \right. 
\\
& \hphantom{=\frac{1}{2N}} 
\left. (-2)^{2n}\tr_{\bf{2}}\left(\rho^{a_1}\rho^{b_1}
\ldots \rho^{a_n}\rho^{b_n}\right)
\tr_{\bf{N}}\left(\frac{(2n+1)}{N} \,\defj\,\defn_{a_1}\defn_{b_1}
\ldots \defn_{a_n}\defn_{b_n}\right)\right], 
\label{eqn:c2n+1trace}
\end{aligned}
\end{equation}
where we have applied (\ref{eqn:jna}) in the second term to move
$\defj$ to the left in each term in which it appears.
In the first line of (\ref{eqn:c2n+1trace}), we again expand out 
the products of $\rho^a$ matrices in pairs.  The terms involving only 
identity matrices in this product multiply with a term proportional
to $\tr_{\bf{2}} \rho^c$, which vanishes.  Hence we need to keep all terms with
one factor of $\varepsilon\indices{^{a_i}^{b_i}^{c_i}}$, since, as before,
each such term will produce a factor of $\frac{1}{N}$ due to commutations
of the $\defn_{a}$ in the second trace.  Applying 
equation (\ref{eqn:nanb}) and dropping terms involving 
$\wh{\mfk{g}}$ Lie brackets, we find for the first line
\beq
(-1)^n
\frac{n}{N} \frac1N\tr_{\bf{N}}\Big(\varepsilon^{cab} \,
\defn_c \wh{\{\nfxn_a,\nfxn_b\}} \defn_{a_2}\defn^{a_2}\ldots
\defn_{a_n}\defn^{a_n}\Big).
\eeq
For the second line of (\ref{eqn:c2n+1trace}), we can simply
keep all terms proportional to the identity in each pairwise 
$\rho^a$ product.  This term then evaluates to 
\beq
(-1)^n
\frac{(2n+1)}{N} \frac{1}{N}\tr_{\bf{N}}\Big(\defj\,\defn_{a_1}\defn^{a_1}
\ldots \defn_{a_n}\defn^{a_n}\Big).
\eeq
Combining these terms and again applying the large $N$
relationship (\ref{eqn:dhatd}), 
we find that $\wh{c}_{2n+1}$ approaches the classical Casimir
after rescaling by $N$, 
\beq
\wh{c}_{2n+1}\rightarrow \frac{(-1)^n}{N} c_{2n+1} + \mathcal{O}(N^{-2}).
\eeq

We can translate these correspondences to matching conditions for
the Casimir operators in the representation $\pi_\phsp$ 
by recalling that each Lie algebra element is rescaled 
by $i\hbar$ in the representation.  Hence we can define the operators
\begin{equation}
    \begin{aligned}
        \wh{C}_{2n}&=(-1)^n (i\hbar)^{2n} \pi_\phsp(\wh{c}_{2n}),
        \\
\wh{C}_{2n+1}&= N(-1)^n(i\hbar)^{2n+1}\pi_{\phsp}(\wh{c}_{2n+1}),
    \end{aligned}
\end{equation}
which can be matched to the classical invariant functions on the 
gravitational phase space, obtained by pulling back
the orbit invariants (\ref{eqn:c2n}) and (\ref{eqn:c2n+1}) via the moment
map.  This matching should again determine the deformation
parameter $N$ as well as the representation of $\mfk{su}(N,N)$
corresponding to a quantization of the phase space.  Carrying out
this matching in detail requires a thorough investigation
into the unitary representations of $\mfk{su}(N,N)$,
which we leave for future work.

It is interesting to examine in more detail the quadratic 
Casimir, whose full expression is 
\beq
\wh{c}_2= -\frac{1}{N}\tr_{\bf{N}}(\defn_a\defn^a) +\frac{1}{N^2}\frac{1}{N}
\tr_{\bf{N}}(\,\defj^2).
\eeq
The first term is the piece that survives in the large $N$ limit,
and is related to the area of the surface embedded in
spacetime.  The area operator can be defined by 
\beq \label{eqn:Ahat}
\wh{A} = \frac{1}{N}\tr_{\bf{N}}\sqrt{\defn_a\defn^a},
\eeq
where the square root should be interpreted an object that 
yields the square root of the  operator $\pi_{\phsp}(\defn_a\defn^a)$ 
in a representation of the algebra.   
Here we find that at infinite $N$,
the area operator is a Casimir of the continuum group,
as first demonstrated in \cite{Donnelly:2016auv}.  However, at finite
$N$, this operator is not a Casimir, and instead is a hyperbolic 
element of $\mfk{su}(N,N)$, up to higher order corrections in the 
universal enveloping algebra.  This suggests that the area operator
becomes noncommutative after including finite $N$ corrections to the 
algebra.  This result is reminiscent of the recent investigations
into large $N$ algebras in holography
\cite{Leutheusser:2021qhd, Leutheusser:2021frk,
Witten:2021unn, Chandrasekaran:2022eqq}, where, in particular,
the failure of the area operator to be central upon including 
$\frac{1}{N}$ corrections leads to a deformation of the 
associated von Neumann algebras from type III to type II.  
It would be interesting to further explore the connection between the
noncommutativity of the area operator in the present context and 
the appearance of deformed von Neumann algebras in holography.

\section{Conclusion and future work}

In this work, we have undertaken the first steps of studying, at a quantum level, the symmetries of a finite region of space identified in Ref.~\cite{Donnelly:2016auv}.
Inspired by the fact, shown in Ref.~\cite{DonnellyFreidelMoosavianSperanza202012} that the Lie algebra of the Wigner little group is
\begin{equation} \label{sdiffsun correspondence}
\sdiff = \lim_{N \to \infty} \su(N),
\end{equation}
we have looked for a deformation of the corner symmetry algebra which would generalize \eqref{sdiffsun correspondence}.
In extending the symmetry to include boost transformations of the normal plane, we arrived at two generalizations of the matrix regularization \eqref{sdiffsun correspondence} to noncompact groups:
\begin{align} \label{sdiffr}
\mfk{sdiff}(S)\oplus_{\mcal{L}}\mbb{R}^S &= \lim_{N \to \infty} \Sl(N,\C) \oplus \R, \\
\mfk{sdiff}(S)\oplus_{\mcal L}\mfk{sl}(2,\mbb{R})^S &= \lim_{N \to \infty} \su(N,N). \label{sdiffslr}
\end{align}
These deformations nontrivially combine the diffeomorphisms of the sphere with normal boosts such that in the large-$N$ limit the semidirect sum structure is recovered.
While we have established identities \eqref{sdiffr}, \eqref{sdiffslr} at the level of the structure constants, in section \ref{sec:casimirs} we extended this analysis to Casimir invariants of the groups, showing that the large-$N$ limits of the well-known $\Sl(N,\C)$ and $\su(N,N)$ Casimirs yield the complete set of invariants of identified for the corresponding infinite-dimensional Lie algebras.
The Casimirs allow us to determine the representation of the symmetry group in terms of physical properties of the surface $S$, and in particular, allow us to argue for a particular scaling of the deformation parameter $N$.
The Casimirs also give a set of commuting operators at the quantum level.
Interestingly, the area operator, which was shown in Ref.~\cite{DonnellyFreidelMoosavianSperanza202012} to play a special role in the classification of orbits, is not among the Casimirs but becomes noncentral at finite $N$.
This fact remains puzzling but may have implications for black hole entropy for which the area plays a crucial role.

Our work opens up many potential avenues for future works, so we spend the majority of this section identifying the most interesting future directions.

\subsection{Detailed Casimir matching}
\label{sec:detcasmatch}

In section \ref{sec:casmatch}, we outlined the Casimir matching procedure
for the case of $\mfk{su}(N)$.  There we found that the matching conditions in an irreducible representation
determine how the deformation parameter $N$ scales with $\frac{A}{4G\hbar}$, 
as displayed in equation (\ref{eqn:NsimA}).  It would be quite interesting
to carry out this matching in more detail to not only determine the value 
of $N$ but to also identify the representation that should be 
employed in the quantization of the phase space.  In the large $N$ limit, we should
expect to find a relation between the shape of the Young diagram for the 
representation and the function $W$ on the sphere, or the associated 
measured Reeb graph derived from $W$, which, as explained in 
\cite{DonnellyFreidelMoosavianSperanza202012,izosimov2016coadjoint}, is an additional invariant of
the continuum algebra $\sdiff$.  Given the large amount
of literature related to the large $N$ limits of representations 
of $\mfk{su}(N)$ (see, e.g.\ \cite{Cordes:1994fc}), 
it seems likely that this more detailed matching would be 
achievable.  

The scaling $\frac{A}{4G\hbar}\sim N^3$ identified in equation (\ref{eqn:NsimA})
is somewhat odd from the perspective of AdS holography or matrix models, in which
it is more common for the entropy to scale like $N_{AdS}^2$. 
This suggests 
that there might be an issue with trying to only quantize the $\sdiff$ subalgebra
in the process of attempting to obtain an understanding of the entropy
of the surface $S$.  Instead, it seems likely that one would need to 
work with one of the enlarged algebras  $\crr$ or $\cslr$ to obtain a sensible 
relation for the entropy from the Casimir matching procedure.  

This motivates further investigating the large $N$ representation
theory of the deformed algebras $\mfk{sl}(N,\mathbb{C})\oplus \mathbb{R}$
and $\mfk{su}(N,N)$.  Unfortunately, the literature on the unitary representations
of these groups is somewhat sparse.  The unitary representations of 
$\mfk{gl}(N,\mathbb{C})$ were classified in \cite{Vogan1986}, and some 
results on $\mfk{su}(N,N)$ are given in \cite{Molchanov1998}.  A standard reference
on the general theory of representations of semisimple groups 
is \cite{Knapp2001}.  It would be interesting to investigate this 
representation theory in more detail, and to identify which 
representations occur in the large $N$ limit when matching to 
the continuum algebras $\crr$ and $\cslr$.  Ultimately, one would hope to 
be able to identify the analog of equation (\ref{eqn:NsimA}) for these 
groups, which may yield the expected relation between $N^2$ and the entropy
$\frac{A}{4G \hbar}$.

Alternatively, it may be that the  $N^2$ scaling is not
appropriate for the localized gravitational subregions
considered here, and the $N$ appearing in our algebra deformation is a priori a different entity that the $N_{\text{AdS}}$ appearing in holography.
In our context,  $N$ appears as a deformation parameter for the corner symmetry algebra and corresponds to a measure of the 
corner surface area in Planck units.
On the other hand $N_{\text{AdS}}$ appearing in holography as a label for the boundary gauge group is related to the ratio of the cosmological scale and the Planck scale through the relation $N_\text{AdS}^2\sim (\hbar G\Lambda)^{-1}$
\cite{Aharony:1999ti, Maldacena:1997re} in four spacetime dimensions (although, see footnote \ref{ftn:Nscale} for 
situations with different parametric dependence of the 
entropy on $N$). It is natural to expect some functional relation between the corner $N$ and the holographic  $N_{\text{AdS}}$. The exact nature of this relationship is not established at this stage.

\subsection{Computation of characters}

We have introduced large-$N$ limits of the groups $\mfk{sl}(N,\C) \oplus \R$ 
and $\su(N,N)$ and shown the continuum limit of the structure constants as well as the Casimir invariants.
These invariants allow us to establish a correspondence not only between the finite and infinite-dimensional Lie groups, but also their representations.
It would be interesting to see how much of the finite-dimensional representation theory can be carried over to the large-$N$ limit.
In particular, it would be interesting to compare characters of the finite-dimensional Lie group representations to those of their continuum counterparts.

Group characters are especially important in physics because they are essentially quantum-mechanical partition functions, encoding the number of states in each irreducible representation as a function of the physical values of the generators.
In the gravitational context, the most important character is that of the global boost, which, in the deformed algebras,
coincides with the $\mathbb{R}$ factor of
$\mfk{sl}(N,\mathbb{C})\oplus\mathbb{R}$, or a generic hyperbolic generator
in $\su(N,N)$.
This operator plays an important role in both black hole thermodynamics, where it defines the time-translation symmetry associated with Killing horizons, and in quantum field theory where it defines the modular Hamiltonian of a quantum field theory restricted to a half-space or conformal field theory restricted to a sphere.
The boost character is therefore essential in relating the entropy of horizons -- which is controlled by the density of states -- to the value of the charges, which are determined by the horizon geometry.

An important first step would be to understand the relation between character formulas for $\su(N)$ and $\sdiff$ in the limit of large $N$.
On the $\su(N)$ side, the large-$N$ limit of the characters can be obtained from the Itzykson-Zuber integral formula \cite{ItzyksonZuber1980}.
This large-$N$ limit was studied in \cite{Matytsin199306} which expressed the leading asymptotics of the character in terms of the complex inviscid Burgers equation (or the Hopf equation) \cite{Bateman191501,Burgers1948} whose solutions have been studied in Ref.~\cite{KenyonOkounkov200507}.
The characters have the leading-order behavior $\exp(N^2 F_0 + F_1 + \ldots)$ where $F_0$ is an on-shell action and $F_1$ the first subleading correction in the $1/N^2$ expansion.
Independently, certain characters of $\sdiff$ have been calculated using the Atiyah-Bott localization formula, and take the form of divergent sums \cite[Equation (3.19)]{Penna201806}. It would therefore be interesting to understand whether these $\sdiff$ characters can be obtained as appropriate limits of the $\su(N)$ characters.
Since the leading term of the $\su(N)$ character diverges at large $N$, it cannot be calculated within $\sdiff$ --- rather, we expect it to appear as a divergence that must be renormalized away.
Having subtracted this leading divergence one expects to find agreement between the renormalized $\sdiff$ characters and $1/N$ corrections to the $\su(N)$ characters: the latter would appear as corrections to the leading-order result of Ref.~\cite{Matytsin199306}.

An important next question is whether the large-$N$ calculation of characters can be extended to large-$N$ limits of $\Sl(N,\C)\oplus \mathbb{R}$ and $\su(N,N)$ and related to character formulas for $\crr$ and $\cslr$ respectively.
Such characters can in principle be computed from the analog of Kirillov's character formula for reductive groups \cite{Rossman1978}, and we expect similar divergent behavior of characters seen for $\su(N)$ to hold for $\Sl(N,\C)\oplus \mathbb{R}$ and $\su(N,N)$. 

In the case of the noncompact groups $\Sl(N,\C)\oplus \mathbb{R}$ and $\su(N,N)$ the calculation of characters plays a further important role.
Since unitary representations of noncompact groups are infinite-dimensional, the direct analog of the formula \eqref{SYM} for Yang-Mills theory cannot apply.
Instead, we expect the global boost $K$ to have a nonzero expectation value which leads to an insertion of $\exp(-2 \pi K)$ in the partition function.
This suggests it is characters of the global boost $K$ (or suitable analytic continuations thereof), and not dimensions, which are the relevant quantities for counting states in representations of $\Sl(N,\C)\oplus\mathbb{R}$ and $\su(N,N)$.

The chief physical application of such characters is in understanding the entropy of a region of space bounded by the corner $S$.
The characters give a way to organize the computation of entropy, see \cite{Anninos:2020hfj} for a concrete example.
They would in principle give a way of calculating the entanglement spectrum in terms of geometric properties of the surface $S$, which would be an intriguing application of the formalism developed in Ref.~\cite{Donnelly:2016auv} and further explored in Ref.~\cite{DonnellyFreidelMoosavianSperanza202012} and this work. 

\subsection{Topological aspects of large-$N$ limit}
\label{subsec:topological aspects of large-N limit}

In section \ref{sec:sdiffcas}, we obtained a correspondence
between the Casimirs of $\mfk{su}(N)$ and an associated set of 
Casimirs for the continuum algebra $\sdiff$, which coincide with 
generalized enstrophies of incompressible hydrodynamics.
However, a complete classification of the invariants of $\sdiff$ involves additional topological information contained in the measured Reeb graph of the function $W$ on the sphere --- see Ref.~\cite{izosimov2016coadjoint} for a proof and Ref.~\cite{DonnellyFreidelMoosavianSperanza202012} for discussion in the context of the corner symmetry algebra. 
Each coadjoint orbit of $\sdiff$ is labeled by a function $W$ on $S^2$, and the Reeb graph encodes the topology of the level sets of $W$. 
This raises the question of how this topological data arises from the large $N$ limit of $\mfk{su}(N)$.  
Since the invariants $\wh{c}_k$, $k = 2,\ldots, N$ comprise a complete set of Casimirs for $\mfk{su}(N)$, there appears to be no topological data present at finite $N$. 
Instead, the topology is contained in the way the limit $N\rightarrow\infty$ is taken.
In this limit, the topology of the surface $S$ restricts the allowed representations of $\mfk{su}(N)$ that have good infinite $N$ limits, and different topologies should single out different representations.  
In order to make this connection more precise, one would like to obtain the Reeb graph from some property of the large $N$ limit, such as the shape of the Young diagrams for the allowed representations.  
At finite $N$, the object corresponding to the function $W$ is a hermitian matrix $\wh{W}$, and a natural way to approach the large-$N$ limit is to study the trace of the resolvent, $\tr[(\lambda I - \wh{W})^{-1}]$ as $N \to \infty$.
In this limit the trace of the resolvent develops a branch cut, and the discontinuity across the cut encodes the spectral density of $\wh{W}$ and hence all of the Casimirs.
A natural conjecture is that the topology of this branch cut is related to the topology of the Reeb graph.
An intriguing possibility arises from the observation that trivalent vertices in the Reeb graph are associated with logarithmic singularities in the eigenvalue density of $\wh{W}$.
It is then tempting to conjecture that the Reeb graph data is encoded in
the branching structure of the resolvent as $N\rightarrow\infty$.

A related topological consideration comes from the interpretation
of the finite $N$ algebra as a sum over all possible topologies of the 
surface $S$ \cite{Bars199706,deWit:1989yb}.  This is related to the fact that 
$\mfk{su}(N)$ can reproduce the group of area-preserving diffeomorphisms
of any Riemann surface as $N\rightarrow\infty$, depending on how
the limit is taken.  For example, we could instead work with torus harmonics
$Y_{mn}$ as opposed to spherical harmonics, and these admit 
a finite $N$ deformation to fuzzy torus harmonics $\wh{Y}_{mn}$ 
which satisfy an $\mfk{su}(N)$ algebra \cite{Floratos1988, Fairlie1989, Barrett2019}.
Therefore, at finite $N$ the fuzzy torus harmonics must be expressible in
terms of fuzzy spherical harmonics by a change of basis,
\beq
\wh{Y}_{mn} = \sum_\alpha B^\alpha_{mn} \wh{Y}_\alpha.
\eeq
This change of basis becomes singular in the large $N$ limit, reflecting 
the fact that this limit requires one to choose a basis appropriate
to the set of smooth functions in the limiting topology.  It would 
be quite interesting to explore ideas related to the finite $N$ algebra
and sums over the topologies of the surface in more detail.  

A different topological aspect arising from the  larger groups 
$\mfk{su}(N,N)$ and $\cslr$ is related to nontrivial $\slr$-bundles over $S$. These were argued to be closely associated with nonzero NUT charges for the surface \cite{DonnellyFreidelMoosavianSperanza202012}. The natural question is whether the information of these nontrivial bundles could somehow be encoded in the $\su(N,N)$ regularization of Section \ref{sec:sdiffsunn}. One possibility is that the information of these nontrivial bundles could be only emergent as we take $N \to \infty$,
similar to the emergence of the topology of $S$ discussed above. Note that the 
continuum algebra is different from $\cslr$ when working 
with nontrivial bundles: rather than taking the form of a semidirect
product, the symmetry algebra is instead a nontrivial
extension of $\sdiff$ by $\mfk{sl}(2,\mathbb{R})^S$.  Presumably
these algebras could be obtained 
by considering a different large $N$ limit involving twisted generators 
$\wt \algX_\alpha = \algX_\alpha + \frac{N}{2}A\indices{_\alpha^\mu^a}
\algZ_{a\mu}$, with the tensor $A\indices{_\alpha^\mu^a}$ subject to 
some consistency conditions needed to ensure a good large $N$ limit.  
Note that $\wt \algX_\alpha$ are divergent in the original large $N$ limit, 
implying that these generators lead to a different continuum algebra
which conjecturally coincides
with the symmetry algebra associated 
with nontrivial $\mfk{sl}(2,\mathbb{R})$ bundles.  
The tensor $A\indices{_\alpha^\mu^a}$ would then be related to the 
curvature of a connection on the resulting $\mfk{sl}(2,\mathbb{R})$ bundle,
which characterizes the Lie algebra 2-cocycle defining the 
 extension, as discussed in \cite[Appendix A]{DonnellyFreidelMoosavianSperanza202012}.
It is thus conceivable that the data of different topologies of $S$ 
along with different $\slr$ bundles  are contained in the finite $N$
algebra $\mfk{su}(N,N)$. 

Finally, throughout this work, we have eliminated the central
generator $\algX_{00}$ since it arises from the constant
function on the sphere, which does not generate a diffeomorphism
in the continuum algebras.  However, a question remains
as to
whether the charges associated with this central generator 
should be
nonzero in the quantum theory.  It would be interesting to 
investigate this, and determine whether these central
charges bear any relation to the NUT charges discussed above.

\subsection{Deformation of the full diffeomorphism algebra}
\label{sec:fulldiff}

This work has focused on three subalgebras of the full corner symmetry
algebra, all of which involve area-preserving diffeomorphisms as opposed 
to the full diffeomorphism algebra of $S^2$. Nevertheless, this raises the 
question whether the deformations considered here could eventually be lifted
to the full corner symmetry algebra $\gslr = \diff\oplus_{\mcal{L}}\slr^S$,
as suggested by figure \ref{fig:subalgs}.
The main challenge here would be to determine the deformation for the full
$\diff$ algebra, after which one may be able to extend it to the corner
symmetry algebra following similar techniques as employed in this paper.  
Our initial investigations on this topic involve an explicit computation
of the structure constants for $\diff$, which are derived in detail in 
appendix \ref{app:diffsc}.  However, there are several indications
that any deformation of this algebra will involve a more complicated 
procedure than the analogous problem for $\sdiff$.  Recently, a no-go theorem
for the existence of such a linear deformation of $\diff$ was proven in
\cite{EnriquezRojoProchazkaSachs202105}.  This suggests that the full corner 
symmetry algebra does not admit such a deformation, although the possibility remains
that $\diff\oplus_{\mcal{L}}\slr^S$ is deformable even though $\diff$ itself is not.

A more likely possibility is that the deformation would involve a nonlinear
algebra, such as those appearing in the theory of quantum groups.  Relatedly, one
might consider looking for a deformation of a larger algebra containing $\diff$,
such as the higher spin Schouten algebra of all symmetric multivector fields on
the sphere.  
This higher spin picture is naturally associated with the higher spin-weighted spherical harmonics, whose deformation was suggested  in appendix \ref{sec:fuzH} to be a set of rectangular matrices.  These matrices are associated
with changes in the value of $N$, and hence one might conjecture that 
the natural deformation of this higher spin algebra involves a sum over all possible
values of the deformation parameter $N$.  It is possible that a deformation
of this higher spin algebra can be consistently defined, and only in the classical
limit do the $\diff$ generators close to form a subalgebra.  We leave further investigation
into these ideas to future work.  

Furthermore, we have restricted our attention to the part of the corner symmetry algebra that preserves the corner $S$ and exclude the so-called corner deformations, which move $S$ itself. By including normal translations of the corner, we would instead end up with the symmetry group 
\cite{Speranza:2017gxd, CiambelliLeigh202104, Freidel:2021cbc, Ciambelli:2021nmv, Freidel:2021dxw, Speranza:2022lxr,Ciambelli:2022cfr}
\begin{equation}
    (\tenofo{Diff}(S)\ltimes\tenofo{SL}(2,\mbb{R})^S)\ltimes(\mbb{R}^2)^S.
\end{equation}
It has been shown  in Ref.~\cite{CiambelliLeigh202104} that this is the maximal subalgebra of the diffeomorphism group of the bulk spacetime that is associated to an isolated corner $S$. Therefore, the full regularization of corner symmetry should include this generalization, and it would be interesting to explore deformations of this algebra as well.

\subsection{Other algebra deformations}
As indicated in figure \ref{fig:subalgs}, this work identified a natural 
nested sequence of deformed algebras $\su(N)\subset\mfk{sl}(N,\mathbb{C})
\oplus \mathbb{R}
\subset \mfk{su}(N,N)$ coinciding with the continuum algebra 
inclusions
$\sdiff\subset \sdiff\oplus_{\mathcal{L}}
\mathbb{R}^S\subset \sdiff\oplus_{\mathcal{L}}
\mfk{sl}(2,\mathbb{R})^S$.  In particular, the intermediate algebra 
$\mfk{sl}(N,\mathbb{C})\oplus 
\mathbb{R}$ arises as the subalgebra of $\su(N,N)$ preserving a complex 
structure, which in the $\pi_{\bf{2N}}$ representation is just the matrix 
$2\wh{Y}_{1,00}$.
It is noteworthy that a number of other interesting algebras appear
as intermediate steps between $\mfk{su}(N,N)$ and $\mfk{su}(N)$.  In particular,
if one instead looks for the algebra preserving the paracomplex structure 
$2\wh{Y}_{0,00} = \begin{pmatrix}\mathbbm{1}&0\\0&-\mathbbm{1}  \end{pmatrix}$,
the result is the maximal compact subalgebra $\su(N)\oplus\su(N)\oplus\mfk{u}(1)$.
This algebra may be relevant as a corner symmetry algebra in Euclidean signature,
where one is interested in rotations instead of boosts in the normal plane.
We can also form the algebra $\mfk{sp}(2N,\mathbb{R})\oplus \mathbb{R}$ 
as the set of generators $\wh{A}$ preserving a real structure, 
meaning that $\wh{A}^*J_r = J_r\wh{A}$, with  $J_r$ a matrix satisfying $J_r^*J_r = \mathbbm{1}$ and  $*$
denotes complex conjugation.  This matrix can be taken to be $J_r = 2\wh{Y}_{2,00}
 = \begin{pmatrix} 0&-i\mathbbm{1} \\ -i\mathbbm{1}&0\end{pmatrix}$.
Finally, one can obtain the quaternionic orthogonal algebra 
$\mfk{so}^*(2N)\oplus \mathbb{R}$ 
by restricting to generators $\wh{B}$ that preserve a pseudoreal
structure, meaning $\wh{B}^* J_p = J_p \wh{B}$, with $J_p$ satisfying 
$J_p^* J_p = -\mathbbm{1}$.  Such a pseudoreal structure is given by
$2\wh{Y}_{1,00} = \begin{pmatrix}0&\mathbbm{1}\\ \mathbbm{1}&0\end{pmatrix}$.
It is an interesting question whether these other intermediate algebras
have large $N$ limits in terms of diffeomorphism algebras of $S$.  

In a different vein, we note that the limit of $\mfk{su}(N,N)$ to the 
continuum algebra required a specific scaling of the generators 
according to (\ref{eqn:pi2NZ}).  There exists a different scaling of 
generators that also yields a finite  limit as $N\rightarrow\infty$: 
we can rescale the $\sdiff$ generators  according to
\begin{align}
\wt{\algX}_\alpha =\frac{1}{N^2}\algX_\alpha.
\end{align}
In terms of these, the algebra becomes
\begin{align}
[\wt{\algX}_\alpha, \wt{\algX}_\beta] 
&= \frac{1}{N^2}\wh{C}\indices{_\alpha_\beta^\gamma}\wt{\algX}_\gamma, \\
[\wt{\algX}_\alpha, \algZ_{a\beta}]
&= \frac{1}{N^2}\wh{C}\indices{_\alpha_\beta^\gamma}\algZ_{a\gamma},  \\
[\algZ_{a\alpha},\algZ_{b \beta}] 
&=\varepsilon\indices{_a_b^c}\wh{E}\indices{_\alpha_\beta^\gamma}\algZ_{c\gamma}
-\eta_{ab}\wh{C}\indices{_\alpha_\beta^\gamma}\wt{\algX}_\gamma.
\label{eqn:sphalg}
\end{align}
The $N\rightarrow\infty$ limit now implements a different contraction of the 
algebra in which the generators $\wt{\algX}_\alpha$ become central, and 
(\ref{eqn:sphalg}) indicates that the resulting algebra is a nontrivial central
extension of the sphere algebra $\mfk{sl}(2,\mathbb{R})^S$
\cite{loop}.  These centrally extended
sphere algebras have been explored, for example, in
\cite{Dowker:1990ss,Frappat:1989gn}, and it is interesting to see that they arise
from a nontrivial limit of the matrix algebra $\mfk{su}(N,N)$.  Whether
this limit has any bearing on the quantization of the corner symmetry algebra
remains to be seen.

\subsection{Connections to holography}

Although the algebras considered in this work arose as deformations 
of classical algebras arising from a bulk gravitational theory, there are 
several connections between these deformations and features of holographic
models of quantum gravity.  Many examples of holography arise as  matrix
models, which naturally are associated with $\mfk{su}(N)$ symmetry
\cite{Banks:1996vh,Aharony:2008ug,Kapustin:2009kz,Saad:2019lba}.  Indeed,
the supermembranes arising in string theory and M-theory were the original 
context in which the identification of $\sdiff$ as the large $N$ limit 
of $\mfk{su}(N)$ arose \cite{Hoppe198201,Hoppe198901,PopeStelle198905}.  
Related ideas appear in the holographic spacetime model of reference
\cite{Banks:2018ypk}.  
Such examples give a motivation for considering the deformed algebras 
described in this paper, and describe models where the exact
diffeomorphism symmetry is an emergent symmetry in the low-energy, classical 
theory.
While we have approached the question from the perspective of gravitational theory, it would be interesting to obtain deformations of the corner symmetry algebra from a more fundamental UV theory.
There has been some progress in understanding the closely related concept of ``entangling branes'' in string field theory \cite{Balasubramanian:2018axm} and in topological string theory \cite{Donnelly:2016jet, Donnelly:2018ppr, Hubeny:2019bje, Donnelly:2020teo, Jiang:2020cqo} but their precise relation to symmetries in the emergent gravitational theory remains unclear.

Finally, the large $N$ limits considered in the present work 
have interesting connections to recent work on von Neumann algebras
arising in the large $N$ limit of holographic conformal field theories
\cite{Leutheusser:2021qhd, Leutheusser:2021frk,
Witten:2021unn, Chandrasekaran:2022cip, Chandrasekaran:2022eqq}.  Particularly
intriguing is the fact that the area operator defined in equation (\ref{eqn:Ahat})
is central at infinite $N$, but becomes noncentral upon including perturbative
$\frac{1}{N}$ corrections.  This bears some resemblance to  
aspects of the crossed
product construction considered in \cite{Witten:2021unn}, 
where,
in particular, it was important to realize that the area operator
is singular in the quantum theory, and only becomes a well-defined 
operator after adding the bulk modular hamiltonian to it, which 
accounts for the noncommutativity at subleading order in Newton's constant.  
 A fruitful future
direction for the present work is to try to make this connection more precise,
and look to understand the corner symmetries and their deformations in
terms of von Neumann algebras.

\bigskip
\paragraph{Acknowledgement} 
We thank Rob Leigh and Lee Smolin  for helpful discussions.
We are grateful to the organizers of the conference ``Quantum Gravity Around 
the Corner'' held at Perimeter Institute.
The work of SFM is funded by the Natural Sciences and Engineering Research Council of Canada (NSERC) and also in part by the Alfred P. Sloan Foundation, grant FG-2020-13768.  AJS is supported by the Air Force Office of
Scientific Research under award number FA9550-19-1-036.
Research at Perimeter Institute is partly supported by the Government of Canada through the Department of Innovation, Science and Economic Development Canada and the Province of Ontario through the Ministry of Colleges and Universities.

\appendix 

\section{Spherical harmonics and fuzzy spherical harmonics}
\label{appendix:spherical_harmonics}
In this appendix, we define the conventions used for continuum spherical harmonics which are used as 
an explicit basis of functions on the unit sphere.  In section \ref{sec:sh}, 
we describe the structure constants for multiplication
and the Poisson bracket with respect to this basis.  The conventions 
for spin-weighted spherical harmonics, which are used 
when evaluating structure constants for differential operators on the sphere,
are subsequently presented in section \ref{sec:swh}.  
We then describe the basis of fuzzy spherical 
harmonics in section \ref{sec:fuzH}
as finite-dimensional Hermitian matrices, and review the standard result showing that 
the structure constants for the commutator of these matrices approaches the structure constants
of $\sdiff$.  The matrix product is given by a simple formula
in terms of the Wigner $6j$ symbol, and we present 
an expression for it that immediately
yields the large $N$ expansion of the product to any desired order.  We demonstrate 
the utility of this formula in section \ref{sec:mpexpansion} 
by determining the $\op(\frac{1}{N^2})$ correction
to the matrix product, and verifying that it takes the form of a Fedosov $\star$-product
for the sphere, viewed as a symplectic manifold.

\subsection{Spherical harmonics} \label{sec:sh}
We use spherical harmonics $Y_{Aa}(\theta,\varphi)$ where $A\in \mathbb{N}$ and $a\in \{-A,\ldots,+A\}$. We work with the 
Racah normalization convention and the Condon-Shortley phase, 
which imply
\begin{equation} \label{intYY}
     \int_{S} \nu_0\,
     Y_{Aa}(\theta,\varphi) Y_{Bb}(\theta,\varphi)  = (-1)^{a} \,\frac{\delta_{A,B} \delta_{a,-b}}{(2 A + 1)}.
\end{equation}
where  $\nu_0=
\frac1{4\pi}\cdot\frac12\epsilon_{AB} \rd \sigma^A \wedge \rd \sigma^B =
\frac1{4\pi}\sin \theta\, \rd \theta \wedge \rd \varphi$ is the unit-normalized volume form on the  sphere.  
It will be convenient to adopt a condensed index notation 
$\alpha = (A,a)$ in which $Y_\alpha$ is shorthand for the spherical harmonic functions $Y_{Aa}(\theta,\varphi)$.
Then \eqref{intYY} defines a real metric 
\beq\label{eqn:shmetric}
\delta_{\alpha \beta} = \frac{(-1)^{a}}{ (2A+1)} \delta_{A,B} \delta_{a, -b},
\eeq
on the vector space 
of functions on the sphere.
The $\alpha$ indices will be raised and lowered with this real metric.
Complex conjugation and orientation reversal act as 
\be \label{eqn:YAareality}
(Y_{A,a}(\theta,\varphi))^*  = (-1)^{a} Y_{A,-a}(\theta,\varphi),
\qquad
(Y_{A,a}(\pi-\theta,\varphi+\pi))^* = (-1)^{A} Y_{A,-a}(\theta,\varphi).
\ee

The multiplication structure constants $E\ind{_\alpha_\beta^\gamma}$ are defined via $Y_\alpha Y_\beta = E\ind{_\alpha_\beta^\gamma}Y_\gamma$. 
Their explicit values are given in terms of Wigner $3j$ symbols
\cite{NIST3j}  
according to 
\begin{equation} \label{eqn:Eabccontra} 
E\ind{_\alpha_\beta^\gamma} = (-1)^{c}\,(2C+1) \tj{A&B&C \\ a&b&-c} 
\tj{A&B&C\\0&0&0}. 
\end{equation}
Lowering one index with the metric (\ref{eqn:shmetric})
gives the totally symmetric tensor $E_{\alpha \beta \gamma}$:
\begin{equation} \label{eqn:Eabc}
   E_{\alpha\beta\gamma} =   
   \tj{A&B&C\\a&b&c}\tj{A&B&C\\ 0&0&0}= \int_S \nu_0\, Y_\alpha  Y_\beta  Y_\gamma  . 
\end{equation}
Note that this is nonzero only when $A+B+C$ is even.

The Poisson bracket of two functions on the sphere is defined as 
\beq \label{eqn:pb}
\{f,g\} = \ep^{AB}\nabla_A f \nabla_B g,
\eeq
where $\epsilon^{AB}$ is the negative inverse of the standard area form $\ep_{AB}$ on the 
unit sphere.\footnote{It is an antisymmetric tensor normalised by the condition  $\epsilon^{12} =1/\sqrt{q}$ where  $q$ is the metric determinant in the coordinate chosen.  Note that this 
Poisson bracket differs from the Poisson bracket $\{,\}_{\nu_0}$ 
defined relative to the unit area volume form $\nu_0$ by a 
factor of $\frac{1}{4\pi}$.}

The structure constants for the Poisson bracket $C\indices{_\alpha_\beta^\gamma}$ are 
defined by $\{Y_\alpha, Y_\beta\} = C\ind{_\alpha_\beta^\gamma} Y_\gamma$.  
The expression for these structure constants is 
\cite{Dowker:1990iy,Dowker:1990ss}
\beq \label{eqn:Cabc}
C\ind{_\alpha_\beta^\gamma} = 
-i (-1)^c (2C+1)
\delta_{[A+B+C]}^1 \;[A]_1 [B]_1 
\tj{A&B&C \\a&b&-c}\tj{A&B&C\\1&-1&0},
\eeq
where we have defined 
\beq \label{eqn:Am}
[A]_m = \sqrt{\frac{(A+m)!}{(A-m)!}},
\eeq
and 
$\delta_{[A]}^n$ is equal to $1$ if $A=n\,\, \mathrm{mod}(2)$ and equal to zero otherwise.
Note that these are nonvanishing only when $A+B+C$ is odd.  These structure constants
can be derived using identities for spin-weighted spherical harmonics, discussed 
in section \ref{sec:swh}.

It is also convenient to introduce a symmetric
bracket constructed from the round sphere metric,
\begin{equation} \label{eqn:sb}
\langle f,g\rangle  = q^{AB}\nabla_A f \nabla_B g.
\end{equation}
Its structure constants $G\ind{_\alpha_\beta^\gamma}$ defined by 
$\langle Y_\alpha, Y_\beta\rangle = G\ind{_\alpha_\beta^\gamma} Y_\gamma$
are given by a similar expression
\begin{equation} \label{eqn:Gabc}
G\ind{_\alpha_\beta^\gamma} = -
(-1)^c (2C+1) \delta_{[A+B+C]}^0 \;[A]_1 [B]_1
\tj{A&B&C\\a&b&-c}\tj{A&B&C\\1&-1&0},
\end{equation}
which are nonvanishing only when $A+B+C$ is even.  
These structure constants have a simple relation to the
product structure constants $E\ind{_\alpha_\beta^\gamma}$ 
arising from the identity, 
\beq \label{eqn:GE}
\int_{S} \nu_0 \nabla^A Y_\alpha \nabla_A Y_\beta Y_\gamma
= -\frac12\int_{S}\nu_0\left(\nabla^2 Y_\alpha Y_\beta Y_\gamma + Y_\alpha \nabla^2 Y_\beta Y_\gamma -
Y_\alpha Y_\beta \nabla^2 Y_\gamma\right),
\eeq
which then implies
\beq
G\ind{_\alpha_\beta^\gamma} = 
\frac12\Big(\la A+ \la B- \la C\Big) E\ind{_\alpha_\beta^\gamma},
\eeq
where 
\beq \label{eqn:A1}
\la A = ([A]_1)^2 = A(A+1),
\eeq
is minus Laplacian eigenvalue on the sphere, i.e.\ $\nabla^2 Y_\alpha = -\la A Y_\alpha$.

\subsection{Spin-weighted spherical harmonics}
\label{sec:swh}
Just as the ordinary spherical harmonics provide a basis with 
respect to which functions on the sphere can be decomposed, 
the spin-weighted spherical harmonics $\sys{s}{\alpha}$ \cite{Goldberg1966}
yield a convenient 
basis for decomposing tensorial objects and differential
operators on the sphere.  They are most easily described by 
introducing the holomorphic coordinate on the sphere,
\beq
z = e^{i\varphi} \cot\frac{\theta}{2},
\eeq
so that the metric is given by
\begin{align}
    ds^2 = \frac{1}{P^2} dz d\bar z,
    \\
    P = \frac12(1+z\bar z).
\end{align}
A complex null basis for the tangent space is provided by 
\begin{align}
    m^A = \sqrt{2}P \partial_z^A, \quad \bar m^A = \sqrt{2}P\partial_{\bar z}^A.
\end{align}
which satisfy
\begin{equation}
    m\cdot m = \bar m \cdot \bar m = 0,\quad m\cdot \bar m = 1.
\end{equation}
The metric and volume form on the sphere are expressed in terms of the holomorphic
basis by
\begin{align}
q_{AB} &= \bar m_A m_B +  m_A \bar m_B \label{eqn:qAB},
\\
\ep_{AB} &= i( \bar m_A m_B -m_A\bar m_B) \label{eqn:epAB}.
\end{align}

A quantity is defined to have  
spin weight $s$ if under the phase rotation $m^a\rightarrow e^{i\psi} m^a$, 
it transforms with a factor of $e^{is\psi}$.  
A general traceless symmetric tensor $T_{A_1\ldots A_n}$ on the sphere has a decomposition in
terms of objects $(T, \bar T)$ of spin weights $(n,-n)$ via
\begin{equation}
    T_{A_1\ldots A_n} = T \bar m_{A_1}\ldots \bar m_{A_n} + \bar T m_{A_1}\ldots m_{A_n}.
\end{equation}
Any function of spin weight $s$ can be decomposed in terms of 
the spin-weighted harmonics $\sys{s}{\alpha}$,\footnote{We use
the notation $\sys{s}{lm}$ instead of the more standard
$\tensor[_s]{Y}{_{lm}}$ for ease of readability.} which form a basis 
for functions of the given spin weight.  Note that the spin-$0$ harmonics 
are simply the usual spherical harmonics discussed in
section \ref{sec:sh}.  Goldberg et.\ al.\ \cite{Goldberg1966} give 
explicit expressions for $\sys{s}{\alpha}$  and a 
detailed discussion of their properties; here, we will simply quote
the relevant properties needed in this work.  Complex conjugation acts via
\beq
({\sys{s}{lm}})^* = (-1)^{l+s} \sys{\mm s}{l,\mm m},
\eeq
and continuing to use the Racah normalization, the integral over the 
sphere of a product is given by
\beq
\int_S \nu_0 \sys{s}{Aa} \sys{\mm s}{Bb} = (-1)^{s+a} \frac{\delta_{A,B}
\delta_{a,-b}}{(2A+1)}.
\eeq

The differential operators $\eth$ and $\bth$ defined in 
\cite{Goldberg1966} act as spin-weight raising and lowering 
operators on $\sys{s}{\alpha}$, whose action is given explicitly by
\begin{align}
\eth \sys{s}{\alpha} &= \frac{[A]_{s+1}}{[A]_s} \,\sys{s+1}{\alpha}, \\
\bth \sys{s}{\alpha} &= - \frac{[A]_s}{[A]_{s-1}} \,\sys{s-1}{\alpha}.
\end{align}
As a consequence, we have that 
these operators satisfy the  relations
\be \label{commr}
[\bth , \eth] \sys{s}{\alpha} = 2s \sys{s}{\alpha},\qquad  (\bth \eth + \eth\bth)\sys{s}{\alpha} =-2[A(A+1)-s^2]\sys{s}{\alpha}.
\ee 
The derivative operators $m^A\nabla_A$ and $\bar m^A \nabla_A$ are closely related
to $\eth$, $\bth$ when acting on totally symmetric traceless tensors, as is seen by 
the following relations:
\begin{align}
    m^A\nabla_A (\sys{s}{\alpha} \bar m^{B_1} \ldots \bar m^{B_s} ) 
    &= \frac{1}{\sqrt{2}}\eth \sys{s}{\alpha} \bar{m}^{B_1}\ldots \bar m^{B_s}, \label{ethdef}
\\
     m^A\nabla_A (\sys{\mm s}{\alpha}  m^{B_1} \ldots  m^{B_s} ) 
    &= \frac{1}{\sqrt{2}} \eth \sys{\mm s}{\alpha} {m}^{B_1}\ldots  m^{B_s},
\\
     \bar m^A\nabla_A (\sys{s}{\alpha} \bar m^{B_1} \ldots \bar m^{B_s} ) 
    &= \frac{1}{\sqrt{2}}\bar \eth \sys{s}{\alpha} \bar{m}^{B_1}\ldots \bar m^{B_s}, 
 \\
     \bar{m}^A\nabla_A (\sys{\mm s}{\alpha}  m^{B_1} \ldots  m^{B_s} ) 
    &= \frac{1}{\sqrt{2}}\bar \eth \sys{\mm s}{\alpha} {m}^{B_1}\ldots  m^{B_s}.
\end{align}
Using these, we can write the gradient and curl of $Y_{lm}$ in terms of spin-weighted
harmonics by
\begin{align}
    q^{CB}\nabla_C Y_{\alpha}  &= \frac{[A]_1}{\sqrt{2}}\left(\sys{1}{\alpha} \bar m^B - \sys{\mm 1}{\alpha} m^B\right),
    \label{eqn:elec} \\
    \epsilon\indices{^C^B} \nabla_C Y_{\alpha} &= 
        \frac{-i[A]_1}{\sqrt{2}} \left(\sys{1}{\alpha} \bar m^B + \sys{\mm 1}{\alpha} m^B\right).
        \label{eqn:mag}
\end{align}
The vectors (\ref{eqn:elec}) and (\ref{eqn:mag}) respectively coincide with the pure-spin electric 
and magnetic vector harmonics, defined in e.g. \cite{Thorne1980},
after normalizing by a factor of $\frac{1}{[A]_1}$. 
This terminology refers to the transformation properties of these 
vectors under parity. 
More general higher order differential operators acting on $Y_\alpha$ 
can be evaluated similarly.  We
define the following operator
\beq
\Delta_{A_1\ldots A_s} = \nabla_{(A_1} \ldots \nabla_{A_s)} - \text{traces}, 
\eeq
which is symmetric and traceless by definition.  For example,
\beq
\Delta_{AB} = \nabla_{(A}\nabla_{B)} - \frac12 q_{AB} \nabla^2.
\eeq
Then the following relation can be shown by inductively applying the above identities
\beq\label{eqn:Deltas}
\Delta_{B_1\ldots B_s} Y_\alpha = \frac{[A]_s}{2^{s/2} } \Big(\sys{s}{\alpha}
\bar m_{B_1} \ldots \bar m_{B_s} +(-1)^s \sys{\mm s}{\alpha} m_{B_1}\ldots
m_{B_s} \Big).
\eeq
The final relation that is useful in obtaining structure constants
for differential operators is the triple integral identity, 
which 
generalizes (\ref{eqn:Eabc}),
\beq\label{eqn:tripint}
    \int_{S} \nu_0 \sys{i}{\alpha} \, \sys{j}{\beta} \, \sys{k}{\gamma} = 
    \begin{pmatrix}
    A & B&C \\ a & b & c
    \end{pmatrix}
    \begin{pmatrix}
    A&B&C \\ \mm i& \mm j & \mm k
    \end{pmatrix}.
\eeq
For example, this equation, along with the gradient and curl expressions
(\ref{eqn:elec}), (\ref{eqn:mag}), provides a straightforward means of evaluating 
the integrals $\int_S \nu_0 \{Y_\alpha, Y_\beta \} Y_\gamma$ and $\int_S\nu_0 
\langle Y_\alpha, Y_\beta\rangle Y_\gamma$ involving the Poisson bracket 
(\ref{eqn:pb}) and symmetric bracket (\ref{eqn:sb}), and this leads directly
to the expressions (\ref{eqn:Cabc}) and (\ref{eqn:Gabc}) for their 
structure constants.

\subsection{Fuzzy spherical harmonics}\label{sec:fuzH}
The fuzzy sphere replaces the algebra of functions on the sphere
by a noncommutative matrix algebra, 
corresponding to the fundamental representation
of the $SU(N)$ Lie algebra.  As with the continuum algebra, these 
matrices decompose into representations of $SU(2)$, and hence 
can be labeled by fuzzy spherical harmonics $\wh{Y}_\alpha$,
with $\alpha = (A,a)$ again denoting the $SU(2)$ representation indices.  
As shown in \cite{Freidel:2001kb}, 
the matrix elements of the fuzzy harmonics can be given explicitly
in terms of a $3j$-symbol according to
\beq\label{eqn:hatYaij}
\big(\wh{Y}_\alpha \big)\ind{_i^j} = \sqrt{N}\, (-1)^{J-j}
\tj{A&J&J\\a&i& \mm j},
\eeq
where $N = 2J+1$.  
The range of the $A$ index is $0\leq A\leq 2J$, since for $A>2J$ the
expression (\ref{eqn:hatYaij}) vanishes, and $J$ can be an integer or 
half integer.  
The fuzzy haronics satisfy the reality condition
$\wh{Y}_{Aa}^\dagger = (-1)^a\, \wh{Y}_{A,-a}$ in direct analogy with the 
continuum harmonics, and are normalized to satisfy
\beq \label{eqn:trrln}
\frac{1}{N} \tr \left(\wh{Y}_\alpha \wh{Y}_\beta\right) = \frac{(-1)^a}{2A+1} \delta_{AB} \delta_{a,-b} = \delta_{\alpha\beta},
\eeq
where the real metric $\delta_{\alpha\beta}$ agrees with the expression 
for the continuum harmonics.  The matrices $\wh{Y}_\alpha$ form a basis
for traceless $N\times N$ matrices, which in general are 
not Hermitian.  However, for each value of $\alpha$, one can 
form the Hermitian combinations $\wh{Y}_\alpha + \wh{Y}^\dagger_\alpha$ and 
$i(\wh{Y}_\alpha - \wh{Y}^\dagger_\alpha)$, just as one would 
form real combinations of the complex continuum harmonics $Y_\alpha$.  
These Hermitian combinations thus provide a matrix version of 
real-valued functions, and since Hermitian matrices generate the 
Lie algebra of $SU(N)$, we see that the matrix regularization of the 
algebra of real functions on the sphere coincides with $\mathfrak{su}(N)$.

The product of two fuzzy harmonics
can be defined via structure constants, $\wh{Y}_\alpha
\wh{Y}_\beta = \wh{M}\ind{_\alpha_\beta^\gamma}\wh{Y}_\gamma$,
explicitly given in terms of the Wigner $6j$ symbol  \cite{NIST3j} 
by \cite{Freidel:2001kb, Alekseev:1999bs}
\beq
\wh{M}\ind{_\alpha_\beta^\gamma} = \sqrt{N} (2C+1) (-1)^{2J +c}\tj{A&B&C\\a&b&-c} \sj{A&B&C\\J&J&J},
\eeq
or more symmetrically with the $\gamma$ index lowered using the metric (\ref{eqn:trrln})
as 
\beq \label{eqn:Mabc}
\wh{M}_{\alpha\beta\gamma} 
=\frac{1}{N} \tr\left(\wh{Y}_\alpha \wh{Y}_\beta \wh{Y}_\gamma\right)
= \frac{\sqrt{N}}{(-1)^{2J}} \tj{A&B&C\\a&b&c} \sj{A&B&C\\J&J&J}.
\eeq
It is convenient to define a deformed $3j$ symbol 
\be \label{eqn:def3j}
\tjn{A&B&C\\0&0&0}_N:=  \frac{ \sqrt{N}}{(-1)^{2J}} \sj{A&B&C\\J&J&J},
\ee 
so that the structure constants take the form
\beq
\wh{M} _{\alpha\beta\gamma} = \tj{A&B&C\\a&b&c}\tjn{A&B&C\\0&0&0}_N,
\eeq
directly analogous to the continuum equation (\ref{eqn:Eabc}).
We can further decompose these structure constants into their 
symmetric $\wh{E}_{\alpha\beta\gamma}$ and antisymmetric
$\wh{C}_{\alpha\beta\gamma}$
pieces on $\alpha$ and $\beta$,
\beq \label{eqn:MEC}
\wh{M}_{\alpha\beta\gamma} = \wh{E}_{\alpha\beta\gamma} + \frac{i}{N}\wh{C}_{\alpha\beta\gamma},
\eeq
and we will see below that as $N\rightarrow\infty$, $\wh{E}_{\alpha\beta\gamma}$ and
$\wh{C}_{\alpha\beta\gamma}$ approach their classical counterparts, $E_{\alpha\beta\gamma}$
and $C_{\alpha\beta\gamma}$, defined in section \ref{sec:sh}.

The large-$N$ expansion of the structure constants $\wh{M}_{\alpha\beta\gamma}$
can be obtained by employing a remarkable identity by Nomura \cite[Eq.  (2.22)]{nomura1989description} 
that expresses the $6j$ symbol
as a single sum in which each term involves a single $3j$ symbol.\footnote{Note that Nomura
\cite{nomura1989description} uses a nonstandard
normalization for the $6j$ symbol, and with the standard normalization
\cite{NIST3j}, the factor of $(2e +1)^{-\frac12}$ that appears in Nomura's equation (2.22)
should be left out.}  Applied to the deformed $3j$ symbol
(\ref{eqn:def3j}), this identity yields 
\begin{align}
\tjn{A&B&C\\0&0&0}_N =&\left[\frac{\theta_N(A)\theta_N(B)}{\theta_N(C)}\right]^{\frac12}  \times \sum_{m=0}^{\min(A,B)}\frac{N!}{(N+m)! m!} 
[A]_m [B]_m
\tj{A&B&C\\m&-m&0}, \label{eqn:6jexp}
\end{align}
where we have made the definitions
\be
\theta_N(A) :=  \frac{(N+A)!(N-A-1)!}{N!(N-1)!}=\frac{(N+A)\ldots (N+1)}{ (N-A)\ldots (N-1)},
\ee
and $[A]_m$ is defined in (\ref{eqn:Am}).
The identity (\ref{eqn:6jexp}) is valid assuming $B$ is an integer, and holds
for $J$ either integer or half integer.  
Each term in the sum (\ref{eqn:6jexp}) is suppressed by an additional factor of 
$\frac{1}{N}$, and hence this sum manifestly yields the large $N$ expansion
of the matrix product of fuzzy spherical harmonics.  The prefactor to the sum has the following expansion
at large $N$,
\beq
1+\frac{1}{2N}\left(\la A+ \la B- \la C  \right)
+\frac{1}{8N^2}\left(\la A+ \la B- \la C\right)^2+\ldots ,\label{eqn:prefac}
\eeq
The leading order
term in the deformed $3j$ symbol expansion is then seen to be 
\beq
\tjn{A&B&C\\0&0&0}_N = \tj{A&B&C\\0&0&0}
+\mathcal{O}\left(1/N\right).
\eeq
Substituting this expression into the structure constants
(\ref{eqn:Mabc}), we see that the leading order piece 
$\wh{M}^{(0)}_{\alpha\beta\gamma}$ coincides exactly with 
the continuum commutative product structure constants
$E_{\alpha\beta\gamma}$ (\ref{eqn:Eabc}).  

Equation (\ref{eqn:def3j}) straightforwardly yields 
the first 
subleading correction to the deformed symbol $3j$ symbol, 
\beq\label{eqn:6j32}
\frac{1}{N} \left[\frac12\left(\la A+\la B-\la C \right)
\tj{A&B&C\\0&0&0}+ [A]_1 [B]_1
\tj{A&B&C\\1&-1&0} \right],
    \eeq
Using the relations (\ref{eqn:GE}) and (\ref{eqn:Gabc})
between the $E_{\alpha\beta\gamma}$ and $G_{\alpha\beta\gamma}$ 
structure constants, we see that the terms with $A+B+C$ 
even cancel in (\ref{eqn:6j32}), leaving 
only the odd piece,
\beq
\frac{1}{N}\delta^1_{[A+B+C]}\;
[A]_1[B]_1 \tj{A&B&C\\1&-1&0}.
\eeq
Comparing to (\ref{eqn:Cabc}), 
this  determines the first-order correction to 
the structure constants (\ref{eqn:Mabc}) in terms of the 
continuum Poisson bracket structure 
constants $C_{\alpha\beta\gamma}$
\beq\label{MCeq}
\wh{M}^{(1)}_{\alpha\beta\gamma} = 
\frac{i}{N} C_{\alpha\beta\gamma}.
\eeq
This verifies that the matrix product of the fuzzy harmonics $\wh{Y}_\alpha$ 
takes the desired form of a valid $\star$-product in the sense of deformation
quantization of the algebra of functions on the sphere
(see e.g.\ \cite{Gutt2011, Bayen1975, Sternheimer1998}); namely, it has the expansion
\beq
\wh{Y}_\alpha \cdot \wh{Y}_\beta = \widehat{Y_\alpha Y_\beta} + \frac{i\hbar}{2}
\widehat{\{Y_\alpha, Y_\beta \} } + \op(\hbar^2),
\eeq
with $\hbar = \frac{2}{N}$.  In fact, we can fix the value of $\hbar$ more 
precisely by recalling that Poisson brackets of the $l=1$ continuum harmonics generate
an $\mathfrak{su}(2)$ algebra, and by requiring that the matrix commutator
exactly reproduce this algebra in the sense 
\beq
[\wh{Y}_{1,m}, \wh{Y}_{1,m'}] = i\hbar\, \widehat{\left\{Y_{1,m}, Y_{1,m'}\right\} },
\eeq
the value of $\hbar$ is determined to be 
\beq
\hbar = \frac{1}{[J]_1} = \frac{2}{\sqrt{N^2-1}} = \frac{2}{N} + \op(N^{-3}),
\eeq
which can be derived by evaluating the exact structure constants for
the matrix product (\ref{eqn:Mabc}) in terms of the $6j$ symbol
\beq\label{eqn:111JJJ}
\sj{1&1&1\\J&J&J} = \frac{(-1)^{2J}}{\sqrt{N}} \frac{\sqrt{6}}{3\sqrt{N^2-1}}.
\eeq
Since the Poisson bracket $\{,\}$ is defined with respect to a 
unit radius sphere with area $A = 4\pi$, the result 
$\hbar = \frac{2}{N}$ is consistent with the standard relation 
$N = \frac{A}{2\pi \hbar}$
between the dimension $N$ of the quantum Hilbert space and the volume
$A$ of the classical phase space.

As an aside, we note that the definition of the deformed $3j$ symbol
(\ref{eqn:def3j}) can be extended to
nonzero magnetic quantum numbers by the equation 
\beq \label{eqn:def3jmag}
\tjn{A&B&C\\i&j&k}_N = \frac{\sqrt{N}}{(-1)^{2J}}\sj{A&B&C\\J-k&J&J+i},
\eeq
where $i+j+k=0$. The Nomura identity \cite{nomura1989description} 
in this case yields the expression
\begin{align}
\tjn{A&B&C\\i&j&k}_N =&\;
\left[\frac{N(N+i-A-1)!(N+A+i)!(N-k+i-B-1)!(N-k+i+B)!}{\big((N+2i)!\big)^2(N-k-C-1)!(N-k+C)!}  \right]^{\frac12} \nonumber \\
& \times \sum_{m=0}^{\min\left(A-i,\, B+j\right)} \frac{(N+2i)!}{(N+2i+m)! m!}\frac{[A]_{i+m}\,[B]_{j}}{[B]_{j-m}\,[A]_i } \tj{A&B&C\\i+m&j-m&k}.
\label{eqn:3jNexp}
\end{align}
The prefactor in this expression approaches $1$ as $N\rightarrow \infty$, and 
hence at leading order the deformed $3j$ symbol approaches the usual $3j$ symbol,
\beq
\tjn{A&B&C\\i&j&k}_N = \tj{A&B&C\\i&j&k} + \op(N^{-1}),
\eeq
which is equivalent to a known asymptotic formula for the $6j$ symbol in terms 
of a $3j$ symbol \cite{Ponzano1968}, although equation
(\ref{eqn:3jNexp}) additionally produces all subleading corrections 
to this asymptotic formula.  

The motivation for the definition
(\ref{eqn:def3jmag}) lies in a product relation for a fuzzy version
of spin-weighted spherical harmonics $\wh{Y}_\alpha^s$, which can
be defined as rectangular matrices whose row and column dimension
differ by the spin weight $s$,
\beq
[\wh{Y}_\alpha^s ] \ind{_i^j} = \sqrt{N} (-1)^{J-j}\tj{A&J+s& J\\a&i&\mm j}.
\eeq
These matrices can be multiplied by appropriately adjusting 
the value of $J$ to ensure that the number of columns 
of the first matrix
matches the number of rows of the second.  The deformed 
$3j$ symbol then appears in the structure constants
for this matrix multiplication, which,
similar to equation (\ref{eqn:Mabc}), can be characterized 
by a trace of a triple product, 
\beq
\wh{M}^{ijk}_{\alpha\beta\gamma} =
\frac1N \tr\left(\wh{Y}_\alpha^i \wh{Y}_\beta^j \wh{Y}_\gamma^k\right)
=\tj{A&B&C\\a&b&c}\tjn{A&B&C\\\mm i&\mm j&\mm k}_N.
\eeq
This equation is the fuzzy analog of the continuum triple integral
expression (\ref{eqn:tripint}).
The above proposal  for a fuzzy version of the spin-weighted harmonics
has not been considered previously, and may provide some hints at determining
a deformation of the full diffeomorphism algebra of the sphere $\diff$.
We leave investigation into this idea to future work.

\subsection{Expansion of the matrix product} \label{sec:mpexpansion}

As mentioned above, the identity (\ref{eqn:6jexp}) provides a means of 
expanding the matrix product of the fuzzy harmonics to higher order
in  $\frac1N$.  This can be used to show that the matrix 
product takes the form of a valid $\star$-product.  Such a product is 
a deformation of the commutative product of functions of the sphere
that admits a formal expansion in powers of $\hbar$ of the form
\cite{Gutt2011, Bayen1975, Sternheimer1998}
\beq\label{stardef}
f\star g = fg +\sum_{n=1}^\infty \left(\frac{i\hbar}{2}\right)^n
\frac{1}{n!}C^{(n)}(f,g),
\eeq
with 
\beq
C^{(1)}(f,g) = \{f, g\},
\eeq
and with each higher-order term $C^{(n)}(f,g)$ given by a bidifferential operator
of order at most $n$, whose highest order piece takes the form expected from
a Moyal product,
\beq \label{eqn:moyalform}
C^{(n)}(f,g) = \ep^{A_1 B_1} \ldots \ep^{A_n B_n} 
(\nabla_{A_1}\ldots \nabla_{A_n} f) (\nabla_{B_1}\ldots \nabla_{B_n} g)
+ B^{(n)}(f,g),
\eeq
with the differential order of $B^{(n)}(f,g)$ strictly less than $n$.  
As an application of the utility of the formula (\ref{eqn:6jexp}), we demonstrate
here that the $\op(N^{-2})$ term in the structure constants
for the matrix product (\ref{eqn:Mabc}) is precisely of this form.  From (\ref{eqn:6jexp}) and
(\ref{eqn:prefac}), the $\op(N^{-2})$ term in the  deformed $3j$ symbol 
is
\begin{equation}
    \begin{aligned}
\frac{1}{N^{2}}\left[\frac{\left([A]_1^2+[B]_1^2-[C]_1^2 \right)^2}{8}
\tj{A&B&C\\0&0&0}
+\right. &\,\left.\frac{[A]_1^2+[B]_1^2-[C]_1^2-2}{2}
[A]_1[B]_1\tj{A&B&C\\1&-1&0}\right.
\\
+&\, \left. \frac12 [A]_2[B]_2
\tj{A&B&C\\2&-2&0}\right].
\end{aligned}
\end{equation}
This expression simplifies using a recursion identity
for the $3j$ symbols \cite{Raynal1979, Raynal1993},
\begin{align}
[A]_2[B]_2
\tj{A&B&C\\2&-2&0} =&\;
-([A]_1^2+[B]_1^2-[C]_1^2-2)[A]_1[B]_1
\tj{A&B&C\\1&-1&0} %
 -[A]_1^2 [B]_1^2\tj{A&B&C\\0&0&0},
 \label{eqn:3jrecursion}
\end{align}
to give
\beq \label{eqn:3jn2nd}
\frac{1}{N^{2}} \left[\frac{[A]_1^4+[B]_1^4+[C]_1^4-2[A]_1^2[B]_1^2
-2[A]_1^2[C]_1^2-2[B]_1^2[C]_1^2}{8}\right]\tj{A&B&C\\0&0&0},
\eeq
which is notably totally symmetric in $A, B, C$, and 
 only nonzero for $A+B+C$ even.  

We now look for the second order bidifferential operator $C^{(2)}(f,g)$
that yields the expression (\ref{eqn:3jn2nd}) when acting on the continuum
harmonics.  The structure constants for $C^{(2)}(\cdot, \cdot)$
are defined by 
\beq
C^{(2)}_{\alpha\beta\gamma} = \int_S \nu_0 \, C^{(2)}(Y_\alpha, Y_\beta)
Y_\gamma.
\eeq
Since we expect the highest order term in this operator to take the 
Moyal product form as in equation (\ref{eqn:moyalform}), we begin
by evaluating the structure constants 
\beq\label{eqn:Pi}
\Pi^{(2)}_{\alpha\beta\gamma} = \int_{S}\nu_0\, \ep^{AB}\ep^{CD}
(\nabla_A\nabla_C Y_\alpha)( \nabla_B\nabla_D Y_\beta)  Y_\gamma.
\eeq
First using $\ep^{AB}\ep^{CD} = q^{AC} q^{BD}-q^{AD}q^{BC}$ and 
$\nabla_{[A}\nabla_{B]} Y_\beta = 0$, we have that 

    \begin{align}
\ep^{AB}\ep^{CD}
(\nabla_A\nabla_C Y_\alpha)( \nabla_B\nabla_D Y_\beta)
&= \nabla^2 Y_\alpha \nabla^2 Y_\beta
-(\nabla^A \nabla^B Y_\alpha)(\nabla_{(A} \nabla_{B)} Y_\beta)
\nonumber \\
&= 
\nabla^2 Y_\alpha \nabla^2 Y_\beta
-(\Delta^{AB}+\frac12q^{AB}\nabla^2) Y_\alpha (\Delta_{AB} +\frac12 q_{AB}\nabla^2)Y_\beta
\nonumber \\
&= 
\frac12\nabla^2 Y_\alpha \nabla^2 Y_\beta - (\Delta^{AB}Y_\alpha)(\Delta_{AB}Y_\beta).
\end{align}

We can expand the second term in spin-weighted harmonics using 
(\ref{eqn:Deltas}) to obtain
\begin{align}
-(\Delta^{AB} Y_\alpha) (\Delta_{AB}Y_\beta)
&=-\frac{[A]_2 [B]_2 }{4} \left(\sys{2}{\alpha}\bar m^A \bar m^B + \sys{\mm 2}{\alpha} m^A m^B\right)
\left(\sys{2}{\beta}\bar m_{A}\bar m_B + \sys{\mm 2}{\beta}m_A m_B\right) \nonumber 
\\
&= -\frac{[A]_2 [B]_2 }{4}(\sys{2}{\alpha} \sys{\mm 2}{\beta} + \sys{\mm 2}{\alpha}\sys{2}{\beta}).
\end{align}
The contribution of this term to the structure constants 
(\ref{eqn:Pi}) then follows directly from 
the triple integral identity (\ref{eqn:tripint}),
\begin{align}
-\frac{[A]_2 [B]_2}{4}  \int_S \nu_0 \left(\sys{2}{\alpha}\sys{\mm 2}{\beta} Y_\gamma
+ \sys{\mm 2}{\alpha} \sys{2}{\beta} Y_\gamma\right)
&=
-\frac{[A]_2 [B]_2}{4}\tj{A&B&C\\a&b&c}
\left[\tj{A&B&C\\ \mm 2 & 2&0} + \tj{A&B&C\\ 2 &\mm 2&0} \right]
\nonumber\\
&=
-\frac{[A]_2 [B]_2}{2} \delta_{[A+B+C]}^0\tj{A&B&C\\a&b&c}\tj{A&B&C\\ 2&\mm 2&0},
\end{align}
where we recall that $\delta^0_{[A+B+C]}$ 
is $1$ if $A+B+C$ is even, and $0$ otherwise. The other contribution
to (\ref{eqn:Pi}) is 
\beq
\frac12\int_S \nu_0\, \nabla^2 Y_\alpha \nabla^2 Y_\beta Y_\gamma 
= \frac{[A]_1^2 [B]_1^2}{2}\tj{A&B&C\\a&b&c}\tj{A&B&C\\0&0&0},
\eeq
So we find the structure constants
\beq\label{eqn:Pi2init}
\Pi^{(2)}_{\alpha\beta\gamma} = \frac{\delta^0_{[A+B+C]} }{2}\tj{A&B&C\\a&b&c} 
\left[[A]_1^2[B]_1^2\tj{A&B&C\\0&0&0} - [A]_2 [B]_2\tj{A&B&C\\2&\mm 2&0} 
\right].
\eeq
This can be simplified using the $3j$ symbol
recursion identities (\ref{eqn:3jrecursion}) and 
\cite{Raynal1979, Raynal1993}
\begin{align}
[A]_1[B]_1 &\tj{A& B&C\\1&\mm 1&0}  = -\frac12\big([A]_1^2+[B]_1^2-[C]_1^2\big)\tj{A&B&C\\0&0&0},
\end{align}
valid for even $A+B+C$,
which reduces the bracketed term in (\ref{eqn:Pi2init})  to 
\beq
-\frac{\Big([A]_1^4+[B]_1^4+[C]_1^4
-2[A]_1^2[B]_1^2-2[A]_1^2[C]_1^2-2 [B]_1^2[C]_1^2\Big)}{2}
\tj{A&B&C\\0&0&0} -2[A]_1[B]_1\tj{A&B&C\\1&\mm 1&0}.
\eeq

The first term matches the expression (\ref{eqn:3jn2nd}) appearing at second
order in the large $N$ expansion of the matrix product structure constants.
The remaining term is a correction
that appears in the structure constant $G_{\alpha\beta\gamma}$
for the symmetric bracket
(\ref{eqn:sb}).
This then shows that 
\beq
C^{(2)}_{\alpha\beta\gamma} = \Pi^{(2)}_{\alpha\beta\gamma} - G_{\alpha\beta\gamma},
\eeq 
or equivalently, that the bidifferential operator $C^{(2)}(f,g)$ is given by
\beq\label{eqn:C2diffop}
C^{(2)}(f,g) = \ep^{A_1 B_1} \ep^{A_2 B_2} (\nabla_{A_1} \nabla_{A_2} f)
(\nabla_{B_1} \nabla_{B_2} g) - q^{AB}\nabla_A f \nabla_B g,
\eeq
where the second term in this expression corresponds to 
the term $B^{(2)}(f,g)$ in the general expression (\ref{eqn:moyalform})
the expansion of the $\star$-product. One noticeable property of $C^{(2)}$ is that it is a symmetric bidifferential operator. As we show in the next section we also have that $C^{(3)}$  is a skew-symmetric bidifferential operator.   This means that the first correction to the commutator is of order $\hbar^2$: 
$[f,g] = i\hbar\left(C^{(1)}(f,g) - \frac14\hbar^2 C^{(3)}(f,g)
 + \mcal{O}(\hbar^3) \right)
$

The appearance of this correction $B^{(2)}(f,g)$
to the naive Moyal product at 
$\op(\hbar^2)$ deserves some attention.  When applying the 
procedure of Fedosov quantization to construct an
associative $\star$-product on a 
symplectic manifold, one generically finds nontrivial $B^{(n)}(f,g)$ terms 
that account for effects coming from the curvature of a chosen symplectic 
connection $\nabla_A$ \cite{Fedosov1994}.  
In the simplest application of the Fedosov
construction, however, such curvature corrections only occur at 
$\op(\hbar^3)$ or higher, whereas the fuzzy matrix product generates 
such a correction at $\op(\hbar^2)$.  Nevertheless, there is no inconsistency
in finding such terms at $\op(\hbar^2)$, since the Fedosov procedure 
contains certain gauge ambiguities that affect the precise expression
for the $\star$-product, and these ambiguities can affect the 
$\op(\hbar^2)$ terms \cite{Fedosov1996}.  
The correction appearing in (\ref{eqn:C2diffop})
can arise in two different ways.  The first is 
as an ambiguity in how one constructs a flat connection on the 
Weyl bundle of the symplectic manifold, which 
is not uniquely determined even after specifying 
a symplectic connection.  The second way it can appear simply comes from 
the standard ambiguity in the quantization map sending a classical 
function $Y_\alpha$ to its quantum operator $\wh{Y}_\alpha$.  
In general, one is free to correct this map at higher order in $\hbar$,
and a shift of the form $\wh{Y}_\alpha \rightarrow \wh{Y}_\alpha 
+ \lambda \hbar^2\left(\widehat{\nabla^2Y_\alpha}\right) + \ldots$ can generate
corrections at $\op(\hbar^2)$ in the $\star$-product as were found above.  

In the present context, the appearance of a nontrivial 
$B^{(2)}(\cdot, \cdot)$ at $\op(\hbar^2)$ in the $\star$-product ensures the 
desirable property that the second-order structure constants
$C^{(2)}_{\alpha\beta\gamma}$ are totally symmetric in the indices 
$\alpha$, $\beta$, $\gamma$. 
This symmetry follows from the permutation 
symmetry of the columns of the $6j$ symbol that appears in the 
fully non-perturbative structure constants $\wh{M}_{\alpha\beta\gamma}$ 
for the matrix product (\ref{eqn:Mabc}).  It would be interesting to 
investigate in future work 
whether there is some deeper meaning to this correction that 
appears in the $\star$-product.

\subsection{Parity of the matrix product}
\label{app:parity}
An interesting feature exhibited by the matrix product structure 
constants $\wh{M}_{\alpha\beta\gamma}$ is that the lowest order
term in the large $N$ expansion is nonzero only when $A+B+C$ is 
even, the $\mathcal{O}(N^{-1})$ term is nonvanishing 
only for $A+B+C$ odd, and, as calculated 
in section \ref{sec:mpexpansion}, the $\mathcal{O}(N^{-2})$ term is again
nonzero only for $A+B+C$ even.  Given the expression
(\ref{eqn:Mabc}) for the structure constants, this translates to the 
statement that the $\mathcal{O}(N^0)$ and $\mathcal{O}(N^{-2})$ 
terms in $\wh{M}_{\alpha\beta\gamma}$ are totally symmetric 
tensors, while the $\mathcal{O}(N^{-1})$ term is totally
antisymmetric.  This conclusion follows from the fact that
the $6j$-symbol is totally symmetric under permutations of 
its columns, while the $3j$-symbol satisfies 
$\tj{B&A&C\\b&a&c} = (-1)^{A+B+C}\tj{A&B&C\\a&b&c}$.  Here we will
show that this pattern persists to all orders in the $\frac{1}{N}$
expansion, namely, that only even powers of $N^{-1}$ appear in 
$\wh{M}_{\alpha\beta\gamma}$ when $A+B+C$ is even, and only
odd powers of $N^{-1}$ appear when $A+B+C$ is odd.  
Given the decomposition (\ref{eqn:MEC}) of $\wh{M}_{\alpha\beta\gamma}$
into its symmetric and antisymmetric parts, this statement then 
implies that $\wh{E}_{\alpha\beta\gamma}$ and 
$\wh{C}_{\alpha\beta\gamma}$ both admit large $N$ expansions involving
only even powers of $N^{-1}$.  

We will say that the rescaled $6j$ symbol 
$\tjn{A&B&C\\0&0&0}_N = \frac{\sqrt{N}}{(-1)^{2J}}\sj{A&B&C\\J&J&J}$
with $A,B,C$ integers
satisfies $N$-parity if its expansion in $N^{-1}$ involves
only powers with the same parity as $A+B+C$.  
To prove the claim that 
$\tjn{A&B&C\\0&0&0}_N$ satisfies $N$-parity, 
we begin by noting that as a base case,  $\tjn{1&1&1\\0&0&0}_N$
has an expansion involving only odd powers of $N^{-1}$, as is 
apparent from its exact expression obtained from (\ref{eqn:111JJJ}).
Similarly,
from the exact expressions
\begin{align}
\tjn{0&0&0\\0&0&0}_N = 1,  \qquad\quad
\tjn{0&1&1\\0&0&0}_N =\frac{-1}{\sqrt{3}},
\end{align}
we see that for these lowest values for which $A+B+C$ is even,
only $N^0$-terms appear, and hence $N$-parity is satisfied.  
Since $\tjn{0&0&1\\0&0&0}_N=0$,
this base case for $A+B+C$ odd also trivially satisfies
$N$-parity.  

To proceed with an inductive proof to higher values of $A,B,C$,
we apply the following recursion relation for $6j$-symbols
\cite{Schulten1975} (see \cite{Bonzom:2011hm} for a geometrical interpretation of this identity in quantum gravity)
\beq \label{eqn:6jrecurs}
A\, E_N(A+1)\sj{A+1&B&C\\J&J&J} = -F(A)\sj{A&B&C\\J&J&J}
-(A+1)E_N(A)\sj{A-1&B&C\\J&J&J}, 
\eeq
where
\beq
E_N(A) =  N\sqrt{1-\frac{A^2}{N^2}}\; A\sqrt{[A^2-(B-C)^2][(B+C+1)^2-A^2]},
\eeq
admits an expansion in odd powers of $N^{-1}$, and 
\beq
F(A) = (2A+1) \la A(\la B +\la C - \la A),
\eeq
is independent of $N$.  Since $J$ (and hence $N$) is fixed,
the recursion relation (\ref{eqn:6jrecurs}) also applies to the rescaled $6j$-symbols
$\tjn{A&B&C\\0&0&0}_N$.

Now, assuming we have shown that $\tjn{A&B&C\\0&0&0}_N$ satisfies
$N$-parity for all $A\leq K$, with $B$, $C$ fixed, the 
recursion identity (\ref{eqn:6jrecurs}) implies that
\beq
\tjn{K+1&B&C\\0&0&0}_N
=
-\frac{F(K)}{KE_N(K+1)}\tjn{K&B&C\\0&0&0}_N 
-\frac{(K+1)E_N(K)}{KE_N(K+1)}\tjn{K-1&B&C\\0&0&0}_N.
\eeq
Since $\frac{F(K)}{KE_N(K+1)}$ involves only odd powers of $N^{-1}$ in
its expansion, the first term on the right-hand side above
will have an $N^{-1}$ expansion with powers of the opposite 
parity of the expansion of $\tjn{K&B&C\\0&0&0}_N$.
Similarly, the expansion of $\frac{(K+1)E_N(K)}{KE_N(K+1)}$ 
involves only even powers of $N^{-1}$, and so the second term on the 
right-hand side will have an $N^{-1}$ expansion with powers of $N$
with the same parity as the expansion of $\tjn{K-1&B&C\\0&0&0}_N$.
Hence, the $N^{-1}$ expansion of both terms on the right-hand side
in the above relation only involves powers of $N^{-1}$ with the 
same parity as $K+B+C-1$, which is the same as $K+B+C+1$.  
We therefore see that $\tjn{K+1&B&C\\0&0&0}_N$ satisfies
$N$-parity, proving the inductive step.  Due to the fact that 
$\tjn{A&B&C\\0&0&0}_N$ is totally symmetric under permutations
of its columns, the same inductive argument applies to
$\tjn{A&K+1&C\\0&0&0}_N$ and $\tjn{A&B&K+1\\0&0&0}_N$, and 
 we can conclude that $\tjn{A&B&C\\0&0&0}_N$ satisfies
$N$-parity for all nonnegative integers $A,B,C$.

\subsection{Star product and Nomura identity}
\label{app:star}
In this section, we demonstrate that the product arising from the 
$6j$-symbol (\ref{eqn:Mabc}) can be viewed as a valid star product 
to all orders in $\frac{1}{N}$, and further show that the Nomura
identity \cite{nomura1989description} arises precisely from the 
$\hbar$ expansion of this star product.  
The star product on the sphere can be induced from a rotationally-invariant
star product on 
$\mathbb{R}^3$ via the natural embedding of the sphere in this space. 
In Cartesian coordinates $x^a$, $a = 1,2,3$, this star product is given by \cite{Presnajder199912,HayasakaNakayamaTakaya200209,AlekseevLachowska,MatsubaraStenmark200402}
\begin{equation}
    f\star g=fg+\sum_{n=1}^\infty C_n\left(\frac{\hbar}{r}\right)J^{a_1b_1}\ldots J^{a_nb_n}\,\pa_{a_1}\ldots\pa_{a_n}f\pa_{b_1}\ldots\pa_{b_n}g,
\end{equation}
where $r=\sqrt{x_1^2+x_2^2+x_3^2}$ is the sphere radius and 
\begin{equation} \label{eqn:cnhbarr}
    \begin{gathered}
    J^{ab}(x)\equiv r^2\delta^{ab}-x^ax^b+i\,r\varepsilon\indices{^{ab}_c}x^c,
    \\
    C_n\left(\frac{\hbar}{r}\right)\equiv \frac{\left(\frac{\hbar}{r}\right)^n}{n!\left(1-\frac{\hbar}{r}\right)\ldots\left(1-(n-1)\frac{\hbar}{r}\right)}.
    \end{gathered}
\end{equation}
$J^{ab}$ is covariant under rotations and therefore  the $\star$-product is rotationally-invariant. Moreover, since $J^{ab} x_b=0$ we have that 
$f\star r=r\star f=rf$, and hence it  can be restricted to the sphere. However, we would like to know the expression in terms of intrinsic coordinates on the sphere and not the above embedding Euclidean coordinates. Denoting the restriction $J|_{S}$ by $J$ obtained by the embedding $\iota:S\hookrightarrow \mbb{R}^3$, we have
\begin{equation}\label{eq:the relation between tensors J on R3 and sphere}
    J^{AB}=\frac{\partial \sigma^A}{\partial x^c}\frac{\partial \sigma^B}{\partial x^d}J^{cd}. 
\end{equation}
where $\sigma^A=\sigma^A(x^1,x^2,x^3)$ are a set of coordinates on sphere. Since we are sitting on a sphere, $r$ is a constant, which we could set to one. We however keep the radius as $r$ to make the formulas general. 

We would like to write the star product in the holomorphic polarization. Celestial coordinates provide an appropriate means for doing so. We use the relation between celestial coordinates $(z,\bar{z})$ and Euclidean coordinates $(x^1,x^2,x^3)$ in the north-pole patch $S-\{(0,0,r)\}$  given by
\begin{equation}
    \begin{gathered}
    x^1=\frac{\left(z+\bar{z}\right)}{1+\bar{z}z}r, \qquad x^2=\frac{-i(z-\bar{z})}{1+\bar{z}z}r, \qquad x^3=\frac{-1+\bar{z}z}{1+\bar{z}z}r,
    \\
    z=\frac{x^1+ix^2}{r-x^3},\qquad \bar{z}=\frac{x^1-ix^2}{r-x^3}.
    \end{gathered}
\end{equation}
Using \eqref{eq:the relation between tensors J on R3 and sphere}, we find
\begin{equation}
    {\begin{aligned}
    J^{zz}&=0, &\qquad J^{z\bar{z}}&=0,
    \\
    J^{\bar{z}z}&=(1+\bar{z}z)^2, &\qquad J^{\bar{z}\bar{z}}&=0.
    \end{aligned}}
\end{equation}
Similarly, the relations between celestial coordinates $(z,\bar{z})$ and Euclidean coordinates $(x^1,x^2,x^3)$ in the south-pole patch $S-\{(0,0,-r)\}$ are  
\begin{equation}
    \begin{gathered}
    z=\frac{x^1+ix^2}{r+x^3},\qquad \bar{z}=\frac{x^1+ix^2}{r+x^3}.
    \end{gathered}
\end{equation}
which by using \eqref{eq:the relation between tensors J on R3 and sphere} gives
\begin{equation}
    {\begin{aligned}
    J^{zz}&=0, &\qquad J^{z\bar{z}}&=(1+\bar{z}z)^2,
    \\
    J^{\bar{z}z}&=0, &\qquad J^{\bar{z}\bar{z}}&=0.
    \end{aligned}}
\end{equation}
Hence, the sphere star product in the holomorphic polarization is parameterized by a weight $\lambda=- \frac{r}{\hbar}$ where $r$ is the sphere radius and $\hbar$ is the formal deformation parameter is known to all orders in perturbation theory. It is explicitly given by 
\be
f \star_H g = \sum_{n=0}^\infty \frac{ \hbar^n}{n!}
C_H^{(n)}(f,g), \qquad 
\ee 
Here the label $H$ stands for Holomorphic.
 The  holomorphic deformation cocycles are 
\be \label{CH}
C^H_n(f,g): = \frac{1}{\prod_{p=0}^{n-1} (r- p\hbar)}  J^{A_1B_1} \ldots J^{A_nB_n} ( \nabla_{A_1}\ldots \nabla_{A_n} f) ({\nabla}_{B_1}\ldots \nabla_{B_n} g).
\ee
Here $J^{AB}= q^{AB}-i\epsilon^{AB}$ is the standard Hermitian form
on the sphere. Expanding the expression \eqref{CH} in powers of $\hbar$ gives a relation between the $\{C_H^{(n)}\}$ and the $\{C^{(m)}\}$ introduced in \eqref{stardef}.

It can be expressed in terms of the null complex frame field $m^A$ introduced earlier.
\be 
J^{AB}= 2\bar{m}^A {m}^B.
\ee 
This and the definition of the spin-raising differential  operator $\eth$
given in \eqref{ethdef} means that we can write the holomorphic star product parametrized by the weight $\lambda:=-\frac{r}{\hbar}$ more concisely as 
\be 
f\star_H g =\sum_{n=0}^\infty \frac{(-1)^n }{n!} \frac{(\bar{\eth}^n f)\, (\eth^n g)   }{\prod_{p=0}^{n-1} (\lambda+p)}.
\ee 
Using the commutation relations \eqref{commr} we see that  the operators $X=\eth$, $Y=-\bar\eth$, $H=2\wh{s}$, where $\wh{s}$ is the operator that measures the spin of the observable, form an $\mathfrak{su}(2)$ algebra
\be
[H,X]=2X, \qquad [H,Y]=-2Y,\qquad [X,Y]=H.
\ee 
It is usually convenient to formalize the construction of the star  product as resulting from the composition of the multiplication operator of functions $m: C(S)\times C(S) \to C(S)$ with the deformation operator 
${\cal F}:C(S)\times C(S)\to C(S)\times C(S)$. The deformation operator encodes the non-triviality of the star product.
The star product can therefore be written in an algebraic form  as
$F \star_H G = m [{\cal F}_H( F\otimes G)]$ where 
\begin{eqaligned}
    {\cal F}_H&= \sum_{n=0}^\infty \frac{1 }{n!} \frac{Y^n\otimes X^n   }{\prod_{p=0}^{n-1} (\lambda+p)}
    \\
    &= 1 + \lambda^{-1} Y\otimes X + \frac{\lambda^{-2}}2 Y^2\otimes X^2
    + \frac{\lambda^{-3}}6( Y^3\otimes X^3 - 3 Y^2\otimes X^2)+\mcal{O}(\lambda^{-4}).
\end{eqaligned}
In this representation, the weight $\lambda$ is the eigenvalue of $-H$ on the SL$(2)$ module representing the sphere sections \cite[Equation (6)]{AlekseevLachowska}.

Other equivalent star-product which are covariant under the rotation group can be obtained after a reparameterization $F\to \theta[F]$ where $\theta= 1 + \sum_{n=1}^\infty \frac{\hbar^n}{n!} \theta_n$ and $\theta_n$ is a differential operator of order $n$ which is invariant under rotation\footnote{This means that $\theta$ is a function of the Laplacian operator
\begin{equation*}
    \nabla^2:=q^{AB}\nabla_A\nabla_B =\frac12(\nabla\bar\nabla+\bar\nabla \nabla),
\end{equation*}
on the sphere.} and  invertible.
The new star product defined by $\theta$ from the holomorphic product is $\theta[F\star_\theta G] = \theta[F]\star \theta[G]$ which in terms of the  operation $\cal F$ means 
\be
f\star_\theta g = m[ {\cal F}_\theta ( f\otimes g)], \qquad 
{\cal F}_\theta = (\Delta \theta)^{-1} {\cal F}_H (\theta\otimes \theta),
\ee 
where $\Delta$ is the coproduct\footnote{The coproduct is a morphism of differential operators $\Delta (D_1 D_2) =\Delta(D_1)\Delta(D_2)$ such that  $D m(F\otimes G)=m(\Delta(D) F\otimes G)$ for and differential operator $D$.} of differential operators $\Delta (\nabla_A) = \nabla_A \otimes \mathbbm{1} + \mathbbm{1}\otimes \nabla_A$ and $m$ is the multiplication of functions.
An interesting subclass of $\star$-products are the parity symmetric ones which are such that $C_{n}(f,g)=(-1)^n C_n(g,f)$. The parity symmetric star product are such that the star commutator $[f,g]= f\star g -g\star f$ and the star-symmetrized product $f\circ g = \frac12 (f\star g +g\star f)$ only involves even powers of $\hbar$. We have seen in section \ref{app:parity} that the $6j$ star product arising from the fuzzy sphere
is parity symmetric.
We can now evaluate the holomorphic star product on a basis. Using that 
\be 
{\bar\eth}^n Y_{\alpha}=(-1)^n [A]_n Y^{-n}_{\beta },
\qquad 
{\eth}^n Y_{\beta}=[B]_n Y^{n}_{\beta }.
\ee 
We find that the star product
\bea \label{Nomurastar}
\int_S ( Y_{\alpha} \star_\theta Y_{\beta})Y_{\gamma} &=&\frac{\theta(A)\theta(B)}{\theta(C)}
\sum_{n=0}^\infty \frac{[A]_n [B]_n}{n! \prod_{p=0}^{n-1} (\lambda+p)} \int_S ( Y_{\alpha}^{-n}  Y_{\beta}^{n} Y_{\gamma}),\cr
&=& \begin{pmatrix}
    A & B&C \\ a & b & c
    \end{pmatrix}
    \frac{\theta(A)\theta(B)}{\theta(C)} 
    \sum_{n=0}^\infty \frac{\Gamma(\lambda) }{n! \Gamma(\lambda +n)}[A]_n [B]_n
    \begin{pmatrix}
    A&B&C \\  n & \mm n & 0
    \end{pmatrix},
\eea 
where we use equation \eqref{eqn:tripint} in the last equality. 
We see that, quite remarkably, one recovers the Nomura expression \eqref{eqn:6jexp} provided we chose 
\be \label{eqn:lambdatheta}
\lambda = N+1, \qquad 
\theta(A) = \sqrt{ \frac{\Gamma(\lambda+A)}{\Gamma(\lambda-A)}}.
\ee 
Note that  the infinite sum truncates since $[A]_n=0$ when $A>n$.

It is curious that the specific choice (\ref{eqn:lambdatheta})
for $\theta(A)$ is needed to reproduce the $\star$-product derived 
from the $6j$ symbol described in section \ref{sec:fuzH}.  We have seen in \eqref{eqn:C2diffop} that the $6j$ star product is parity symmetric to all orders in $\hbar$. It would be interesting to  have a different star product derivation of this property. More generally, it would be really interesting to have 
an independent argument for choosing the form \eqref{eqn:lambdatheta} for
 $\theta$ and relating it to the $6j$-symbol. Such an argument would give us an independent  derivation of the Nomura identity
\cite{nomura1989description}.  Note that the star product algebra described here appears in the physics literature as a higher spin symmetry algebra called $hs[\lambda]$ 
\cite{EnriquezRojoProchazkaSachs202105, Vasiliev:1989re, Bordemann:1989zi,
Bergshoeff:1989ns}.

\section{Structure constants for $\mfk{diff}(S^2)$}\label{app:diffsc}
In this appendix we show how to parameterize the $\mfk{diff}(S^2)$ Lie
algebra and structure constants 
in terms of two functions on the sphere, and also derive explicit
expressions for the structure constants in a spherical harmonic basis.  
Working in this basis ensures that all generators correspond to 
smooth vector fields on the sphere.  The results on the explicit
form of these structure constants 
given in equations (\ref{eqn:BB}), (\ref{eqn:BE}), and (\ref{eqn:EE})
are novel (although see \cite{EnriquezRojoProchazkaSachs202105} 
for some partial results).  
In particular, they differ from treatments such as 
\cite{Schwarz:2022dqf} based on commuting holomorphic and antiholomorphic
subalgebras, most of whose generators possess singularities on the sphere.

We begin by fixing a round metric $q_{AB}$ on the sphere, which also determines 
a preferred volume form $\epsilon_{AB}$, as in equations
(\ref{eqn:qAB}) and (\ref{eqn:epAB}).   We can then decompose an arbitrary vector $\xi^A$
on the sphere into a curl and a gradient according to the Hodge decomposition,
\beq
\xi^A = \ep^{BA}\nabla_B \phi_\xi + q^{BA}\nabla_B \psi_\xi.
\eeq
Given an arbitrary vector $\xi^A$, its constituent functions $(\phi_\xi, \psi_\xi)$
can be determined according to the equations
\begin{align}
\psi_\xi &= \frac{1}{\nabla^2} \nabla_A \xi^A \label{eqn:psixi},\\
\phi_\xi &= \frac{1}{\nabla^2} \nabla_A\left(\epsilon\ind{^A_B}\xi^B\right).
\label{eqn:phixi}
\end{align}
Here, $\frac{1}{\nabla^2}$ is the operator that inverts the Laplacian 
$\nabla^2 = q^{AB}\nabla_A \nabla_B$ associated
with the metric $q_{AB}$ on the sphere.  Since the constant functions
lie in the kernel of $\nabla^2$, the inverse $\frac{1}{\nabla^2}$ is defined 
to produce a function with no constant piece, which we take to mean a function
that integrates to zero with respect to the volume form $\epsilon_{AB}$.  
In the spherical harmonic basis $Y_{lm}$, this space 
of functions is spanned by all harmonics
with $l\geq 1$.  

We can therefore decompose the space of all vectors on the sphere into 
subspaces of pure curl and pure gradient vectors: 
\begin{align}
\ap_\phi^A &= \ep^{BA}\nabla_B \phi, \\
\ac_{\phi}^A &= q^{BA}\nabla_B \phi.
\end{align}
We will call the $\ap_\phi^A$ vectors ``magnetic'' and the $\ac_\phi^A$ vectors 
``electric'', in line with their properties under parity transformations
\cite{Thorne1980}.  Note that the tensor $\ep\ind{^A_B}$ defines an
integrable almost complex structure that maps magnetic and electric vectors
into each other according to 
\begin{align} \label{eqn:epJ}
\ep \cdot \ap_\phi = \ac_\phi, \qquad \ep\cdot \ac_\phi = - \ap_\phi.
\end{align}

The effect of multiplying the generators $\ap_\phi^A$ and $\ac_\phi^A$ by 
a scalar function can be expressed in terms of the antisymmetric and 
symmetric brackets defined in  (\ref{eqn:pb}), (\ref{eqn:sb}):
\begin{lem} \label{lem:scalarmult}
Multiplication by a scalar function $\lambda$
acts on the vectors $\ap_\phi$, $\ac_\phi$ 
according to
\begin{align}
\lambda \ap_\phi &= \ap_{\frac{1}{\nabla^2} \big(\langle \lambda,\phi\rangle + \lambda
\nabla^2 \phi\big)} - \ac_{\frac{1}{\nabla^2} \{ \lambda,\phi\} },
\label{eqn:lambdaJ}\\
\lambda \ac_\phi &= \ap_{\frac{1}{\nabla^2} \{\lambda,\phi\} } 
+ \ac_{\frac{1}{\nabla^2} \big(\langle \lambda,\phi\rangle + \lambda
\nabla^2 \phi\big)}.
\label{eqn:lambdaK}
\end{align}
\end{lem}
\begin{proof}
These identities come from applying (\ref{eqn:psixi}) and (\ref{eqn:phixi})
to extract the electric and magnetic potentials of the resulting vector:
\begin{align}
\nabla_A(\lambda \ep\ind{^A_B} \ap_\phi^B) &= \nabla_A(\lambda \nabla^A\phi)
= \langle\lambda,\phi\rangle + \lambda\nabla^2\phi,
\\
\nabla_A(\lambda \ap_\phi^A) &= \ep^{BA}\nabla_B\phi \nabla_A\lambda = 
-\{\lambda,\phi\},
\end{align}
which then leads to (\ref{eqn:lambdaJ}).  An analogous computation 
leads to (\ref{eqn:lambdaK}), which can also be obtained by acting 
with the complex structure $\ep$ on (\ref{eqn:lambdaJ}) and using 
(\ref{eqn:epJ}).
\end{proof}

Before computing the Lie brackets of these vector fields, we will
need some identities satisfied by the brackets $\langle \cdot , \cdot \rangle$
and $\{\cdot,\cdot\}$:
\begin{lem} \label{lem:brackets}
The brackets $\langle \cdot , \cdot \rangle$
and $\{\cdot,\cdot\}$ satisfy 
\begin{align}
\{\phi,\{\psi,\lambda\}\} - \{\psi,\{\phi,\lambda\}\} &= \{\{\phi,\psi\},\lambda\},  \label{eqn:pbpb}  
\\
\{ \phi,\langle\psi,\lambda\rangle \} - \langle\psi,\{\phi,\lambda\}\rangle
&=-\{\langle \phi,\psi\rangle, \lambda\} + \{\phi,\lambda\} \nabla^2\psi,
\label{eqn:sbpb} \\
\langle\phi,\langle \psi,\lambda\rangle\rangle 
-\langle\psi,\langle\phi,\lambda\rangle\rangle
&=\{\{\phi,\psi\},\lambda\} +\langle\phi,\lambda\rangle \nabla^2 \psi
-\langle\psi,\lambda\rangle\nabla^2\phi.
\label{eqn:sbsb}
\end{align}
\end{lem}
\begin{proof}
We handle each case separately:

\begin{itemize} 
\item Proof of (\ref{eqn:pbpb}): This identity is simply the statement of the 
Jacobi identity for the Poisson bracket $\{\cdot, \cdot\}$. We can explicitly check it as follows: 
\begin{equation}
    \begin{aligned}
\{\{\phi,\psi\},\lambda\} 
&=
\ep^{AB}\ep^{CD}\nabla_C(\nabla_A\phi \nabla_B\psi)\nabla_D\lambda 
\\
&=
\ep^{AB}\ep^{CD}\big(\nabla_C\nabla_A \phi \nabla_B \psi \nabla_D\lambda
+ \nabla_A\phi \nabla_C\nabla_B\psi\nabla_D\lambda\big) 
\\
&=
\ep^{AB}\ep^{CD}\big(
\del_A\phi \del_B(\del_C\psi\del_D\lambda)
- \del_A\phi \del_C\psi \del_B\del_D\lambda
\\
&\hphantom{=\ep^{AB}\ep^{CD}\big( }
+\del_B\psi\del_A(\del_C\phi \del_D\lambda)
-\del_B\psi \del_C\phi \del_A\del_D\lambda
\big)
\\
&=
\{\phi,\{\psi,\lambda\}\} - \{\psi,\{\phi,\lambda\}\}.
\end{aligned}
\end{equation}
\item Proof of (\ref{eqn:sbpb}):
This can be derived straightforwardly by first evaluating the 
Lie derivative of $\ep^{AB}$ with respect to the vector $\ac_\psi$.  
Since the resulting tensor remains antisymmetric in its indices, we must 
have $\lie_{\ac_\psi}\ep^{AB} = \alpha \ep^{AB}$ with $\alpha = \frac12
\ep_{AB}\lie_{\ac_\psi}\ep^{AB} = -\frac12\ep^{AB}\lie_{\ac_\psi}\ep_{AB}$.
Then since $\lie_{\ac_\psi}\ep_{AB} = (\div K_\psi) \ep_{AB} 
= \del^2 \psi \,\ep_{AB}$, we conclude that $\lie_{\ac_\psi}\ep^{AB} = 
-\nabla^2\psi\, \ep^{AB}$.  We then evaluate the nested brackets
\begin{equation}
    \begin{aligned}
\langle \psi,\{\phi,\lambda\}\rangle
&=
\lie_{\ac_\psi} \big(\ep^{AB}\nabla_A \phi \nabla_B\lambda\big)
\\
&=
\big(\lie_{\ac_\psi} \ep^{AB} \big) \del_A \phi \del_B\lambda
+\ep^{AB}\big( \nabla_A (\lie_{\ac_\psi} \phi) \nabla_B \lambda
+ \nabla_A\phi \nabla_B (\lie_{\ac_\psi}\lambda)\big)
\\
&=
-\nabla^2\psi\,\{\phi,\lambda\} +\{\langle \psi,\phi\rangle,\lambda\}
+\{ \phi, \langle\psi,\lambda\rangle\}.
\end{aligned}
\end{equation}
\item Proof of (\ref{eqn:sbsb}): To derive this, 
we first work out an expression for $\del^A \del^B \phi \del_A\psi
\del_B\lambda$ in terms of the symmetric brackets:
\begin{equation}
    \begin{aligned}
\del^A\del^B\phi \del_A\psi \del_B\lambda
&=
\del^A(\del^B\phi\del_B\lambda)\del_A\psi - 
\del^B\phi\del_A\psi \del^A\del_B\lambda \\
&=
\langle\langle \phi,\lambda\rangle,\psi\rangle
-\langle\phi,\langle\psi,\lambda\rangle\rangle
+ \del^B\phi \del_B\del_A\psi\del^A\lambda  \\
&=
\langle\langle \phi,\lambda\rangle,\psi\rangle
-\langle\phi,\langle\psi,\lambda\rangle\rangle
+\langle\langle \phi,\psi \rangle,\lambda\rangle
-\del^A\del^B\phi\del_B\psi \del_A\lambda, 
\end{aligned}
\end{equation}
and hence 
\beq \label{eqn:4derivident}
\del^A\del^B\phi \del_A\psi \del_B\lambda
=\frac12 \Big(
\langle\langle \phi,\lambda\rangle,\psi\rangle
+\langle\langle \phi,\psi \rangle,\lambda\rangle
-\langle\phi,\langle\psi,\lambda\rangle\rangle
\Big).
\eeq
We can then evaluate the nested bracket $\{\{\phi,\psi\},\lambda\}$
by first exploiting the relation $\ep^{AB}\ep^{CD} = q^{AC}q^{BD}-q^{AD}q^{BC}$,
and then applying (\ref{eqn:4derivident}) to obtain
\begin{equation}
    \begin{aligned}
\{\{\phi,\psi\},\lambda\}
&=
\ep^{AB}\ep^{CD}\del_A(\del_C\phi \del_D\psi) \del_B\lambda 
\\
&=
(q^{AC}q^{BD}-q^{AD}q^{BC})\del_A(\del_C\phi \del_D\psi) \del_B\lambda
\\
&=
\del^2 \phi \langle\psi,\lambda\rangle
+\del^A\phi \del_A\del^B\psi \del_B\lambda
-\del_A\del^B\phi \del^A\psi\del_B\lambda 
-\langle\phi,\lambda\rangle \del^2\psi
\\
&=
\langle  \phi, \langle\psi,\lambda\rangle\rangle
-\langle \psi,\langle\phi,\lambda\rangle\rangle
+ \langle \psi,\lambda\rangle \del^2\phi - \langle \phi,\lambda\rangle\del^2\psi.
\end{aligned}
\end{equation}
\end{itemize}
\end{proof}

The  Lie brackets of $\ap_\phi$ and $\ac_\psi$
can now be computed by examining how these vectors act on scalar
functions.  Any vector field acts as a derivation on the space of functions,
and these derivations can be expressed in terms of the brackets,
\begin{align}
\ap_\phi (\lambda) &= \ap_\phi^A\nabla_A \lambda = \{\phi, \lambda\},
\\
\ac_\phi (\lambda) &= \ac_\phi^A\nabla_A\lambda = \langle \phi,\lambda\rangle.
\end{align}
The Lie bracket of two vector fields is then given by the commutator of the 
two associated derivations acting on a function.  The left hand 
side of the bracket identities in Lemma \ref{lem:brackets} expresses the 
three options for these commutators, and the expressions on the right hand
side give the equivalent derivation acting on the function $\lambda$.  
Hence, these identities immediately allow us to write down expressions for 
the Lie brackets of the vector fields:
\begin{align}
[\ap_\phi, \ap_\psi] &= \ap_{\{\phi,\psi\}}, \label{eqn:JJ}\\
[\ap_\phi, \ac_\psi] &= -\ap_{\langle \phi,\psi\rangle} + \nabla^2\psi \, \ap_\phi,
\label{eqn:KJ} \\
[\ac_\phi,  \ac_\psi] &= \ap_{\{\phi, \psi\} } +\nabla^2\psi\,\ac_\phi
- \nabla^2\phi \,\ac_\psi.
\label{eqn:KK}
\end{align}
Then using the identities in Lemma \ref{lem:scalarmult} for multiplication
of a vector by a scalar function, we can reduce equations (\ref{eqn:KJ})
and (\ref{eqn:KK}) to
\begin{align}
[\ap_\phi, \ac_\psi] &= \ac_{\frac{1}{\nabla^2} \{\phi, \nabla^2\psi\} }
- \ap_{\langle \phi,\psi\rangle -\frac{1}{\nabla^2} \big(
\langle\phi,\nabla^2\psi\rangle + \nabla^2\phi\nabla^2\psi\big)}, 
\label{eqn:JKred}\\
[\ac_\phi, \ac_\psi] &= \ap_{\{\phi, \psi\} 
-\frac1{\nabla^2}\big( \{\nabla^2\phi, \psi\} +\{\phi,\nabla^2\psi\}\big) }
+\ac_{\frac{1}{\nabla^2}\big( \langle\phi,\nabla^2\psi\rangle 
-\langle \nabla^2\phi,\psi\rangle\big)}.
\label{eqn:KKred}
\end{align}

We can now explicitly parameterize the generators and structure 
constants by decomposing the potentials $(\phi, \psi)$ in a spherical
harmonic basis.  Following the conventions and notation of 
appendix \ref{sec:sh}, we let $Y_\alpha$ denote a spherical
harmonic with $\alpha = (A,a)$ denoting its total angular momentum $A$
and magnetic quantum number $a$, with $-A\leq a\leq A$. 
We then employ the shorthand $\ap_\alpha = \ap_{Y_\alpha}$, $\ac_\alpha = \ac_{Y_\alpha}$ to denote the 
 generators in the spherical harmonic basis.  Since the Lie brackets (\ref{eqn:JJ}), (\ref{eqn:JKred}), 
 and (\ref{eqn:KKred}) are expressed in terms of the Poisson bracket, 
 symmetric bracket, and product of functions, we can express the 
 right hand sides of these equations using the structure constants
 for these operations given in equations (\ref{eqn:Cabc}), (\ref{eqn:Gabc}),
and (\ref{eqn:Eabccontra}).  Recalling also that $\nabla^2 Y_\alpha
= -\la{A} Y_\alpha$, 
$\frac{1}{\nabla^2} Y_\alpha = \frac{-1}{\la{A} } Y_\alpha$
with $\la{A} = A(A+1)$ defined in (\ref{eqn:A1}), the structure constants 
are immediately found to be
\begin{align}
[\ap_\alpha, \ap_\beta] &= C\ind{_\alpha_\beta^\gamma} \ap_\gamma,
\label{eqn:BB}
\\
[\ap_\alpha, \ac_\beta] &=
\frac{\la B}{\la C} C\ind{_\alpha_\beta^\gamma} \ac_\gamma
+\frac{1}{\la C} \left((\la B-\la C) G\ind{_\alpha_\beta^\gamma}
-\la A \la B E\ind{_\alpha_\beta^\gamma} \right) \ap_\gamma, \label{eqn:BE}
\\
[\ac_\alpha,\ac_\beta] &=
\frac{\la C- \la A - \la B}{\la C}
C\ind{_\alpha_\beta^\gamma} \ap_\gamma
+\frac{\la B- \la A}{ \la C} G\ind{_\alpha_\beta^\gamma}\ac_\gamma,
\label{eqn:EE}
\end{align}

Restricting these relations to $A = B = 1$, we find that the algebra 
of the six generators $(B_{1a}, E_{1a})$ closes, and reduces to
\begin{align}
[B_{1a}, B_{1b}] &= C\indices{_{(1a)}_{(1b)}^{(1c)}} B_{1c}, 
\\
[B_{1a}, E_{1b}] &=  C\indices{_{(1a)}_{(1b)}^{(1c)}} E_{1c},
\\
[E_{1a}, E_{1b}] &=  -C\indices{_{(1a)}_{(1b)}^{(1c)}} B_{1c}.
\end{align}
This algebra is readily recognized as the $\mfk{sl}(2,\mathbb{C})$
subalgebra of $\mfk{diff}(S^2)$ consisting of the six globally defined conformal
Killing vectors of $S^2$.  

Finally, it is  interesting to note that we can form a new 
holomorphic basis for this algebra by forming combinations
of $B_\alpha$ and $E_\alpha$ that are eigenvectors for 
the complex structure $\epsilon\indices{^A_B}$.  These are given by
$F_\alpha^+ = E_\alpha + i B_\alpha$ and $F_\alpha^-
 =  -E_\alpha + i B_\alpha$
, which satisfy
$\epsilon\cdot F_\alpha^\pm = \pm i F_\alpha^\pm$. The 
$F_\alpha^+$ form a subalgebra within (the complexification of) the
full diffeomorphism algebra,
as do $F_\alpha^-$, which follows from the fact that the 
complex structure $\epsilon\indices{^A_B}$ is integrable.\footnote{Recall
that an integrable complex structure is one in which the Nijenhuis 
tensor $N\indices{^A_B_C}$ vanishes.  This tensor is defined by the 
relation
\beq
N(X,Y) = [X,Y] + \epsilon \cdot([\epsilon \cdot X, Y] + [X, \epsilon \cdot Y])
-[\epsilon \cdot X, \epsilon \cdot Y],
\eeq
where $X, Y$ are vectors, and the brackets are vector field Lie brackets.  
The vanishing of this tensor implies that eigenvectors 
$X^\pm$ of the complex structure $\epsilon$ form a subalgebra, i.e.
\beq
\epsilon[X^\pm, Y^\pm] = \pm i [X^\pm, Y^\pm].
\eeq
}

The algebra in the $(F_\alpha^+, F_\alpha^-)$ basis can be computed following a
similar method as in the $(B_\alpha, E_\alpha)$ basis.  Given an arbitrary 
scalar function $\phi$, we can construct holomorphic
and antiholomorphic vector fields $F_\phi^\pm = \pm E_\phi +iB_\phi$.  These
vector fields act on functions as derivations, and this action can be 
equivalently expressed in terms of two new brackets for scalar functions
$(\cdot,\cdot)_\pm$, defined by 
\beq
(\phi,\psi)_\pm = \pm\langle\phi,\psi\rangle +i\{\phi,\psi\}.
\eeq
Note that these are neither symmetric nor antisymmetric, but instead
satisfy
\beq
(\phi,\psi)_+ = -(\psi,\phi)_-.
\eeq
The action of the vector fields on a function $\lambda$ can then be expressed 
as 
\beq
F_\phi^\pm(\lambda) = (\phi,\lambda)_\pm.
\eeq

Next we note that the multiplication of $F_\phi^\pm$ by a scalar function
can be derived from the relations (\ref{eqn:lambdaJ}), (\ref{eqn:lambdaK}),
and leads to 
\beq \label{eqn:Fpmscalarmult}
\lambda F_\phi^\pm = F^\pm_{\frac1{\nabla^2}\big(\mp(\lambda,\phi)_\mp + 
\lambda \nabla^2\phi\big)}.
\eeq
Interestingly, unlike the vectors $B_\phi, E_\phi$, scalar multiplication
maps the set of vectors $F_\phi^+$ into themselves, and similarly for 
$F_\phi^-$.

The Lie bracket of these vector fields is most straightforwardly obtained 
by computing the nested relations of the $(\cdot,\cdot)_\pm$ brackets.  
These follow from the relations in Lemma \ref{lem:brackets}, and lead to
\begin{align}
\Big(\phi,(\psi,\lambda)_+\Big)_+ - \Big(\psi,(\phi,\lambda)_+\Big)_+
 &= \nabla^2\psi (\phi,\lambda)_+ - \nabla^2\phi(\psi,\lambda)_+, 
 \\
\Big(\phi,(\psi,\lambda)_-\Big)_- - \Big(\psi,(\phi,\lambda)_-\Big)_-
&= -\nabla^2\psi(\phi,\lambda)_- +\nabla^2\phi(\psi,\lambda)_-,
\\
\Big(\phi,(\psi,\lambda)_-\Big)_+ - \Big(\psi,(\phi,\lambda)_+\Big)_-
&= \Big((\phi,\psi)_+,\lambda\Big)_+ - \nabla^2\psi(\phi,\lambda)_+
+\Big((\phi,\psi)_+,\lambda\Big)_- - \nabla^2\phi(\psi,\lambda)_-.
\end{align}
These then imply relations for the brackets of the $F_\phi^\pm$ vector fields,
which can now be expressed as 
\begin{align}
[F_\phi^+, F_\psi^+] &= \nabla^2\psi\, F_\phi^+ - \nabla^2\phi\, F_\psi^+,
\\
[F_\phi^-, F_\psi^-] &= -\nabla^2\psi\, F_\phi^- + \nabla^2\phi\, F_\psi^-, 
\\
[F_\phi^+, F_\psi^-] &=
F^+_{(\phi,\psi)_+} - \nabla^2\psi\, F_\phi^+ + F^-_{(\phi,\psi)_+} - 
\nabla^2\phi\, F_\psi^-.
\end{align}
We then can apply the formula (\ref{eqn:Fpmscalarmult}) for scalar multiplication
acting on the vectors to derive
\begin{align}
[F_\phi^+, F_\psi^+] &= F^+_{\frac{1}{\nabla^2}\big((\phi,\nabla^2\psi)_+
+ (\nabla^2\phi,\psi)_-\big)}, \label{eqn:F+F+fxn}\\
[F_\phi^-, F_\psi^-] &= F^-_{\frac{1}{\nabla^2}\big((\phi,\nabla^2\psi)_-
+ (\nabla^2\phi,\psi)_+\big)}, \label{eqn:F-F-fxn}\\
[F_\phi^+, F_\psi^-] &=
F^+_{(\phi,\psi)_+ 
-\frac{1}{\nabla^2}\big((\phi,\nabla^2\psi)_+ + \nabla^2\phi\nabla^2\psi\big)}
+ F^-_{(\phi,\psi)_+
-\frac{1}{\nabla^2}\big((\nabla^2\phi,\psi)_+ +\nabla^2\phi\nabla^2\psi\big)}.
\label{eqn:F+F-fxn}
\end{align}
Note that these relations explicitly verify that the $F^+_\phi$ vectors 
form a subalgebra, as do the $F^-_\phi$ vectors, but these two 
subalgebras do not commute.

From these relations, the structure constants in the spherical harmonic basis
follow straightforwardly.  Defining the structure constants for the 
$(\cdot,\cdot)_\pm$ brackets according to $(Y_\alpha, Y_\beta)_\pm 
= \tensor[^\pm]{H}{_\alpha_\beta^\gamma}Y_\gamma$, we see that they are 
related to  $C\indices{_\alpha_\beta^\gamma}$ and $G\indices{_\alpha_\beta^\gamma}$
via
\beq
\tensor[^\pm]{H}{_\alpha_\beta^\gamma}
= \pm G\indices{_\alpha_\beta^\gamma} +iC\indices{_\alpha_\beta^\gamma}.
\eeq
Using the expressions (\ref{eqn:Cabc}) and (\ref{eqn:Gabc}) for the 
$G$ and $C$ structure constants and 
lowering an index with the metric $\delta_{\alpha\beta}$ defined 
by (\ref{eqn:shmetric}), 
we can express $\tensor[^\pm]{H}{_\alpha_\beta_\gamma}$
explicitly in terms of $3j$-symbols according to
\begin{equation}
\tensor[^\pm]{H}{_\alpha_\beta_\gamma} = \mp [A]_1[B]_1
\tj{A&B&C\\a&b&c} \tj{A&B&C\\\mp1&\pm1&0}.
\end{equation}
In terms of these, we can immediately translate the expressions
(\ref{eqn:F+F+fxn}), (\ref{eqn:F-F-fxn}), and (\ref{eqn:F+F-fxn})
into formulas for the structure constants in the spherical
harmonic basis:
\begin{align}
[F^+_\alpha, F^+_\beta] &=
\frac{1}{\la C}\left(\la{B}\, \tensor[^+]{H}{_\alpha_\beta^\gamma}
+\la{A}\, \tensor[^-]{H}{_\alpha_\beta^\gamma}\right)F^+_\gamma, \\
[F^-_\alpha, F^-_\beta] &=
\frac{1}{\la C}\left(\la{B}\, \tensor[^-]{H}{_\alpha_\beta^\gamma}
+\la{A}\, \tensor[^+]{H}{_\alpha_\beta^\gamma}\right)F^-_\gamma, \\
[F^+_\alpha, F^-_\beta] &=
\frac{1}{\la C}
\left(\big(\la C - \la B\big)\,\tensor[^+]{H}{_\alpha_\beta^\gamma}
+\la A \la B E\indices{_\alpha_\beta^\gamma}\right)F^+_\gamma \nonumber \\
&\;
+\frac{1}{\la C}
\left(\big(\la C - \la A\big)\,\tensor[^+]{H}{_\alpha_\beta^\gamma}
+\la A \la B E\indices{_\alpha_\beta^\gamma}\right)F^-_\gamma.
\end{align}

\section{Algebra deformation}\label{appsec:algebra deformation}

The deformation of symmetry algebras has a long and fruitful history in physics:
deformation of the abelian phase-space algebra into a centrally extended algebra with physical parameter $\hbar$ is at the core of the discovery of quantum mechanics.
Another type of physical deformation involves deforming 
a semi-direct product algebra into a semi-simple algebra which is much more regular; 
a standard mathematical reference on the theory of deformation of Poisson algebras, Lie algebras, and algebras is \cite{Flato:1995vm}. Two key examples are 
\begin{enumerate}
    \item the deformation of the Poincar\'e group  into the de Sitter group. The former can be obtained from the latter by contraction \cite{InonuWigner195306}. The deformation parameter is the cosmological constant $\Lambda$. 
    
    \item the deformation of the Galilean group into the Poincar\'e algebra. Again, the former can be obtained from the latter by contraction \cite{InonuWigner195306}. The deformation parameter is the inverse of the speed of light $c$.
\end{enumerate}
The corner symmetry algebra $\mfk{g}_{\mfk{sl}(2,\mbb{R})}$ is a semi-direct product. It is therefore natural to look for a deformation of $\mfk{g}_{\mfk{sl}(2,\mbb{R})}$ which is semi-simple. In this paper, we  focused on the deformation of $\mfk{sdiff}(S)$ and of the centralizer algebras
$\mfk{c}_{\mathbb{R}} = \sdiff \oplus_{\mcal L} \mathbb{R}^S$ and $\mfk{c}_{\slr} = \sdiff\oplus_{\mcal L}\slr^S$.

What we propose here is a deformation of $\mfk{c}_{\mfk{sl}(2,\mbb{R})}$, denoted $\mfk{c}_{\mfk{sl}(2,\mbb{R})}[\lambda]$, which will prove to be invaluable at the quantum level. The deformation parameter is a real parameter $\lambda$. It is a new constant that still needs to be interpreted in a physical term as a constant of nature and which, we hope, could be promoted to the same status as $\hbar$, $\Lambda$ and $c$  have reached. One proposal for such a dimensionless deformation parameter is that it is given as a measure of the ratio  of the Planck scale over the cosmological scale, which is the only universal dimensionless number we naturally encounter in quantum gravity. In our analysis and for irreducible representations we have seen through the Casimir matching
procedure described in section \ref{sec:casimirs} that $\lambda\sim \frac{1}{N^2}$ 
is related to the quantum of the area associated with the corner sphere.

Deformations of Poisson algebras arise as follows.  
Let $M$ be a Poisson manifold, meaning it is equipped with a bilinear map 
on functions
$\{\cdot,\cdot\}:C(M)\times C(M)\to C(M)$\footnote{Here, $C(M)$ is the space of smooth functions on $M$.}
which satisfies antisymmetry,  Leibniz properties, and the Jacobi identity. 
The deformation of a Poisson algebra is characterized by a Poisson two-cocycle, that is, a map $D :C(M)\times C(M)\to C(M)$  which is a skew-symmetric bi-derivation and satisfies the Poisson 2-cocycle  identity. Explicitly this means that
\bea
D(f,g)&=&-D(g,f),\nonumber
\\
D(f,gh)&=& D(f,g)h + g D(f,h),\label{cocycleapp}
\\
\{f,D(g,h)\}+\{g,D(h,f)\}+\{h,D(f,g)\} &=&-[D(f,\{g,h\})+D(g,\{h,f\})+D(h,\{f,g\})]\nonumber .
\eea
These identities simply imply that the deformed bracket $\{f,g\}_\lambda:= \{f,g\}+ \lambda D(f,g)$ satisfies the Jacobi identity to first order in $\lambda$. A Poisson 2-cocycle is trivial if it can be written in terms of a 1-cocycle $D_1: C(M)\to C(M)$ as \be\label{1cocycle}
D(f,g)= D_1(\{f,g\})-\{D_1(f),g\}-\{f, D_1(g)\},
\ee 
where $D_1(f)$ is a differential operator. In such a case, the Poisson deformation is trivial and simply amounts to a redefinition of the variables $f\to f- \lambda D_1(f)$.
Note that the last identity in \eqref{cocycleapp} and the 1-cocycle deformation \eqref{1cocycle} can  be written as $\delta D (f,g,h)=0$ and $D=\delta D_1$, respectively,  where $\delta$ is the Chevalley coboundary operator \cite{Gutt}.

The existence of deformation for $\sdiff$ follows straightforwardly from the construction of the star product operation done in sections \ref{sec:mpexpansion} and \ref{app:star}: From the analysis done there, the fact that $C^{(2)}$ is symmetric  that $C^{(3)}$ is skew-symmetric and the proof of associativity of the star product, we know that  $C^{(3)}$ is a Poisson cocycle for the sphere Poisson bracket
$\{Y_\alpha,Y_\beta\}_\epsilon= C_{\alpha \beta}{}^\gamma Y_\gamma$. 
In particular if we define $C^{(3)}(Y_\alpha,Y_\beta)  := D_{\alpha \beta}{}^\gamma$ this means that 
\bea\label{cocycD}
C_{\alpha \delta}{}^\sigma D_{\beta \gamma}{}^\delta + 
D_{\alpha \delta}{}^\sigma C_{\beta \gamma}{}^\delta + \mathrm{cycl}[\alpha,\beta,\gamma] =0,
\eea
where cycl$[\alpha,\beta,\gamma]$ means that we perform a cyclic permutation of the indices. 

To obtain a nontrivial Poisson deformation of $\sdiff$, one first uses the fact that $D$ is a bi-derivation. This implies that the knowledge of $D$ on arbitrary functions  is entirely determined by the knowledge of $D(J_\alpha,J_\beta)$, where,
following the notation (\ref{eqn:Jalpha}),
$J_\alpha = J[Y_\alpha]$ is a basis for $C(M)$, since
$D(f,g)=\sum_{\alpha,\beta} \frac{\pa f}{\pa J_\alpha} \frac{\pa g}{\pa J_\beta} D(J_\alpha,J_\beta)$. 
The identity \eqref{cocycD} then implies that there exists a Poisson deformation of $\sdiff$ Poisson algebra simply given by 
\be
D(J_\alpha,J_\beta)= J[ C^{(3)}(Y_\alpha,Y_\beta)] = D_{\alpha \beta}{}^\gamma J_\gamma.
\ee 

The deformation of $\mfk{c}_{\R}(S)$ goes along the same line. From the differentiability property, one learns that it is enough to give the prescription on the generators $(J_\alpha, N_\alpha)$ with bracket given in \eqref{NNcommutatorR}. One chooses 
    \begin{equation}\label{defalgebra for cR}
    \begin{aligned}
    D(J_\alpha, J_{\beta})&=D_{\alpha \beta}{}^\gamma J_\gamma, \\
    D(J_\alpha, N_{\beta})&=D_{\alpha \beta}{}^\gamma N_\gamma,
    \\
    D(N_\alpha, N_{\beta})&= -C_{\alpha \beta}{}^\gamma J_\gamma.
    \end{aligned}
    \end{equation}
To verify the cocycle property $\delta D(P_0,P_1,P_2)=0$ with $P_i$ denoting the 
arbitrary Lie-algebra generators and the coboundary $\delta$ acting as 
\begin{eqaligned}
    \delta D(P_0,P_1,P_2)&:=\{P_0,D(P_1,P_2)\}-\{P_1,D(P_0,P_2)\}+\{P_2,D(P_0,P_1)\}
    \\
    &\hphantom{:}-D(\{P_0,P_1\},P_2)+D(\{P_0,P_2\},P_1)-D(\{P_1,P_2\},P_0),
\end{eqaligned}
we need to investigate 4 different cases depending on whether the  argument is: $(J_\alpha,J_\beta,J_\gamma)$, $(J_\alpha,J_\beta,N_\gamma)$,$(J_\alpha,N_\beta,N_\gamma)$ or $(N_\alpha,N_\beta,N_\gamma)$.
The proof goes by inspection of each case separately. The  cocycle identity \eqref{cocycD} proves the first two cases, and  the Jacobi identity of the sphere bracket proves the following two cases:
\begin{eqaligned}
    \delta D(J_\alpha,J_\beta,J_\gamma)&= 
\left(C_{\alpha \delta}{}^\sigma D_{\beta \gamma}{}^\delta + 
D_{\alpha \delta}{}^\sigma C_{\beta \gamma}{}^\delta + \mathrm{cycl}[\alpha,\beta,\gamma]\right)J_\sigma =0,
\\
    \delta D(J_\alpha,J_\beta,N_\gamma)&=
\left(C_{\alpha \delta}{}^\sigma D_{\beta \gamma}{}^\delta + 
D_{\alpha \delta}{}^\sigma C_{\beta \gamma}{}^\delta + \mathrm{cycl}[\alpha,\beta,\gamma]\right)N_\sigma=0,
\\
\delta D(J_\alpha,N_\beta,N_\gamma)&=
\left(C_{\alpha \delta}{}^\sigma C_{\beta \gamma}{}^\delta + \mathrm{cycl}[\alpha,\beta,\gamma]\right)J_\sigma=0,
\\
\delta D(N_\alpha,N_\beta,N_\gamma)&=
\left(C_{\alpha \delta}{}^\sigma C_{\beta \gamma}{}^\delta + \mathrm{cycl}[\alpha,\beta,\gamma]\right)N_\sigma=0.
\end{eqaligned}

We can finally describe the deformation for the algebra $\mfk{c}_{\slr}$ given by \eqref{NNcommutatorslr}. This deformation involves the coefficient $C^{(2)}$ and we denote $C^{(2)}(Y_\alpha,Y_\beta):= F_{\alpha \beta}{}^\gamma Y_\gamma$ which is symmetric under the exchange $(\alpha,\beta)$. $F\indices{_\alpha_\beta^\gamma}$
is essentially the 
$N^{-2}$ term in the expansion of $\wh{E}\indices{_\alpha_\beta^\gamma}$. 
From the associativity of the star product,  we obtain that 
\be  
Y_\alpha \circ (Y_\gamma\circ Y_\beta)- (Y_\alpha \circ Y_\gamma) \circ Y_\beta = [Y_\gamma,[Y_\alpha,Y_\beta]],
\ee
where $[f,g]= f\star g - g\star f$ and $f \circ g = \frac12(f\star g + g\star f)$.
Expanding this identity at second order we obtain that  
\be  \label{EFCC}
(E_{\beta \gamma}{}^\delta F_{\alpha \delta}{}^\sigma + F_{\beta \gamma}{}^\delta E_{\alpha \delta}{}^\sigma) - 
(E_{ \gamma \alpha}{}^\delta F_{\beta \delta}{}^\sigma + F_{\gamma \alpha }{}^\delta E_{\beta \delta}{}^\sigma)
= C_{\alpha \beta}{}^\delta C_{\gamma \delta}{}^\sigma.
\ee 
The deformation cocycle is now taken to be 
   \begin{equation}
   \label{defalgebra for csl(2,R)}
    \begin{aligned}
    D(J_\alpha, J_{\beta})&=D_{\alpha \beta}{}^\gamma J_\gamma, \\
    D(J_\alpha, N_{a\beta})&=D_{\alpha \beta}{}^\gamma N_{a\gamma},
    \\
    D(N_{a\alpha}, N_{b\beta})&=\varepsilon_{ab}{}^c F_{\alpha\beta}{}^\gamma N_{c \gamma} -\eta_{ab} C_{\alpha\beta}{}^\gamma J_\gamma.
    \end{aligned}
    \end{equation}
One sees that this deformation restricts to the previous one if one chooses $N_\alpha= N_{1 \alpha}$.

The proof for the cocycle identities  follows similarly. We need to look at them case by case.
The proof or the combinations $(J_\alpha,J_\beta,J_\gamma)$, $(J_\alpha,J_\beta,N_{c\gamma})$
is the same as before. For the combination $(N_{a\alpha},N_{b\beta},N_{c\gamma})$ we use that 
\begin{eqaligned}
    D(N_{a\alpha},\{N_{b\beta},N_{c\gamma}\})
&= \varepsilon_{bc}{}^d \varepsilon_{ad}{}^s  
E_{\beta \gamma}{}^\delta F_{\alpha \delta}{}^\sigma N_{s\sigma} - 
\varepsilon_{bca} E_{\beta \gamma}{}^\delta C_{\alpha \delta}{}^\sigma J_{\sigma},
\\
\{N_{a\alpha},D(N_{b\beta},N_{c\gamma})\}
&=\varepsilon_{bc}{}^d \varepsilon_{ad}{}^s F_{\beta \gamma}{}^\delta E_{\alpha \delta}{}^\sigma N_{s\sigma}
- \eta_{bc}  C_{\beta \gamma}{}^\delta C_{\alpha \delta}{}^\sigma  N_{a \sigma}.
\end{eqaligned}
To evaluate the sum over the cyclic permutation, we first use that 
$E_{\beta \gamma}{}^\delta C_{\alpha \delta}{}^\sigma + \mathrm{cycl}[\alpha,\beta,\gamma] =0$ which follows from the fact that the Poisson bracket is a bi-derivation.
Then, we collect the terms proportional to 
$\eta_{ab} N_{c\sigma}$ which are proven to be proportional\footnote{One simply needs to use that $\varepsilon_{bc}{}^d \varepsilon_{ad}{}^s = \delta_c^s \eta_{ab} - \delta_b^s \eta_{ac} $.} to the identity \eqref{EFCC}.

Finally for the combination
$(J_{\alpha},N_{b\beta},N_{c\gamma})$ we use that 
\begin{eqaligned}
    D(J_{\alpha},\{N_{b\beta},N_{c\gamma}\})
&= \varepsilon_{bc}{}^s 
E_{\beta \gamma}{}^\delta D_{\alpha \delta}{}^\sigma N_{s\sigma},
\\
\{J_{\alpha},D(N_{b\beta},N_{c\gamma})\}
&=\varepsilon_{bc}{}^s F_{\beta \gamma}{}^\delta C_{\alpha \delta}{}^\sigma N_{s\sigma} -\eta_{bc} C_{\beta \gamma}{}^\delta C_{\alpha \delta}{}^\sigma J_\sigma,
\\
D(N_{b \beta},\{N_{c\gamma}, J_\alpha\})
&=\varepsilon_{bc}{}^s   
C_{ \gamma \alpha }{}^\delta F_{\beta \delta}{}^\sigma N_{s\sigma} - 
\eta_{bc} C_{ \gamma \alpha }{}^\delta C_{\beta \delta}{}^\sigma  J_{\sigma},
\\
\{N_{b \beta },D(N_{c\gamma}, J_\alpha)\}
&=\varepsilon_{bc}{}^s  D_{ \gamma \alpha}{}^\delta E_{\beta \delta}{}^\sigma N_{s\sigma}
- \eta_{bc}  C_{ \gamma \alpha }{}^\delta C_{\beta \delta}{}^\sigma  J_{\sigma}.
\end{eqaligned}
The cocycle identity then follows from the expansion of the differential identity for the star product
\be 
[Y_\alpha, Y_\beta \circ Y_\gamma]= [Y_\alpha,Y_\beta]\circ Y_\gamma + Y_\beta\circ [Y_\alpha,Y_\gamma].
\ee 
In components this means that 
\be  
E_{\beta \gamma}{}^\delta D_{\alpha \delta}{}^\sigma + F_{\beta \gamma}{}^\delta C_{\alpha \delta}{}^\sigma 
= 
( C_{ \gamma \alpha}{}^\delta F_{\beta \delta}{}^\sigma - D_{ \gamma \alpha}{}^\delta E_{\beta \delta}{}^\sigma  )
+ 
( C_{ \beta \alpha}{}^\delta F_{\gamma \delta}{}^\sigma - D_{ \beta \alpha}{}^\delta E_{\gamma \delta}{}^\sigma  ),
\ee 
which completes the proof.

\section{$\mfk{su}(N)$ and $\mfk{su}(N,N)$ relations}\label{SUN}
In this appendix, we establish some key $\mfk{su}(N)$ and $\mfk{su}(N,N)$ identities.

\subsection{$\mfk{su}(N)$ } \label{app:SUNrelns}
We start with the facts necessary for the proof of equivalence between \eqref{eqn:XaXb} and \eqref{eqn:EijEkl}. In particular, we prove \eqref{eq:the identity for product of Yalpha and Ybeta}. We have seen that in the fundamental and the adjoint  representations of $\mfk{su}(N)$, which we denote by $\pi_{\mbf N}$ and $\pi_{\mbf{ad}}$, respectively, we have 
\begin{equation}
    [\pi_{\bf N}(\algX_\alpha)]_i{}^j=\frac{N}{2i}[\wh{Y}_\alpha]_i{}^j, \qquad 
    [\pi_{\bf ad}(\algX_\alpha)]\indices{_\beta^\gamma}=\wh{C}\indices{_{\alpha\beta}^\gamma}.\label{repX}
\end{equation}
In the following we denote $i,j \in \{1,\ldots, N\}$ the vectorial indices and 
$\alpha,\beta \in \{1,\ldots, N^2-1\} $ are the adjoint indices.
We denote by $U= \exp(u^\alpha X_\alpha) $ to be an abstract group element. It is well-known that the  adjoint action is simply given by
 \be 
 \rho_{\bf Ad}(U) \wh{Y}^\beta : = \pi_{\bf N}(U) \wh{Y}^\beta  \pi_{\bf N}(U^{-1}).
 \ee 
 This relation can be written in components in terms of the components  
 $\pi_{\bf N}U$ and  $\pi_{\bf ad}U$ as
 which in components means 
 \be \label{IdN}
\sum_{\alpha \in I_N}(\wh{Y}^\alpha)_i{}^l [\pi_{\bf Ad}(U)]_\alpha{}^\beta  = 
[\pi_{\bf N}(U)]_i{}^j  [\pi_{\bf N}(U^{-1})]_k{}^l (\wh{Y}^\beta)_j{}^k. 
 \ee 
 Next, we establish the relationship
 \be \label{eqn:fuzHcompleteness}
 \sum_{\beta } (\wh Y^\beta)_j{}^k(\wh Y_\beta)_{k'}{}^{j'}  = N \delta_j^{j'}\delta_{k'}^k.
 \ee 
 To see this just contract the LHS with $(\wh Y^\alpha)_{j'}{}^{k'}$. We get that this is equal to 
 $\sum_{\beta } (\wh Y^\beta)_j{}^k {\mathrm{Tr}}_{\mbf{N}}
 (\wh Y_\beta \wh Y^\alpha ) = N (\wh Y^\beta)_j{}^k$, where we used that 
 ${\mathrm{Tr}}_{\mbf{N}}
 (\wh Y_\beta \wh Y^\alpha ) = N\delta_\beta^\alpha$. Since $\wh Y^\alpha$ with $\alpha\in I_N$ is a complete basis of matrices we get the desired equality. Using this identity and contracting \eqref{IdN} with $(\wh{Y}_\beta)\indices{_{k'}^{j'}}$ and summing over $\beta$ gives the relation
 \be \label{FundR}
\sum_{\alpha, \beta }(\wh Y^\alpha)_i{}^l [\pi_{\bf Ad}(U)]_{\alpha}{}^\beta (\wh Y_\beta)_{k'}{}^{j'}=
N\, [\pi_{\bf N}(U)]_i{}^{j'}  [\pi_{\bf N}(U^{-1})]_{k'}{}^l.
\ee
If we expand  $U=\exp (u^{\gamma} X_\gamma )$ to first order in $u^\gamma$, we get from \eqref{repX} that 
\be \label{Fund1}
(\wh Y_\gamma)_i{}^j \delta_k^l - \delta_i^j (\wh Y_\gamma)_k{}^l 
= \frac{2i}{N^2} \sum_{\alpha, \beta } \wh{C}_{\gamma }{}^{\alpha \beta } (\wh Y_\alpha)_i{}^l(\wh Y_\beta)_k{}^j.
\ee 
Next consider the matrix elements of $\wh{{J}}$, which in two different bases are
\begin{equation}
    X_{\alpha} =\frac{N}{2i}\sum_{i,j} E\indices{^j_i}(\wh{Y}_{\alpha})_{j}{}^{i},\qquad  E\indices{^j_i}=\frac{2i}{N^2}\sum_{\alpha \in I_N} X^\alpha (\wh{Y}_{\alpha})_{i}{}^{j},
\end{equation}
where $(\wh{Y}_{\alpha})_{i}{}^{j}$ denotes the $ij$\textsuperscript{th} component of the matrix $\wh{Y}_\alpha$ (see Appendix \ref{sec:fuzH} for details). The proof of equivalence between \eqref{eqn:XaXb} and \eqref{eqn:EijEkl} and  goes as follows. We first assume \eqref{eqn:XaXb} and prove \eqref{eqn:EijEkl} as follows
\begin{equation}
    \begin{aligned}
     [E\indices{^j_i},E\indices{^l_k}]&= \left(\frac{2i}{N^2}\right)^2 \sum_{\alpha,\beta\in I_N}[X^\alpha,X^\beta](\wh{Y}_{\alpha})_i{}^j(\wh{Y}_{\beta})\indices{_k^l}
     \\
     &= \left(\frac{2i}{N^2}\right)^2 \sum_{\alpha,\beta \in I_N}\wh{C}\indices{^\alpha^\beta_\gamma} X^\gamma(\wh{Y}_{\alpha})_i{}^j(\wh{Y}_{\beta})\indices{_k^l}
     \\
     &= \frac{2i}{N^2} \sum_{\gamma \in I_N} X^\gamma (\delta_i{}^l(\wh{Y}_\gamma)_k{}^j-\delta^j{}_k(\wh{Y}_\gamma)_i{}^l)
     \\
     &=\delta_i{}^l E\indices{^j_k}-\delta_k{}^j E\indices{^l_i},
    \end{aligned}
\end{equation}
which is the $\mfk{su}(N)$ defining relation \eqref{eqn:EijEkl} and we have use the identity \eqref{Fund1}. Using this, one can show the following
\be 
[\pi_{\bf N}(E\indices{^j_i})]_a{}^b
= \delta_a^j\delta_i^b,
\ee
then the commutator of any two matrices $A:= A_j{}^i E_i{}^j$ and $B  := B_j{}^i E_i{}^j$ is given by
\begin{equation}
    \begin{aligned}
        [A,B]&= A_j{}^i B_l{}^k [E\indices{^j_i},E\indices{^l_k}]
        = A_j{}^i B_l{}^k (\delta_i{}^l E\indices{^j_k}-\delta_k{}^j E\indices{^l_i})
        \\
        &= (A B)_j{}^k E\indices{^j_k} - (BA)_l{}^i E\indices{^l_i}
        \\
        &= [A,B]_j{}^i E\indices{^j_i}.
    \end{aligned}
\end{equation}
Conversely, we can derive \eqref{eqn:XaXb} from \eqref{eqn:EijEkl} easily as follows
\begin{equation}
    \begin{aligned}
     [X_\alpha,X_\beta]&=\frac{N^2}{(2i)^2}
[E\indices{^j_i},E\indices{^l_k}] (\wh{Y}_{\alpha}) _{j}{}^{i}(\wh{Y}_{\beta})_{l}{}^{k}\cr
&= \frac{N^2}{(2i)^2}(\delta_i{}^l E\indices{^j_k}-\delta_k{}^j E\indices{^l_i}) (\wh{Y}_{\alpha})_{j}{}^{i}(\wh{Y}_{\beta})_{l}{}^{k} \cr
&= \frac{N^2}{(2i)^2}
 E\indices{^j_i}  [\wh{Y}_{\alpha}, \wh{Y}_{\beta}]_{j}{}^{i} \cr
&=\frac{N}{2i}\wh{C}_{\alpha\beta}{}^\gamma E\indices{^j_i} (\wh{Y}_\gamma)_j{}^i 
\\
&= \wh{C}_{\alpha\beta}{}^\gamma X_\gamma.
    \end{aligned}
\end{equation}
This completes the proof of equivalence of \eqref{eqn:XaXb} and \eqref{eqn:EijEkl}. The equivalence of \eqref{AJa} and \eqref{eq:Jij and Jki commutaror} then follows.

\subsection{$\mfk{su}(N,N)$}\label{sec:su(N,N) idesntities}
We can provide a similar identity for $\mfk{su}(N,N)$. The Lie algebra generators in the vector representations are 
\begin{equation}
    \wh{Y}_{\bullet\alpha}=\mathbbm{1}_{2}\otimes\wh{Y}_\alpha, \qquad \wh{Y}_{a\alpha}=\rho_a\otimes\wh{Y}_\alpha.
\end{equation}
We see that 
\be 
 \frac12 (\mathbbm{1}_{2}\otimes \mathbbm{1}_{2}) - 2 (\rho_a \otimes \rho^a)= \mathbbm{1}_{\mathrm{Mat(2)}}.
\ee
Therefore, using \eqref{eqn:fuzHcompleteness} which states that $\sum_{\alpha\in I_N} 
    (\wh{Y}_\alpha \otimes \wh Y^{\alpha})
    = N \mathbbm{1}_{\mathrm{Mat}(\mbf{N})}
    $ we have the identity decomposition
\begin{equation}\label{eq:Id2}
    \sum_{\alpha \in I_N}\left(\frac{1}{2}\wh{Y}_{\bullet\alpha}\otimes\wh{Y}^{\bullet\alpha}-2\sum_{a=0,1,2}\wh{Y}_{a\alpha}\otimes\wh{Y}^{a\alpha}\right)=N \mathbbm{1}_{\mathrm{Mat}(2N)}.
\end{equation}
From this, we can show that
\begin{equation*}
    \begin{aligned}
       (\wh Y_{\id \gamma})_{\msf n}{}^{\msf q} \delta_{\msf p}{}^{\msf m} - \delta_{\msf n}{}^{\msf q} (\wh Y_{\id\gamma})_{\msf p}{}^{\msf n} 
        &= \frac{2i}{N^2} \sum_{\alpha,\beta \in I_N} \wh{C}_{ \gamma }{}^{\alpha\beta } \Big(\tfrac12 (\wh Y_{\id \alpha})_{\msf n}{}^{\msf m}(\wh (Y_{\id \beta})_{\msf n}{}^{\msf m}- 2\sum_{a,b}\eta^{ab} (\wh Y_{a \alpha})_{\msf n}{}^{\msf m}(\wh Y_{ b \beta})_{\msf p}{}^{\msf q} \Big),
        \\
        (\wh Y_{c \gamma})_{\msf n}{}^{\msf q} \delta_{\msf p}{}^{\msf m} - \delta_{\msf n}{}^{\msf q} (\wh Y_{c \gamma})_{\msf p}{}^{\msf n} 
        &=\frac{2i}{N^2} \sum_{\alpha,\beta \in I_N} \wh{C}_{ \gamma }{}^{\alpha\beta } \left((\wh Y_{c \alpha})_{\msf n}{}^{\msf m}(\wh (Y_{\id \beta})_{\msf n}{}^{\msf m}
        - (\wh Y_{\id \alpha})_{\msf n}{}^{\msf m}((\wh{Y}_{c \beta})_{\msf n}{}^{\msf m} \right)
        \\
        &- \frac{2i}{N^4} \sum_{\alpha,\beta \in I_N} \wh{C}_{ \gamma }{}^{\alpha\beta } \sum_{a,b} \varepsilon_c{}^{ab} (\wh Y_{a \alpha})_{\msf n}{}^{\msf m}(\wh Y_{ b \beta})_{\msf p}{}^{\msf q},
    \end{aligned}
\end{equation*}
where $\msf{m}=(A,i)$ with $A=1,2$ and $i=1,\ldots,N$. 
These identities are exactly what is needed to establish that 
\begin{equation}
    [E\indices{^{\msf m}_{\msf n}}, E\indices{^{\msf p}_{\msf q}}] =
\delta^{\msf p}_{\msf n}E\indices{^{\msf m}_{\msf q}}
-\delta^{\msf m}_{\msf q}E\indices{^{\msf p}_{\msf n}},
\end{equation}
where we have defined  the $\su(N,N)$ generator
\begin{equation}
    E\indices{^{\msf m}_{\msf n}} := 
\frac{i}{N}\left(\frac{1}{N} 
\algX^\alpha\big(\wh{Y}_{\id\alpha}\big)\indices{_{\msf{n}}^{\msf m}}
-2 \algZ^{a\alpha}\big(\wh{Y}_{a\alpha}\big)\indices{_{\msf n}^{\msf m}}
\right).
\end{equation}
The equality \eqref{eq:Id2} can also be written in terms of the $\mfk{su}(N,N)$ element $E$ as 
\be (\pi_{\bf 2N}\otimes \mathbbm{1}_{\mbf{2N}}) E=  \mathbbm{1}_{\mathrm{Mat}(\mbf{2N})}.
\ee
where $\pi_{\bf 2N}$ is the representation \eqref{eqn:pi2NZ}.

\section{Identities for Casimir computations}
\label{app:casimirs}

In this appendix, we collect a number of computations relevant for the 
discussion of Casimirs for the continuum and deformed algebras 
described in section \ref{sec:casimirs}.  

\paragraph{Proof of (\ref{eqn:phij}).}  This equation can be derived 
using the mode decomposition (\ref{eqn:jmodes}) of $\defj(\sigma)$.  
This gives
\begin{eqaligned}
[\phi, \jfxn(\sigma)]_{{\mfk{g}}} 
&=
[\phi, \jfxn_\alpha]_{{\mfk{g}}} \, {Y}^\alpha(\sigma) 
\\
&=
\int_S d\sigma' \{\phi, Y_\alpha\}(\sigma')\; \jfxn(\sigma') Y^\alpha(\sigma) 
\\
&=
-\int_S d\sigma' \{\phi,\jfxn\}(\sigma')\; Y_\alpha(\sigma') Y^\alpha(\sigma)
\\
&= 
-\int_S d\sigma' \{\phi,\jfxn\}(\sigma')\; \delta(\sigma-\sigma') 
\\
&=
-\{\phi,\jfxn\}(\sigma).
\end{eqaligned}

\paragraph{Proof of (\ref{eqn:Xaj}). }  Here we use the mode decomposition
(\ref{eqn:defjmodes}) of $\defj$:
\begin{eqaligned}
[\algX_\alpha, \defj]_{\wh{\mfk{g}}}
&=
[\algX_\alpha, \algX_\beta]_{\wh{\mfk{g}}} \wh{Y}^\beta
\\
&=
\wh{C}\indices{_\alpha_\beta^\gamma}\algX_\gamma \wh{Y}^\beta
\\
&=
-\frac{N}{2i}[\wh{Y}_\alpha,\wh{Y}^\gamma]\algX_\gamma
\\
&=
-\frac{N}{2i}[\wh{Y}_\alpha, \defj].
\end{eqaligned}

\paragraph{Proof of (\ref{eqn:dhatd}).}
We begin with the expression for $d_{\alpha_1\ldots \alpha_n}$, 
which, employing (\ref{eqn:ECclassical}) and (\ref{eqn:shmetric}), 
is given by
\begin{eqaligned}
d_{\alpha_1\ldots\alpha_n} &= \int_S\nu_0 \,
Y_{\alpha_1}\ldots Y_{\alpha_n}
\\
&= E\indices{_{\alpha_1}_{\alpha_2}^{\beta_1}}
E\indices{_{\beta_1}_{\alpha_3}^{\beta_2}} \ldots
E\indices{_{\beta_{n-3}}_{\alpha_{n-1}}^{\beta_{n-2}}}
\int_S\nu_0Y_{\beta_{n-2}} Y_{\alpha_n}
\\
&=E\indices{_{\alpha_1}_{\alpha_2}^{\beta_1}} \ldots
E\indices{_{\beta_{n-3}}_{\alpha_{n-1}}_{\alpha_n}}.
\end{eqaligned}
Note that for $n=2$, the expression instead reads $d_{\alpha\beta} = 
\int_S\nu_0\, Y_\alpha Y_\beta = \delta_{\alpha\beta}$.
On the other hand, employing (\ref{TrM}),
we have
\begin{eqaligned}
\wh{d}_{\alpha_1\ldots \alpha_n} &=
    \frac{1}{N}\tr_{\mbf{N}}(\wh{Y}_{\alpha_1}\ldots \wh{Y}_{\alpha_n})
    \\
    &=\wh{M}\indices{_{\alpha_1\alpha_2}^{\beta_1}}
    \wh{M}\indices{_{\beta_1}_{\alpha_3}^{\beta_2}} \ldots \wh{M}\indices{_{\beta_{n-3}\alpha_{n-1}}^{\beta_{n-2}}}\frac1N
    \tr_{\mbf{N}}(\wh{Y}_{\beta_{n-2}} \wh{Y}_{\alpha_{n}})
    \\
    &=
    \wh{M}\indices{_{\alpha_1\alpha_2}^{\beta_1}} \ldots
    \wh{M}\indices{_{\beta_{n-2}}_{\alpha_{n-1}}_{\alpha_n}}.
\end{eqaligned}
Then using that $\wh{M}\indices{_\alpha_\beta^\gamma} = 
\wh{E}\indices{_\alpha_\beta^\gamma}
+\frac{i}{N}\wh{C}\indices{_\alpha_\beta^\gamma}$, we find that 
\beq
\wh{d}_{\alpha_1\ldots \alpha_n} = 
    \wh{E}\indices{_{\alpha_1\alpha_2}^{\beta_1}} \ldots
    \wh{E}\indices{_{\beta_{n-2}}_{\alpha_{n-1}}_{\alpha_n}}
    +\op(N^{-1}).
\eeq
Since $\wh{E}\indices{_\alpha_\beta^\gamma} = E\indices{_\alpha_\beta^\gamma}
+\op(N^{-2})$ (see sections \ref{sec:fuzH} and \ref{sec:mpexpansion}), this shows that 
\beq
\wh{d}_{\alpha_1\ldots\alpha_n} = d_{\alpha_1\ldots \alpha_n}+\op(N^{-1}).
\eeq

\paragraph{Proof of (\ref{eqn:CkA}).} 
The Casimir matching considered in section \ref{sec:casmatch} requires us
to determine the relation between the $J_\alpha$ Hamiltonians defined by
equation (\ref{eqn:Jalpha}), and the gravitational charges 
constructed in \cite{DonnellyFreidelMoosavianSperanza202012}.  The latter
were defined in terms of the vector fields $\xi^A$ on $S$, while 
the former are written in terms of the stream functions $Y_\alpha$.  
According to the conventions of section \ref{sec:PJ}, these are 
related by $\xi_{Y_\alpha}^A = \epsilon^{BA}\partial_B Y_\alpha = \frac{1}{4\pi}
\nu_0^{BA}\partial_B Y_\alpha$.  The Hamiltonians in
\cite{DonnellyFreidelMoosavianSperanza202012} were written in terms of a 
1-form density $\wt P_A$ which has a geometrical interpretation in
spacetime as a component of a connection on the normal bundle of 
$S$.  The relation for the charges in terms of this is then given by
\begin{eqaligned}
J_\alpha &= \frac{1}{16\pi G}\int_S \wt{P}_A\, \xi^A_{Y_\alpha}
\\
&=\frac{A}{16\pi G}\int_S\nu_0 P_A\, \xi^A_{Y_\alpha}
\\
&=\frac{A}{16\pi G}\int_S \nu_0 \frac1{4\pi}\nu_0^{AB}P_A \partial_B Y_\alpha
\\
&=\frac{A}{16\pi G}\int_S \frac{1}{4\pi} dY_\alpha \wedge P 
\\
&=
\frac{-A}{16\pi G} \int_S Y_\alpha \frac{dP}{4\pi}
\\
&=
\frac{A}{16\pi G}\int_S\nu_0Y_\alpha \left(\frac{-AW}{4\pi}\right),
\label{eqn:JaAW}
\end{eqaligned}
where, following the conventions of 
\cite{DonnellyFreidelMoosavianSperanza202012}, we have used that $\widetilde{P}_A$ is related to $P_A$
via the physical volume form $\nu = A\nu_0$, so that 
$\widetilde{P}_A = \nu P_A = A\nu_0 P_A$, and 
the last equality uses that $dP$ is related to the outer curvature 
scalar $W$  according to $dP = W \nu = AW\nu_0$. 
As a curvature scalar, $W$ has dimensions $[\text{length}]^{-2}$,
and $A$ has dimensions $[\text{length}]^2$ in $4$ spacetime 
dimensions, so the function $AW$ is dimensionless. 
The final integral in (\ref{eqn:JaAW}) is therefore dimensionless,
and hence $J_\alpha$ has dimensions of angular momentum, as 
expected since the $\diff$ charges are generalizations 
of angular momentum.
Comparing to equation (\ref{eqn:Jalpha}), we see that the function $J(\sigma)$
is related to the geometrical data according to 
\beq
J(\sigma) = \frac{A}{16\pi G} \left(\frac{-A W(\sigma)}{4\pi}\right).
\eeq
Plugging this relation into the expression (\ref{eqn:Cn}) for the gravitational
Casimirs then immediately reproduces equation (\ref{eqn:CkA}).

\paragraph{Proof of (\ref{eqn:ad*pha}).}
Using the definition of the coadjoint action and the pairing 
(\ref{eqn:crrpairing}) for $\crr$, we compute
\begin{eqaligned}
\langle \ad^*_{(\phi,\alpha)}(f,a), (\psi,\beta)\rangle
&=
-\langle (f,a),(\{\phi,\psi\}, \{\phi,\beta\} - \{\psi,\alpha\})\rangle
\\
&=
-\int_S\nu_0\Big(f\{\phi,\psi\} + a(\{\phi,\beta\}+\{\alpha,\psi\}) \Big)
\\
&=
\int_S \nu_0\Big( \{\phi,f\}\psi +\{\phi,a\}\beta + \{\alpha,a\}\psi\Big)
\\
&=
\langle (\{\phi,f\} + \{\alpha,a\}, \{\phi,a\}), (\psi,\beta)\rangle,
\end{eqaligned}
which determines the coadjoint action to be (\ref{eqn:ad*pha}).

\paragraph{Proof of (\ref{eqn:phialphajfxn}) and (\ref{eqn:phialphanfxn}).}
These relations are derived as follows:
\begin{eqaligned}
\langle (f,a),[(\phi,\alpha),\jfxn]_{\mfk{g}}\rangle
&=
-\langle \ad^*_{(\phi,\alpha)}(f,a),\jfxn\rangle
\\
&=
-\langle(\{\phi,f\}+\{\alpha,a\},\{\phi,a\}),\jfxn\rangle
\\
&=
-\{\phi,f\} - \{\alpha,a\}
\\
&=\langle(f,a),-\{\phi,\jfxn\}-\{\alpha,\nfxn\}\rangle,
\end{eqaligned}
verifying (\ref{eqn:phialphajfxn}).  Similarly,
\begin{eqaligned}
\langle (f,a),[(\phi,\alpha),\nfxn]_{\mfk{g}}\rangle
&=
-\langle \ad^*_{(\phi,\alpha)}(f,a),\nfxn\rangle
\\
&=
-\langle(\{\phi,f\}+\{\alpha,a\},\{\phi,a\}),\nfxn\rangle
\\
&=
-\{\phi,a\}
\\
&=\langle(f,a),-\{\phi,\nfxn\}\rangle,
\end{eqaligned}
verifying (\ref{eqn:phialphanfxn}).

\paragraph{Identities for deriving (\ref{eqn:alphacmn})}
The Lie-algebra-valued functions $\jfxn$ and $\nfxn$ can be shown
to satisfy
\begin{eqaligned}
\jfxn\nfxn &=\jfxn_\alpha\nfxn_\beta Y^\alpha Y^\beta 
\\
&=
\nfxn_\beta \jfxn_\alpha Y^\alpha Y^\beta
+[\jfxn_\alpha,\nfxn_\beta]_{\mfk{g}} E\indices{^\alpha^\beta^\gamma} Y_\gamma
\\
&=
\nfxn \jfxn
+C\indices{_\alpha_\beta^\gamma}\nfxn_\gamma
E\indices{^\alpha^\beta^\gamma}Y_\gamma
\\
&=
\nfxn \jfxn.
\end{eqaligned}

\begin{eqaligned}
\{\nfxn,\nfxn\} 
&=
\nfxn_\alpha \nfxn_\beta\{Y^\alpha,Y^\beta\}
\\
&=
\nfxn_\alpha \nfxn_\beta C\indices{^\alpha^\beta^\gamma} Y_\gamma
\\
&=
\frac12[\nfxn_\alpha,\nfxn_\beta]_{\mfk{g}} C\indices{^\alpha^\beta^\gamma}
Y_\gamma = 0.
\end{eqaligned}

\begin{eqaligned}
\{\jfxn,\nfxn\} &=
\jfxn_\alpha \nfxn_\beta C\indices{^\alpha^\beta^\gamma}Y_\gamma
\\
&=
\frac12[\jfxn_\alpha,\nfxn_\beta]_{\mfk{g}} C\indices{^\alpha^\beta^\gamma}
 Y_\gamma
\\
&=
\frac12 C\indices{_\alpha_\beta^\mu} C\indices{^\alpha^\beta^\gamma}
\nfxn_\mu Y_\gamma.
\end{eqaligned}
Although the sum over $\alpha$ and $\beta$ in this final expression is 
divergent, we can take $C\indices{_\alpha_\beta^\mu} C\indices{^\alpha^\beta^\gamma}$ to be proportional to $\delta^{\mu\gamma}$
times a divergent coefficient.  This then demonstrates that 
$\{\jfxn,\nfxn\}\propto \nfxn$.

\paragraph{Identity satisfied by $\nfxn_a$.}  The quantities 
$\nfxn_a(\sigma)$ and $\nfxn_b(\sigma')$ do not commute due to 
being valued in a Lie algebra, but instead have a $\delta-$function
contribution coming from coincident points.  This can be derived by
\begin{eqaligned}
\nfxn_a(\sigma)\nfxn_b(\sigma')
&= 
\nfxn_{a\alpha} \nfxn_{b \beta} Y^{\alpha}(\sigma) Y^\beta(\sigma')
\\
&=
\nfxn_{b \beta}\nfxn_{a\alpha}Y^\alpha(\sigma)Y^\beta(\sigma')
+\varepsilon\indices{_a_b^c}E\indices{_\alpha_\beta^\gamma}
\nfxn_{c\gamma}Y^\alpha(\sigma) Y^\beta(\sigma')
\\
&=
\nfxn_b(\sigma')\nfxn_a(\sigma) +
\varepsilon\indices{_a_b^c}\nfxn_{c\gamma}\int_Sd\sigma''
Y_\alpha(\sigma'')Y_\beta(\sigma'')Y^\gamma(\sigma'')
Y^\alpha(\sigma)Y^\beta(\sigma')
\\
&=
\nfxn_b(\sigma')\nfxn_a(\sigma) +
\varepsilon\indices{_a_b^c}\nfxn_{c\gamma}
\int d\sigma'' \delta(\sigma'' - \sigma)\delta(\sigma'' - \sigma')
Y^\gamma(\sigma'')
\\
&=
\nfxn_b(\sigma')\nfxn_a(\sigma) +
\varepsilon\indices{_a_b^c}\delta(\sigma-\sigma')\nfxn_c(\sigma).
\end{eqaligned}

\paragraph{Proof of (\ref{eqn:jna}) and (\ref{eqn:nanb}).}
These relations are once again derived using the mode decompositions 
of $\defj$ and $\defn_a$:
\begin{eqaligned}
\defj\defn_a 
&= 
\algX_\alpha \algZ_{a\beta}\wh{Y}^\alpha \wh{Y}^\beta
\\
&=
\algZ_{a\beta}\algX_\alpha \wh{Y}^\beta\wh{Y}^\alpha
+\algZ_{a\beta}\algX_\alpha [\wh{Y}^\alpha,\wh{Y}^\beta]
+[\algX_\alpha,\algZ_{a\beta}]_{\wh{\mfk{g}}} \wh{Y}^\alpha \wh{Y}^\beta
\\
&=
\defn_a \jfxn + 
\algZ_{a\beta}\algX_\alpha [\wh{Y}^\alpha,\wh{Y}^\beta]
+\wh{C}\indices{_\alpha_\beta^\mu}\algZ_{\mu a}\wh{Y}^\alpha\wh{Y}^\beta
\\
&=
\defn_a \jfxn +
\algZ_{a\beta}\algX_\alpha [\wh{Y}^\alpha,\wh{Y}^\beta]
+\frac12[\algX_\alpha,\algZ_{a\beta}]_{\wh{\mfk{g}}}[\wh{Y}^\alpha,\wh{Y}^\beta]
\\
&=
\defn_a \jfxn +
\frac{i}{N}(\algZ_{a\beta}\algX_\alpha + \algX_\alpha \algZ_{a\beta})
\wh{C}\indices{^\alpha ^\beta ^\mu}\wh{Y}_\mu
\\
&=
\defn_a \jfxn +\op(N^{-1}).
\end{eqaligned}

\begin{eqaligned}
\defn_a \defn_b
&=
\algZ_{a\alpha}\algZ_{b \beta}\wh{Y}^\alpha \wh{Y}^\beta
\\
&=
\algZ_{b \beta}\algZ_{a\alpha}\wh{Y}^\beta \wh{Y}^\alpha+
[\algZ_{a\alpha},\algZ_{b \beta}]_{\wh{\mfk{g}}} \wh{Y}^\alpha \wh{Y}^\beta
+\algZ_{a\alpha}\algZ_{b \beta}[\wh{Y}^\alpha,\wh{Y}^\beta]
\\
&=
\defn_b\defn_a + [\defn_a,\defn_b]_{\wh{\mfk{g}}}
+ \frac{2i}{N}\algZ_{a\alpha}\algZ_{b \beta}\wh{\{Y_\alpha,Y_\beta\}}
+\op(N^{-3})
\\
&=
\defn_b\defn_a + [\defn_a,\defn_b]_{\wh{\mfk{g}}}
+ \frac{2i}{N}\wh{\{\nfxn_a,\nfxn_b\}}
+\op(N^{-3}).
\end{eqaligned}

\paragraph{Proof of (\ref{casint}).}
We now put a derivation of \eqref{casint} which gives an alternate derivation that \eqref{eqn:c2n+1} is a Casimir.
\begin{equation}
    \begin{aligned}
        \tenofo{ad}^*_{(\phi, \alpha)}(fa^{2n})&=(\tenofo{ad}^*_{(\phi, \alpha)}f)a^{2n}+f\tenofo{ad}^*_{(\phi, \alpha)}a^{2n}
        \\
        &=\{\phi,f\}_{\nu_0}a^{2n}+\{\alpha_a,a^a\}_{\nu_0}a^{2n}+f\{\phi,a^{2n}\}_{\nu_0}
        \\
        &=\{\phi, fa^{2n}\}_{\nu_0} + \{\alpha^a, a_a\}_{\nu_0} a^{2n},
    \end{aligned}
\end{equation}
and also
\begin{equation} \label{eqn:aaacoadfoint}
    \begin{aligned}
        &\tenofo{ad}^*_{(\phi, \alpha)} (\varepsilon_{abc} \{a^a, a^b\} a^c\, a^{2(n-1)})=\{\phi, \varepsilon_{abc} \{a^a, a^b\} a^c a^{2(n-1)} \}
        \\
        +&
        \varepsilon_{abc} (\{[\alpha,a]^a, a^b\} a^c 
        + \{a^a, [\alpha,a]^b\} a^c + \{a^a, a^b\} [\alpha,a]^c) a^{2(n-1)}
        \\
        =&\{\phi, \varepsilon_{abc} \{a^a, a^b\} a^c \, a^{2(k-1)}\} + 2 \varepsilon_{abc} \varepsilon^{ade}(\{\alpha_d , a^b\} a^c a_e) a^{2(n-1)}
        \\
        =&\{\phi, \varepsilon_{abc} \{a^a, a^b\} a^c\,a^{2(n-1)}\} - 
        2 (\{\alpha_b , a^b\} a^c a_c - \{\alpha_c , a^b\} a^c a_b) a^{2(n-1)}
        \\
        =&\{\phi, \varepsilon_{abc} \{a^a, a^b\} a^c\,a^{2(n-1)}\} -
        2 \{\alpha_b , a^b\} a^{2n} + \{\alpha_c , a^2\} a^c a^{2(n-1)}
        \\
        =&\{\phi, \varepsilon_{abc} \{a^a, a^b\} a^c\, a^{2(n-1)}\} -
        \frac{2n+1}{n} \{\alpha_b , a^b  \}a^{2n}   + \frac{1}{n} \{\alpha_b , a^{2n} a^b\}.
    \end{aligned}
\end{equation}
We can thus define
\begin{equation}
    w_n := \left[ (2n+1) fa^{2} + n\varepsilon_{abc} \{a^a, a^b\} a^c\right] a^{2(n-1)},
\end{equation}
which transforms as 
\begin{equation}\label{eq:coadjoint action on vorticity}
    \tenofo{ad}^*_{(\phi, \alpha)}w_n = \{\phi, w_n\}+ \{\alpha_a,a^{2n} a^a\}.
\end{equation}
The proof that \eqref{eqn:c2n+1} is a Casimir then follows from the fact that the integral of \eqref{eq:coadjoint action on vorticity} on $S$ vanishes, i.e.\ we have
\begin{equation}
    \bigintsss_S\{f,g\}_\nu=0, \qquad \forall f,g\in C(S).
\end{equation}
This can be proven as follows. Let $(M,\omega)$ be an $2n$-dimensional symplectic manifold with symplectic form $\omega$. The volume form on $M$ is given by $\omega^n$. Then, for any two functions $f,g\in C(M)$, we have
\begin{equation}
    \begin{aligned}
    \{f,g\}\omega^n&=\omega(X_f,X_g)\omega^n
    \\
    &=X_g(f)\omega^n=\mcal{L}_{X_g}f\omega^n
    \\
    &=\mcal{L}_{X_g}(f\omega^n)=(\rd\iota_{X_g}+\iota_{X_g}\rd)f\omega^n
    \\
    &=\rd\iota_{X_g}(f\omega^n)
    \\
    &=\rd(f\iota_{X_g}\omega^n),
    \end{aligned}
\end{equation}
where $X_g$ is the Hamiltonian vector field associated with $g$ defined by $\iota_{X_g}\omega:=\rd g$, and in the fourth equality we used the fact that the Lie derivative of the symplectic form along a Hamiltonian vector field vanishes; this can be seen as follows
\begin{equation*}
    \mcal{L}_{X_g}\omega=(\rd\iota_{X_g}+\iota_{X_g}\rd)\omega=\rd(\rd g)+0=0,
\end{equation*}
where we have used the closeness of the symplectic form $\rd\omega=0$. We thus have
\begin{equation}\label{eq:the integral of Poisson bracket on a symplectic manifold}
    \bigintsss_M\omega^n\,\{f,g\}=\bigintsss_M \rd(f\iota_{X_g}\omega^n)=\bigintsss_{\partial M}f\iota_{X_g}\omega^n.
\end{equation}
In the case of sphere $M=S$, $\omega=\sqrt{q}\epsilon_{AB}\rd\sigma^A\wedge \rd\sigma^B$, $\pa M=\emptyset$, and $\omega^n\{f,g\}\to\{f,g\}_\nu$. Therefore, \eqref{eq:the integral of Poisson bracket on a symplectic manifold} vanishes and we thus end up with the desired result.

\bibliography{reps}
\bibliographystyle{JHEP}

\end{document}